\documentclass[11pt,a4paper]{article}
\usepackage[T1]{fontenc}
\usepackage{amsmath,amssymb,amsthm}
\usepackage{mathtools}
\usepackage[a4paper,margin=1.03in,includefoot]{geometry}
\usepackage[colorlinks=true,linkcolor=blue,citecolor=blue,urlcolor=blue]{hyperref}
\usepackage{graphicx}
\usepackage{accsupp}
\providecommand{\BeginAccSupp}[1]{}
\providecommand{\EndAccSupp}[1]{}
\makeatletter
\@ifundefined{KV@Gin@alt}{\define@key{Gin}{alt}{}}{}
\makeatother
\hypersetup{pdftitle={Optimal GHZ extraction from MABK violations},
 pdfauthor={Avraham Kreindel and Aryeh Lev Zabokritskiy (Yohananov)},pdflang={en-US}}
\newtheorem{theorem}{Theorem}
\newtheorem{lemma}[theorem]{Lemma}
\newtheorem{proposition}[theorem]{Proposition}
\theoremstyle{remark}
\newtheorem{remark}[theorem]{Remark}
\newcommand{\ket}[1]{\lvert #1\rangle}
\newcommand{\bra}[1]{\langle #1\rvert}
\newcommand{\proj}[1]{\ket{#1}\!\bra{#1}}
\newcommand{\braket}[2]{\langle #1\vert #2\rangle}
\newcommand{\tr}{\operatorname{tr}}
\newcommand{\GHZ}{\mathrm{GHZ}}
\title{Optimal GHZ extraction from MABK violations}
\author{Avraham Kreindel\\
{\small Department of Computer Science}\\
{\small Reichman University, Herzliya, Israel}\and
Aryeh Lev Zabokritskiy (Yohananov)\\
{\small Department of Computer Science}\\
{\small Tel-Hai University of Kiryat Shmona in the Galilee}\\
{\small MIGAL---Galilee Research Institute}\\
{\small ORCID: \href{https://orcid.org/0000-0003-3151-6192}{0000-0003-3151-6192}}}
\date{}
\begin{document}
{\renewcommand{\and}{\\[0.75em]}\maketitle}

\begin{abstract}
For \(n\geq3\) parties, we prove that the least Greenberger--Horne--Zeilinger (GHZ)
extractability compatible with a Mermin--Ardehali--Belinskii--Klyshko (MABK)
score between the biseparable bound and the quantum maximum equals \((1+t)/2\),
where \(t\) rises linearly with the score from \(t=0\) to \(t=1\) across this
range. The minimum is over all states
and binary measurements in arbitrary local dimensions, and the extractability of a
state is the supremum over local channels of the overlap with the GHZ state.
With a qubit at every party, the minimum is strictly larger for \(0<t<1\); we
determine it exactly.
We construct a mixture on one qutrit and \(n-1\) qubits that attains \((1+t)/2\)
with measurements independent of \(t\). We prove that, for \(0<t<1\), its product
of local ranks (ranks of its single-party reduced states), \(3\cdot2^{n-1}\),
is the smallest among attaining states and, for \(n\geq4\), it is the only
attaining state with this product, up to local isometries and party permutations.
On one qutrit and \(n-1\) qubits with arbitrary binary measurements, for
\(n\geq4\) and \(t\) bounded away from zero, we prove that an extractability
within a small \(\varepsilon\) of \((1+t)/2\) forces trace-norm distance
\(O(\sqrt{\varepsilon})\) from the mixture at the same \(t\), up to local
unitaries, uniformly in \(n\) and with optimal exponent.
Since the Bell statistics do not determine the extractability, we also bound
this distance for \(n\geq3\) and \(t>0\) by the MABK score and the expectations
of five calibrated binary measurements on the qutrit, assuming known local
dimensions and a known two-dimensional qutrit subspace invariant under both Bell
observables of that party; the bound vanishes when these expectations equal
their values in the mixture.
\end{abstract}

\section{Introduction}

A violation of a multipartite Bell inequality can certify a shared
entangled state device-independently.
At maximal violation the certification can be exact; below it, the
question is how much of the target state is guaranteed by a given
score, that is, by a given value of the Bell expression.
We study Greenberger--Horne--Zeilinger (GHZ) states and the
Mermin--Ardehali--Belinskii--Klyshko (MABK) inequalities
\cite{Mermin1990,Ardehali1992,BelinskiiKlyshko1993}, with two binary
measurement choices at each party.
The figure of merit is the extractability~\cite{Kaniewski2016}, a
measure first proposed by Bardyn et~al.~\cite{Bardyn2009}:
the largest overlap with the GHZ state obtainable from the shared state
by separate local quantum channels.
Its least value among states and measurements with a prescribed
score is the strongest guarantee that the score alone provides.
Brito, Amaral and Chaves quantified a different object, the nonlocality
of the observed correlations, by the trace distance of these
correlations from the local set, including for multipartite Mermin
correlations~\cite{BritoAmaralChaves2018}.

At maximal violation, the MABK inequality self-tests the GHZ state for
every number \(n\geq3\) of parties
\cite{Kaniewski2017,PanwarPandyaWiesniak2023}.
Bao et~al.\ give a four-party self-test of the GHZ state and local
Pauli observables at maximal MABK violation \cite[Theorem~3]{Bao2025}.
For nonmaximal violations, Kaniewski determined the tight affine
bound for the three-party Mermin inequality in finite dimension
\cite{Kaniewski2016}.
Jha, Singh and Pan state the corresponding tight bounds for four
and five parties \cite[Theorem~2]{Jha2026} and indicate that the proof is similar to that of
their Theorem~1, which concerns a different Bell operator.
Cao, Zhang, Shi and Zhao obtained, for every number of parties, a
lower bound linear in the relative deficiency of the score from its
quantum maximum \cite{Cao2026}.
The coefficient of that deficiency is twice the optimal one, and
their linear estimate vanishes at the biseparable bound, the largest
MABK score of biseparable states
\cite{WernerWolf2000,NagataKoashiImoto2002}, which is \(1/\sqrt2\) times
the quantum maximum.
They conjectured the optimal affine bound for \(n\geq6\) and
identified two remaining scalar inequalities, with numerical evidence
up to \(n=100\). The two inequalities were subsequently proved in Lean
\cite{ElonSG2026}.

We prove that, for every \(n\geq3\),
the least extractability is \((1+t)/2\), where \(t\) rises linearly with
the score from \(t=0\) at the biseparable bound to \(t=1\) at the quantum
maximum, and that every point of this line is attained.
The result holds for arbitrary binary positive operator-valued
measures (POVMs) and density operators on tensor products of local
spaces of any finite or infinite dimension.
It extends the known three-party bound from finite to arbitrary local
dimension. For four or more parties, our uniform proof yields spectral
estimates that we use to analyze the states attaining the bound and their
stability.
Thus every score above the biseparable bound guarantees GHZ
extractability strictly above one half, the overlap obtained for every
state when each party outputs \(\ket0\).

For every \(n\geq3\), when each party holds a qubit and may use
arbitrary binary POVMs, we determine the exact minimum to be
\((1+\sqrt{2b^2-1})/2\), where \(b\) is the score divided by its quantum
maximum.
Here \(1/\sqrt2\leq b\leq1\) and \(t=(\sqrt2\,b-1)/(\sqrt2-1)\).
For every number of qubits with projective measurements, Scarani and
Gisin showed that a score above the biseparable bound requires overlap
larger than one half with a locally rotated GHZ
state~\cite{ScaraniGisin2001}.
For two qubits and the Clauser--Horne--Shimony--Holt (CHSH) inequality,
Bardyn et~al.~\cite{Bardyn2009} obtained this curve for the singlet
fidelity maximized over local unitaries; our result extends it to every
number of parties greater than two.
The qubit curve lies strictly above the unrestricted affine minimum
at every interior score, \(0<t<1\), so attaining the affine minimum there
requires at least one local dimension above two.

Earlier attaining constructions use a classical register and a data
qubit at every party~\cite{Kaniewski2016,Jha2026}.
We attain the entire affine bound on one qutrit and qubits at all
remaining parties, with one collection of measurements fixed as the
score varies, as in the register constructions.
The state mixes, with weights \(t\) and \(1-t\), an \(n\)-party GHZ
state and the product of the qutrit's third level with a GHZ state on
the remaining parties; both qutrit Bell observables take the value
\(+1\) on that level.
At interior scores, we prove that among all states attaining the affine
minimum the smallest product of local ranks is \(3\cdot2^{n-1}\),
achieved by this mixture; the local ranks are the ranks of the
single-party reduced states.

We next ask what extractability at or near the affine minimum forces.
For \(n\geq4\) and projective measurements in arbitrary finite local
dimensions, we consider the state after each party measures projectively
onto subspaces of dimension at most two that are invariant under both of
its Bell observables; this local measurement leaves the score unchanged.
We prove that if the original state's extractability exceeds the affine
minimum by \(\varepsilon\), then the outcomes at which, at two or more parties, the
angle between the two observables on the selected subspace is \(\varphi\)
or more away from the right angle at which they anticommute have total
probability at most a constant times \(\varepsilon/\varphi^4\), uniformly
in \(n\) and in the score; commuting observables count as collinear.
At zero excess, with irreducible invariant subspaces, we show that every
conditional state is a locally embedded GHZ state on all parties, with
anticommuting observables at every party, or the product of a state of
one party, on which both of its Bell observables take fixed values, with
such a GHZ state on the others; the two kinds have total weights \(t\)
and \(1-t\), respectively.
These two results concern the state after the local measurements.

For \(n\geq4\) and arbitrary binary POVMs, the smallest product of local
ranks determines the original state itself: at each interior score, we
show that the mixture is the unique attaining state whose local ranks
have product \(3\cdot2^{n-1}\), up to local isometries on its supports
and a permutation of the parties.
On one qutrit and qubits elsewhere, still with arbitrary binary POVMs,
this uniqueness is stable: for \(n\geq4\), if the score is at least a
fixed fraction of the way from the biseparable bound to the quantum
maximum, a sufficiently small excess \(\varepsilon\) of extractability
above the affine minimum implies trace-norm distance at
most \(C\sqrt\varepsilon\) from the attaining mixture of the same score,
up to local unitaries.
The constant and the admissible excess depend only on that fraction,
uniformly in \(n\); the score interval includes the quantum maximum, and
we show that the exponent \(1/2\) is optimal and that the lower cutoff on
the score cannot be removed.
For each \(b>\sqrt{2\sqrt2-2}\approx0.910\), we bound this distance by
\(C(b)\sqrt\varepsilon\) for every excess and every \(n\geq3\) on one
qutrit and \(n-1\) qubits with arbitrary binary POVMs.

These equality and stability results assume that the actual
extractability is minimal or nearly minimal at the measured score.
The Bell statistics do not determine this: we show that states with
identical outcome probabilities for all prescribed Bell measurements
can have different extractabilities, as in the attaining mixture with
coherence added between its two components.
Additional local data can take the place of this assumption.
Suppose that the local dimensions are known to be those of one qutrit and
\(n-1\) qubits, and that a two-dimensional qutrit subspace invariant under
both Bell observables of the qutrit is known; the Bell observables may
otherwise be unknown.
Then, for \(n\geq3\) and every score above the biseparable bound, we
bound the distance of the original state from the attaining mixture of
the same score, up to local unitaries, by the MABK score and five
additional calibrated binary measurements at the qutrit, which measure
the population of that subspace and its coherences with the orthogonal
level.
We show that the population is at least \(t\) and that the distance is at
most a constant times the square root of the difference between the
population and \(t\), plus a constant times the size of the coherences,
with constants independent of the score and of \(n\); the bound vanishes
when the population equals \(t\) and the coherences vanish.
With the Bell observables of the attaining strategy fixed and calibrated
at every party, three additional settings at the qutrit suffice, and the
estimate includes the biseparable endpoint.
In both estimates the measured expectation values are taken as exactly
known.

The target of this certificate is a single mixed state, the attaining
mixture, which has rank two at interior scores, whereas
Li, Han and Zhu constructed local verification protocols for pure GHZ
states~\cite{LiHanZhu2020}, Chen et~al.\ treat verification of quantum
subspaces~\cite{ChenEtAl2025}, and Han et~al.\ study verification of
pure GHZ states in networks with trusted and untrusted
parties~\cite{HanLiWangZhu2021}.

Section~\ref{sec:model} specifies the model and normalization.
Section~\ref{sec:results} states the unrestricted and qubit optima and the
minimum product of local ranks.
Section~\ref{sec:bound} proves the lower bound.
Section~\ref{sec:dimension} treats qubit strategies and constructs the
attaining mixture.
Section~\ref{sec:stability} gives the concentration estimates and the
equality classification after local projective measurements.
Section~\ref{sec:original} proves uniqueness and stability of the original
state, and Section~\ref{sec:local-certification} gives the certificate from
additional local measurements.
Section~\ref{sec:discussion} concludes with consequences and open questions.

\section{MABK correlations and GHZ extractability}
\label{sec:model}

We now fix the normalization of the MABK operators, which differs
between the works cited above, and make precise the scores and the
extractability used in the Introduction.
The least extractability at a given score is taken over all states and
measurements that realize that score.

A strategy consists of a density operator \(\rho\) on a tensor product of
local Hilbert spaces of \(n\geq3\) parties, each of finite or infinite
dimension, and, at each party \(j\), two binary POVMs
\(\{E_{j,+\mid x},E_{j,-\mid x}\}\), one for each input \(x\). Their observables, the Bell observables of party \(j\),
\[
 A_{j,x}\coloneqq E_{j,+\mid x}-E_{j,-\mid x},
 \qquad x\in\{0,1\},
\]
are Hermitian contractions. Define the Bell operators recursively by
\begin{equation}
 W_1\coloneqq2A_{1,0},\qquad
 W_n\coloneqq\frac12W_{n-1}\otimes(A_{n,0}+A_{n,1})
       -\frac12\overline W_{n-1}\otimes(A_{n,0}-A_{n,1}),
 \label{eq:bell}
\end{equation}
where a bar interchanges the two inputs at every preceding party.
This is the MABK convention used in \cite{Cao2026}. The classical bound
of \(W_n\) is \(2\), by induction on \(n\) from \(W_1=2A_{1,0}\): for a deterministic last party, with \(A_{n,0},A_{n,1}\in\{\pm1\}\), the recursion selects either
\(\pm W_{n-1}\) or \(\pm\overline W_{n-1}\).
A permutation of the parties, combined with exchanging the two inputs of
the parties moved to and from the first position, leaves \(W_n\)
unchanged, and the extractability \(\Xi_n\) defined below is invariant
under permutations of the parties
(Remark~\ref{rem:qutrit-position}). Results stated below with
a distinguished party therefore hold, after relabeling, with any party
in that role.
The score of a strategy is its MABK value. Write
\[
 \beta\coloneqq\tr(\rho W_n),\qquad
 \beta_Q\coloneqq2^{(n+1)/2},\qquad \beta_B\coloneqq2^{n/2}.
\]
The value \(\beta_B\) is the largest MABK value attainable by biseparable
states, that is, mixtures of states each of which is a product across
some bipartition of the parties, in arbitrary finite local dimension and
with binary POVMs. It is the case of a partition into two subsets in
\cite[Eq.~(18)]{WernerWolf2000}. That inequality holds for every state
with positive partial transpose across the bipartition and for all
observables with spectrum in \([-1,1]\); it is stated in the
normalization with classical bound \(1\), where all values are half of
ours. See also \cite[Eq.~(28)]{NagataKoashiImoto2002} (equation numbers
of the arXiv versions).
The value \(\beta_Q\) is the quantum maximum, which
\cite[Sec.~III]{WernerWolf2000} states as \(2^{(n-1)/2}\) in that paper's normalization; it also follows from the block formula
\eqref{eq:Wblock} below once the observables are brought to the
canonical form \eqref{eq:canonical}.
For \(\beta_B\leq\beta\leq\beta_Q\), the normalized score and its
position between the two endpoints are
\[
 b\coloneqq\frac{\beta}{\beta_Q},\qquad
 t\coloneqq\frac{\beta-\beta_B}{\beta_Q-\beta_B}.
\]
Thus \(1/\sqrt2\leq b\leq1\) and \(t=(b-1/\sqrt2)/(1-1/\sqrt2)\).
We use the affine functions
\begin{equation}
 L(\beta)\coloneqq\frac{1+t}{2},\qquad
 \mathcal L(b)\coloneqq\left(1+\frac1{\sqrt2}\right)b-\frac1{\sqrt2}
             =L(\beta_Qb).
 \label{eq:affine-notation}
\end{equation}

Let
\[
 \ket{\GHZ_n}\coloneqq\frac{\ket{0}^{\otimes n}+\ket{1}^{\otimes n}}{\sqrt2},
 \qquad G_n\coloneqq\proj{\GHZ_n}.
\]
The extractability is
\begin{equation}
 \Xi_n(\rho)\coloneqq
 \sup_{\Lambda_1,\ldots,\Lambda_n}
 \tr\!\left[G_n(\Lambda_1\otimes\cdots\otimes\Lambda_n)(\rho)\right],
 \label{eq:extractability}
\end{equation}
where each \(\Lambda_j\) is a completely positive trace-preserving map
from the \(j\)-th local system to a qubit. The same definition applies to normal states (density operators) on infinite-dimensional tensor products.
For a pure target, the overlap in \eqref{eq:extractability} is the
fidelity of the extracted state with \(G_n\), in the convention
\(\|\sqrt\rho\sqrt\sigma\|_1^2\) for the fidelity of two states
\(\rho,\sigma\) \cite{Kaniewski2016},
where \(\|X\|_1\coloneqq\tr|X|\) is the trace norm.

\section{The unrestricted and qubit optima}
\label{sec:results}

This section states two optimal tradeoffs between the MABK score and
GHZ extractability, the unrestricted one for arbitrary local dimensions
and the one for qubits. It also gives the smallest product of local
ranks (the ranks of the single-party reduced states, whose product is
the dimension of the tensor product of the local supports) with which the
unrestricted optimum is attained at an interior score.
Section~\ref{sec:bound} proves the lower bound of
Theorem~\ref{thm:main}, and Section~\ref{sec:dimension} proves its
attainment together with Theorems~\ref{thm:qubit-profile}
and~\ref{thm:minimal-dimension}.

\begin{theorem}
\label{thm:main}
For every \(n\geq3\) and every strategy in this model whose score
\(\beta\) lies in \([\beta_B,\beta_Q]\),
\begin{equation}
 \Xi_n(\rho)\geq\frac{1+t}{2}
 =\frac{1+\sqrt2}{2^{n/2+1}}\beta-\frac1{\sqrt2}.
 \label{eq:main}
\end{equation}
For every \(\beta\in[\beta_B,\beta_Q]\), the infimum of \(\Xi_n(\rho)\) over
all states and measurements of score \(\beta\) is exactly \((1+t)/2\). This infimum is a minimum: it is attained with local
dimensions \(3,2,\ldots,2\) and with measurements that may be chosen
independently of \(t\).
\end{theorem}

At \(n=3\), \eqref{eq:main} for finite-dimensional strategies follows
from Kaniewski's Proposition~2 combined with the reduction sketched in his
Appendix~A \cite{Kaniewski2016} (numbering of arXiv:1604.08176v3
throughout), once \(W_3\) is identified with his
Mermin operator; Section~\ref{sec:extraction} gives this identification,
carries out the reduction and extends the bound to every strategy in the
model. The remainder of the lower-bound proof
treats every \(n\geq4\).
The tight \(n=4,5\) bounds on the overlap in \eqref{eq:extractability} are stated by Jha, Singh and Pan
\cite[Theorem~2]{Jha2026}, whose Bell operator is \(2^{(n-3)/2}W_n\) up to
local relabelings of inputs and outcomes (see Remark~\ref{rem:qutrit-position}), and who indicate that the proof
is similar to that of their Theorem~1, which concerns a different Bell operator.
For \(n\geq6\), the two remaining scalar cases in the reduction of
\cite[Supplemental C]{Cao2026} were proved in \cite{ElonSG2026}.
Our proof treats every \(n\geq4\) in one argument and supplies the block
estimates used in Theorem~\ref{thm:uniform-gap}.

The construction in Theorem~\ref{thm:qutrit-attainer} attains the bound with these local dimensions; Theorem~\ref{thm:minimal-dimension} shows that at interior scores every attaining strategy has local ranks with product at least \(3\cdot2^{n-1}\), the value for this construction. Appendix~\ref{app:flags} gives the earlier type of attaining
construction, with a classical register and a data qubit at every party
\cite{Kaniewski2016,Jha2026}; Theorem~\ref{thm:qutrit-attainer} replaces these registers by one additional level at a single party.

Local channels with the constant output \(\ket0\) give \(\Xi_n(\rho)\geq1/2\) for every state.
Thus, for every score \(\beta\), the lower bound takes the globally valid form
\[
 \Xi_n(\rho)\geq
 \max\!\left\{\frac12,\frac{1+\sqrt2}{2^{n/2+1}}\beta-\frac1{\sqrt2}\right\}.
\]
In terms of the relative deficiency \(1-b\) used by Cao et al.~\cite{Cao2026}, the affine term is \(1-(1+1/\sqrt2)(1-b)\).

When every local system is a qubit, the least extractability is a
different function of the normalized score \(b=\beta/\beta_Q\).

\begin{theorem}
\label{thm:qubit-profile}
Among strategies with a qubit at every party and arbitrary binary
POVMs, the least extractability at score \(\beta_Qb\) is
\begin{equation}
 F(b)\coloneqq\frac{1+\sqrt{2b^2-1}}2,
 \qquad \frac1{\sqrt2}\leq b\leq1.
 \label{eq:qubit-profile}
\end{equation}
\end{theorem}

This extends to every number of parties greater than two the tight
two-qubit CHSH curve of Bardyn et al.~\cite{Bardyn2009}, which they
obtained for the singlet fidelity maximized over local unitaries, with
projective measurements. Section~\ref{sec:qubit} combines their
argument with the spectral facts of Scarani and
Gisin~\cite{ScaraniGisin2001}, which hold for every number of parties.
Since \(F''(b)=-(2b^2-1)^{-3/2}<0\) for \(1/\sqrt2<b<1\), the function
\(F\) is strictly concave on \([1/\sqrt2,1]\). Its endpoint values are
\(F(1/\sqrt2)=1/2\) and \(F(1)=1\), and \(\mathcal L\) is its chord
between these endpoints. Hence
\begin{equation}
 F(b)>\mathcal L(b),\qquad \frac1{\sqrt2}<b<1,
 \label{eq:strict-gap}
\end{equation}
that is, at every interior score the least extractability of qubit
strategies is strictly larger than the unrestricted minimum.
This strict gap gives a lower bound on the local ranks of every strategy
that attains the unrestricted minimum at an interior score.

\begin{theorem}
\label{thm:minimal-dimension}
Let \(n\geq3\), \(1/\sqrt2<b<1\), and let a strategy of score
\(\beta_Qb\) satisfy \(\Xi_n(\rho)=\mathcal L(b)\).
Let \(r_i\) be the rank of the \(i\)-th reduced state (the \(i\)-th local rank). Then every
\(r_i\geq2\), at least one \(r_i\geq3\), and
\begin{equation}
 \prod_{i=1}^n r_i\geq3\cdot2^{n-1}.
 \label{eq:minimal-dimension}
\end{equation}
The construction \eqref{eq:qutrit-attainer} of Theorem~\ref{thm:qutrit-attainer} attains equality.
If some \(r_i\) is infinite, \eqref{eq:minimal-dimension} holds trivially.
\end{theorem}

\begin{figure}[!ht]
\centering
\includegraphics[width=280pt,alt={Least GHZ extractability versus MABK score: the solid qubit curve lies above the dotted unrestricted line between their common endpoints.}]{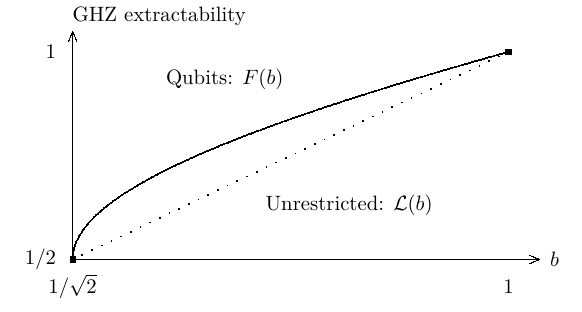}%
\caption{Exact least GHZ extractability at a normalized MABK score
\(b=\beta/\beta_Q\), for every \(n\geq3\). The solid curve
\(F(b)=(1+\sqrt{2b^2-1})/2\) applies when every local system is a
qubit. The dotted affine line
\(\mathcal L(b)=(1+1/\sqrt2)b-1/\sqrt2\) applies with arbitrary local
dimensions and is attained with one qutrit and qubits elsewhere.
The curves coincide at \(b=1/\sqrt2\) and \(b=1\), and are strictly
separated in the interior. Each curve is a minimum over its own
admissible class of strategies.}
\label{fig:tradeoffs}
\end{figure}

Theorem~\ref{thm:minimal-dimension} concerns strategies attaining the least extractability at an interior score; by \eqref{eq:strict-gap}, qubit strategies reach such a score only with strictly greater extractability. Figure~\ref{fig:tradeoffs} shows the strict separation of the two minima.

The remaining results are of three kinds, with different hypotheses. The
first two concern strategies whose extractability equals or nearly equals
the unrestricted minimum.
Section~\ref{sec:stability} concerns the state after each party measures
projectively onto subspaces of dimension at most two that are invariant
under both of its Bell observables (local pinching). For \(n\geq4\) and
finite-dimensional projective strategies, a small excess
\(\Xi_n(\rho)-L(\beta)\) forces the outcome distribution of these local
measurements to concentrate near configurations with anticommuting
observables at all but at most one party
(Theorem~\ref{thm:actual-stability}), and at zero excess the pinched state
is classified (Theorem~\ref{thm:refined-equality}).

Section~\ref{sec:original} concerns the original state, with arbitrary
binary POVMs. For \(n\geq4\) and interior scores, the attaining state
whose local ranks have product \(3\cdot2^{n-1}\) is unique up to local
isometries and a permutation of the parties
(Theorem~\ref{thm:minimal-equality-state}). On one qutrit and qubits it is
also stable for \(n\geq4\): on every fixed interior score interval, a small excess bounds
the trace-norm distance to the attaining mixture \eqref{eq:qutrit-attainer} of the same score, up to local unitaries, by
a constant times the square root of the excess
(Theorem~\ref{thm:original-state-stability}), and at high scores an
explicit estimate holds for \(n\geq3\)
(Proposition~\ref{prop:high-score-stability}). For \(n\geq4\) the two
combine to a square-root bound on every interval \(\eta\leq t\leq1\)
with \(\eta>0\) (Proposition~\ref{prop:stability-to-maximum}).

Section~\ref{sec:local-certification} bounds the same distance when the hypothesis on the actual
extractability is replaced by the MABK score and five calibrated binary measurements on
the qutrit, for \(n\geq3\) and scores above the biseparable bound, given
the local dimensions and a known two-dimensional qutrit subspace invariant under
both qutrit Bell observables (Theorem~\ref{thm:local-certification}).

\section{Proof of the lower bound}
\label{sec:bound}

For \(n\geq4\) we prove the lower bound of Theorem~\ref{thm:main} here;
for \(n=3\) the bound for finite-dimensional strategies is Kaniewski's
\cite{Kaniewski2016}, and the end of this section identifies his
operator with \(W_3\) and extends the bound to every strategy in the
model.
The proof uses the local extraction channel of
Kaniewski~\cite[Eq.~(9)]{Kaniewski2016} at every party, as do Cao et
al.~\cite{Cao2026}.
For given measurements, we seek one product of these channels, each
depending only on its own party's measurements, whose target overlap
is at least \(L(\beta)\) for every state.
Pulling the target projector back through the channels turns that
requirement into a state-independent operator inequality, valid for
all measurement angles and all \(n\geq4\).
Jordan's lemma reduces its verification to tensor products of local
invariant subspaces of dimension at most two (the joint sectors defined
below).
The channel at each party must depend only on that party's sector:
separately optimized extractions on joint sectors need not be
implementable by separate parties, as already observed in the
bipartite setting \cite{Bardyn2009}.
Section~\ref{sec:stability} strengthens the same operator inequality to
a spectral bound and derives the structural results from it.
We first construct the channels and compute the two-dimensional blocks
of the operators in this inequality, then prove the scalar positivity estimate and return to local
extraction.

Throughout, set
\begin{equation}
 s\coloneqq\sqrt2,\qquad k\coloneqq1+s,\qquad a\coloneqq\frac{k}{s},\qquad
 \kappa\coloneqq1-\frac1s.
 \label{eq:constants}
\end{equation}
Note that \(s k=2+s\).
The optimal slope and intercept are
\[
 \bar s_n\coloneqq\frac{k}{2^{n/2+1}},\qquad \bar\mu\coloneqq-\frac1s.
\]

\subsection{Canonical qubits and the target}

First, the two binary POVMs of each party admit a Naimark dilation with
one common embedding, after which both observables are reflections, that is, Hermitian unitaries (projective binary observables). In finite dimension, Jordan's lemma then decomposes each local space into an orthogonal sum of Jordan blocks: subspaces of dimension at most two that are invariant under both reflections. We call these Jordan blocks local sectors, keeping the word block for the two-by-two blocks of the operators in Lemma~\ref{lem:blocks}, and we call a tensor product of one local sector per party a joint sector.
Appendix~\ref{sec:lift} gives the dilation and the Jordan decomposition, with their local implementation. On a two-dimensional sector on which the reflections do not commute, the two observables have, in suitable local coordinates (a canonical qubit),
the form
\begin{equation}
 A_{j,0}=\cos\theta_j X+(-1)^{j-1}\sin\theta_j Z,\qquad
 A_{j,1}=\cos\theta_j X-(-1)^{j-1}\sin\theta_j Z,
 \qquad 0\leq\theta_j\leq\frac{\pi}{2}.
 \label{eq:canonical}
\end{equation}
Here \(X,Y,Z\) are Pauli matrices, and \eqref{eq:canonical} is the
canonical pair of angle \(\theta_j\).

The decomposition may also contain two-dimensional sectors on which the
reflections commute, provided they take the form \eqref{eq:canonical}
with \(\theta_j\in\{0,\pi/2\}\); the decomposition therefore need not be
unique, and Section~\ref{sec:equality} refines it into irreducible
sectors. A one-dimensional sector (a scalar sector), on which the two observables take values
\(a_0,a_1\in\{\pm1\}\), is embedded in a canonical qubit as the common
eigenvector on which the canonical pair acts by \(a_0,a_1\): at \(\theta=0\)
if \(a_0=a_1\), and at \(\theta=\pi/2\) if \(a_1=-a_0\). The proofs below act
on each joint sector separately, so Theorem~\ref{thm:main} and Theorems~\ref{thm:actual-stability}
and~\ref{thm:necessary-equality} hold for every such decomposition;
Theorem~\ref{thm:refined-equality} uses its refinement into irreducible sectors.

With \eqref{eq:canonical}, the recursion \eqref{eq:bell} becomes
\begin{align}
 W_n&=W_{n-1}\otimes\cos\theta_n X
       +(-1)^n\overline W_{n-1}\otimes\sin\theta_n Z,
 \label{eq:qubit-recursion}\\
 \overline W_n&=\overline W_{n-1}\otimes\cos\theta_n X
       +(-1)^{n-1}W_{n-1}\otimes\sin\theta_n Z.
 \nonumber
\end{align}

For \(\sigma\in\{+1,-1\}\), write
\[
 \ket{\sigma}_Y\coloneqq\frac{\ket0+\mathrm i\sigma\ket1}{\sqrt2}.
\]
For a sign string \(\sigma=(\sigma_1,\ldots,\sigma_n)\), put \(\ket{\sigma}_Y\coloneqq\bigotimes_j\ket{\sigma_j}_Y\). Define
\[
 \tau_1\coloneqq1,\qquad \tau_j\coloneqq(-1)^j\quad(j\geq2),\qquad
 \delta_n\coloneqq\frac{2+(-1)^n}{4}\pi.
\]
The fixed canonical target is
\begin{equation}
 \ket{\Phi_n}
 \coloneqq\frac{e^{-\mathrm i\delta_n/2}\ket{\tau}_Y
          +e^{\mathrm i\delta_n/2}\ket{-\tau}_Y}{\sqrt2}.
 \label{eq:target}
\end{equation}
When every \(\theta_j=\pi/4\), it is the eigenvector of \(W_n\) with eigenvalue \(\beta_Q\) (see after Lemma~\ref{lem:blocks}). It is locally equivalent to \(\ket{\GHZ_n}\). Indeed,
\[
 V_j\coloneqq\ket0\bra{\tau_j}_Y+\ket1\bra{-\tau_j}_Y
\]
maps the two local \(Y\) eigenvectors to the computational basis.
After applying \(\bigotimes_jV_j\), the additional unitary
\(\operatorname{diag}(e^{\mathrm i\delta_n/2},
e^{-\mathrm i\delta_n/2})\) on the first output qubit gives
\(\ket{\GHZ_n}\).

The target \eqref{eq:target} agrees with the recursive target of
\cite{Cao2026}; Appendix~\ref{app:canonical} verifies the phase convention.

\subsection{Extraction channels and the determinant}
\label{sec:channels}

Use at every party the channel of \cite[Eq.~(9)]{Kaniewski2016}, as in \cite{Cao2026}:
\begin{equation}
 \Lambda_\theta(B)\coloneqq\frac{1+g(\theta)}2 B+
                \frac{1-g(\theta)}2\Gamma(\theta)B\Gamma(\theta),
 \label{eq:channel}
\end{equation}
where
\[
 g(\theta)\coloneqq k(\sin\theta+\cos\theta-1),\qquad
 \Gamma(\theta)\coloneqq
 \begin{cases}
 X,&0\leq\theta\leq\pi/4,\\
 Z,&\pi/4<\theta\leq\pi/2.
 \end{cases}
\]
Since \(0\leq g(\theta)\leq1\), each map is completely positive,
trace-preserving and self-adjoint. Here \(g=1\) (no dephasing) at the
ideal angle \(\pi/4\) and \(g=0\) (full dephasing) at the commuting
endpoints \cite{Kaniewski2016}. Write
\[
 K\coloneqq
 \left(\bigotimes_{j=1}^n\Lambda_{\theta_j}\right)(\proj{\Phi_n}),
 \qquad T\coloneqq K-\bar s_n W_n+\frac1s I.
\]
Self-adjointness identifies the expectation of \(K\) with the
extracted target overlap. Therefore the bound of Theorem~\ref{thm:main}
for \(n\geq4\) follows from the operator inequality \(T(\theta)\geq0\) for
all \(\theta\in[0,\pi/2]^n\), which
Sections~\ref{sec:channels}--\ref{sec:positivity} prove, with
Appendices~\ref{app:covariance} and~\ref{app:positivity}, and
Theorem~\ref{thm:uniform-gap} strengthens. Equivalently, \(K\geq\bar s_nW_n+\bar\mu I\), whose coefficients are the slope and intercept of \(L\).

Appendix~\ref{app:covariance} proves that reflection of any angle
about \(\pi/4\) conjugates \(K,W_n,T\) by a product unitary preserving
the target. Positivity may therefore be checked on the reduced cube \([0,\pi/4]^n\);
the implemented channel still uses only its own local angle.
For \(0\leq\theta_j\leq\pi/2\) put
\begin{equation}
 v_j\coloneqq1-\frac{\cos\theta_j+\sin\theta_j}{s},\quad
 p_j\coloneqq av_j,\quad q_j\coloneqq1-av_j,\quad
 b_j\coloneqq v_j(2-v_j),\quad A_j\coloneqq(1-v_j)^2.
 \label{eq:variables}
\end{equation}
These formulas are invariant under \(\theta_j\mapsto\pi/2-\theta_j\); until
the end of Section~\ref{sec:positivity} we assume \(0\leq\theta_j\leq\pi/4\),
where the stated ranges and square roots hold.
Thus \(v_j\in[0,\kappa]\), \(p_j\in[0,1/2]\),
\(b_j\in[0,1/2]\), and \(A_j=1-b_j\). Moreover,
\[
 \sqrt{A_j}=\frac{\cos\theta_j+\sin\theta_j}{s},\qquad
 \sqrt{b_j}=\frac{\cos\theta_j-\sin\theta_j}{s}.
\]
The weights \(p_j=(1-g(\theta_j))/2\) and \(q_j=(1+g(\theta_j))/2\) are the
probabilities of the dephasing and identity terms of \eqref{eq:channel}. The weights \(b_j,A_j\) determine the Bell amplitudes \(x,y\) defined below and are distinct from the normalized score \(b\)
and the observables \(A_{j,x}\).
For a set \(J\) of indices define
\[
 \begin{gathered}P_J\coloneqq\prod_{j\in J}p_j,\quad Q_J\coloneqq\prod_{j\in J}q_j,\quad
 B_J\coloneqq\prod_{j\in J}b_j,\quad A_J\coloneqq\prod_{j\in J}A_j,\\
 E_J\coloneqq P_J\sqrt{B_J},\quad C_J\coloneqq Q_J\sqrt{A_J}.\end{gathered}
\]
An empty product is \(1\).

Every Pauli tensor product in \(W_n\) flips all local \(Y\) eigenvalues, and
every product of the local dephasing operators \(\Gamma(\theta_j)\) in the
expansion of \(\bigotimes_j\Lambda_{\theta_j}\) maps \(\ket{\Phi_n}\) into one
of the planes \(\operatorname{span}\{\ket{\sigma}_Y,\ket{-\sigma}_Y\}\). Hence
each such plane is invariant under \(K\) and \(W_n\), and \(T\geq0\) reduces
to trace and determinant conditions on \(2^{n-1}\) two-by-two blocks.
Let \(H\sqcup D=[n]\), \(h\coloneqq|H|\), \(d\coloneqq|D|\), and
\(\sigma_j\coloneqq\tau_j(-1)^{\mathbf1_H(j)}\).
The vectors \(\ket{\sigma}_Y,\ket{-\sigma}_Y\) span a common block
of \(K\) and \(W_n\); write \(K_H\), \(W_H\) and \(T_H\) for the restrictions of \(K\), \(W_n\) and \(T\) to it. Complementary sets describe the same
block with reversed basis. Taking \(1\in D\) gives all
\(2^{n-1}\) orthogonal blocks. Set
\[
 w\coloneqq P_HQ_D,\quad z\coloneqq Q_HP_D,\quad
 x\coloneqq\sqrt{B_HA_D},\quad y\coloneqq\sqrt{A_HB_D},\quad
 \zeta\coloneqq(-1)^h e^{-\mathrm i\delta_n}.
\]

\begin{lemma}
\label{lem:blocks}
In the ordered basis \((\ket{\sigma}_Y,\ket{-\sigma}_Y)\),
\begin{align}
 K_H
 &=\frac12
 \begin{pmatrix}
 w+z&\zeta(w-\mathrm iz)\\
 \overline\zeta(w+\mathrm iz)&w+z
 \end{pmatrix},
 \label{eq:Kblock}\\
 W_H
 &=\beta_Q
 \begin{pmatrix}
 0&\zeta(x-\mathrm iy)\\
 \overline\zeta(x+\mathrm iy)&0
 \end{pmatrix}.
 \label{eq:Wblock}
\end{align}
Consequently \(\tr(T_H)=w+z+s>0\), and
\begin{align}
 \lambda_{H,D}\coloneqq2\det(T_H)
 ={}&1+P_HQ_HP_DQ_D+s(Q_HP_D+P_HQ_D)
 \nonumber\\
 &-k^2(A_HB_D+B_HA_D)
       +s k(C_HE_D+E_HC_D).
 \label{eq:lambda}
\end{align}
\end{lemma}

\begin{proof}
On the reduced angle cube, the channel is a mixture of tensor products
of \(I\) and \(X\). Exactly two of these products send the target into the indicated block: the one with \(X\) factors precisely on \(H\) and the one with \(X\) factors precisely on \(D\). Their weights in
this mixture are \(w\) and \(z\), respectively. The corresponding
target projectors have upper off-diagonal entries
\((-1)^h e^{-\mathrm i\delta_n}/2\) and
\((-1)^{n-h}e^{\mathrm i\delta_n}/2\).
The identity
\(e^{2\mathrm i\delta_n}=(-1)^{n+1}\mathrm i\)
gives \eqref{eq:Kblock}.

Each Pauli tensor product in \(W_n\) has an \(X\) or \(Z\) factor at
every site, so it flips all the local \(Y\) eigenvalues. Its block has
zero diagonal.
For local eigenvalue \(\sigma\), the upper entries of \(X,Z\) are
\(-\mathrm i\sigma,1\). If \(u_{n-1}\) is the upper entry of
\(W_{n-1}\), then the upper entry of \(\overline W_{n-1}\) is
\((-1)^{n-1}\overline{u_{n-1}}\), where the bar on \(u_{n-1}\) denotes complex conjugation. This follows by conjugating
with \(X^{\otimes(n-1)}\), which interchanges the two observables.
The recursion therefore gives
\[
 u_n=-\mathrm i\sigma_n\cos\theta_n\,u_{n-1}
            -\sin\theta_n\,\overline{u_{n-1}}.
\]
The base is
\(u_1=2(\sin\theta_1-\mathrm i\sigma_1\cos\theta_1)\), which equals
\(\beta_Q\zeta(x-\mathrm iy)\) at \(n=1\) for either sign \(\sigma_1\), that is, for \(1\in D\) and for \(1\in H\).
For the induction step, let \(\zeta',x',y'\) be formed as above from
\(H\cap[n-1]\) and \(D\cap[n-1]\), and suppose that
\(u_{n-1}=2^{n/2}\zeta'(x'-\mathrm iy')\). Insert \(\sigma_n=(-1)^n\) for
\(n\in D\) and its negative for \(n\in H\). Since
\(\delta_n-\delta_{n-1}=(-1)^n\pi/2\), we have
\(e^{-\mathrm i\delta_{n-1}}=\mathrm i(-1)^ne^{-\mathrm i\delta_n}\), and the
identity above gives \(\mathrm i(-1)^ne^{2\mathrm i\delta_n}=1\). With
\(\zeta\) formed from \(H\) and \(D\), so that \(h\) counts \(n\) when
\(n\in H\), the recursion becomes
\[
 u_n=\frac{\beta_Q\zeta}{\sqrt2}
 \begin{cases}
 (\cos\theta_n+\sin\theta_n)x'-\mathrm i(\cos\theta_n-\sin\theta_n)y',
   &n\in D,\\
 (\cos\theta_n-\sin\theta_n)x'-\mathrm i(\cos\theta_n+\sin\theta_n)y',
   &n\in H.
 \end{cases}
\]
In both cases the right side is \(\beta_Q\zeta(x-\mathrm iy)\): for
\(n\in D\), \(x=\sqrt{A_n}\,x'\) and \(y=\sqrt{b_n}\,y'\), while for
\(n\in H\), \(x=\sqrt{b_n}\,x'\) and \(y=\sqrt{A_n}\,y'\). With the base,
this proves \eqref{eq:Wblock} by induction.

Since \(\bar s_n\beta_Q=k/s\), the common diagonal of \(T_H\) is
\((w+z)/2+1/s\), and its upper off-diagonal entry is
\[
 \zeta\left[\frac{w-\mathrm iz}{2}-\frac{k}{s}(x-\mathrm iy)\right].
\]
Directly taking its determinant yields
\[
 2\det(T_H)=1+wz+s(w+z)-k^2(x^2+y^2)+s k(wx+zy),
\]
which is \eqref{eq:lambda}.
\end{proof}

The expression \(x^2+y^2\) is the sum of two complementary atoms of
the product distribution with local probabilities \(A_j,b_j\).
Hence \(x^2+y^2\leq1\), and \eqref{eq:Wblock} gives
\(\|W_n\|\leq\beta_Q\).
At \(\theta_j=\pi/4\) for all \(j\), the target block has \(x=1,y=0\);
its positive eigenvector is \(\ket{\Phi_n}\), with eigenvalue \(\beta_Q\).
Formula \eqref{eq:lambda} agrees with the scalar expression in
\cite[Supplemental C, Eq.~S57]{Cao2026}. Like that expression, it is \emph{twice} the block determinant.

\subsection{Uniform positivity}
\label{sec:positivity}

The determinant combines products of channel probabilities and Bell
amplitudes; the next lemma bounds these products. Put
\(r\coloneqq(a+1)/2\).

\begin{lemma}
\label{lem:products}
For every nonempty set \(J\) of size \(m\),
\begin{equation}
 Q_J\geq A_J,\qquad
 C_J\geq1-r(1-A_J),\qquad
 B_J\leq\frac{1-A_J}{2^m-1}.
 \label{eq:products}
\end{equation}
\end{lemma}

The proof of Lemma~\ref{lem:products} is given in
Appendix~\ref{app:products}. Its three estimates couple all products
to the same quantity \(A_J\), so that the determinant can be bounded
by two simpler expressions.

Discard the three nonnegative terms
\(P_HQ_HP_DQ_D\), \(sQ_HP_D\), \(s k C_HE_D\) from
\eqref{eq:lambda}, then apply Lemma~\ref{lem:products} to \(D\).
The result is
\begin{equation}
 \lambda_{H,D}\geq A_Df_h+(1-A_D)g_{h,d},
 \label{eq:interpolation}
\end{equation}
where
\begin{equation}
 f_h\coloneqq1+sP_H-k^2B_H+s k E_H,\qquad
 g_{h,d}\coloneqq1-\frac{k}{2}E_H-\frac{k^2A_H}{2^d-1}.
 \label{eq:fg}
\end{equation}
Here \(s k(r-1)=k/2\). Both interpolation weights in
\eqref{eq:interpolation} are nonnegative.

Because \eqref{eq:lambda} is symmetric in \(H,D\), assume
\(h\leq d\). In the case \(h=1\), the exact factorization \eqref{eq:f1} and the bound \(g_{1,d}>0\) below show that the determinant can vanish only when the index in \(H\) is at a commuting endpoint and every index in \(D\) is ideal.

For \(h=1\), write \(v\coloneqq v_H\) for the variable of the single index in \(H\), and
\(u\coloneqq\sqrt{v/(2-v)}\in[0,1/k]\), so that \(v=2u^2/(1+u^2)\) and
\(\sqrt{v(2-v)}=2u/(1+u^2)\). Since \(k^2=2k+1\),
\begin{align}
 f_1(v)
 &=1+kv-k^2v(2-v)+k^2v\sqrt{v(2-v)}
 \nonumber\\
 &=\frac{(1-ku)^2(1+2ku+u^2)}{(1+u^2)^2}\geq0.
 \label{eq:f1}
\end{align}
For \(n\geq4\), \(d\geq3\). With \(p\coloneqq av\),
\(A_H\leq1-p\) and \(E_H\leq p/s\), so
\[
 g_{1,d}\geq1-\frac{k}{2s}p-\frac{k^2}{7}(1-p)>0.
\]
The right side is affine on \(0\leq p\leq1/2\); its endpoint values are
\((4-2s)/7\) and \((30-15s)/56\), both positive.
Equation \eqref{eq:interpolation} proves positivity in this case.

The remaining cases are verified in Appendix~\ref{app:positivity}.
For \(h\geq2\), one estimate for all these partitions, based on \(p_j\geq b_j^{9/8}\), gives a uniform positive lower bound on the determinant. The case \(h=0\) uses a
polynomial lower bound that is concave in each variable and checks that
polynomial at every type of vertex, including the vertex
\((\kappa,\ldots,\kappa)\).
The case \(h\geq3\) was proved in \cite[Supplemental C, Case 1]{Cao2026}.
For \(h=0\) we follow the concavity argument of
\cite[Supplemental C, Case 2]{Cao2026} and bound the lower polynomial at
\((\kappa,\ldots,\kappa)\) by \eqref{eq:lastvertex}
(Remark~\ref{rem:cao-vertex}). For \(n\geq6\), the trureturing project
proved \(h=1,2\) in Lean~\cite{ElonSG2026}. That proof uses the
complementary product as an interpolation weight and an equivalent
factorization in the case \(h=1\). Our interpolation uses the product
estimates of Lemma~\ref{lem:products}, and one estimate covers every
\(h\geq2\).
For \(n=4,5\), all cases are contained in the operator inequality that
\cite[Theorem~2]{Jha2026} states together with the tight bounds.
Section~\ref{sec:results} describes the status of its proof.
All partitions are now covered for \(n\geq4\). The hypothesis \(n\geq4\)
enters through \(d\geq3\) in the case \(h=1\) and through the vertex bound
\eqref{eq:lastvertex}; both estimates fail for \(n=3\), where \(T\geq0\)
nevertheless holds, by \cite[Proposition~2]{Kaniewski2016}
(Section~\ref{sec:extraction}).
Each two-by-two block \(T_H\) has positive trace, by
Lemma~\ref{lem:blocks}, and nonnegative determinant, by the cases above.
Therefore \(T\geq0\), first on the reduced angle cube and then on the
full cube.

\subsection{Return to local extraction}
\label{sec:extraction}

The operator inequality \(T\geq0\) now applies in every joint sector.
The construction can be assembled locally even when the sector labels
are correlated.
Each party may measure its own sector label, rotate or embed the sector into its canonical qubit, apply \eqref{eq:channel}, and discard the label. This local pinching deletes all matrix elements between different joint sectors and leaves the score unchanged. The constants in \(T\geq0\) are independent of every sector label, so averaging over the joint sectors yields the same affine bound.

On a joint sector in canonical coordinates, self-adjointness of the channels gives
\[
 \bra{\Phi_n}
 \left(\bigotimes_j\Lambda_{\theta_j}\right)(\rho)
 \ket{\Phi_n}
 =\tr(\rho K)
 \geq\bar s_n\tr(\rho W_n)+\bar\mu.
\]
Compose with the fixed output unitaries defined after \eqref{eq:target} to obtain the standard GHZ target.
This proves the finite-dimensional lower bound for \(n\geq4\).

For \(n=3\), \eqref{eq:bell} gives
\(W_3=\sum_{x,y\in\{0,1\}}(-1)^{xy}A_{1,x}\otimes A_{2,1-y}\otimes A_{3,x\oplus y}\).
In the same notation for the observables, Kaniewski's Mermin operator
\cite{Kaniewski2016} is
\(\sum_{x,y}(-1)^{xy}A_{1,x}\otimes A_{2,y}\otimes A_{3,x\oplus y}\),
so \(W_3\) is his operator with the two inputs of
party~\(2\) exchanged; this relabeling leaves the state and its
extractability unchanged. His operator also has quantum maximum \(4\),
and he uses the same squared-fidelity convention. He writes each party's
two observables as \(\cos\theta\,X+(-1)^x\sin\theta\,Z\), \(x\in\{0,1\}\),
with an angle \(\theta\in[0,\pi/2]\)
\cite[Eq.~(8)]{Kaniewski2016}; relative to this parametrization, the
sign \((-1)^{j-1}\) in \eqref{eq:canonical} already exchanges the two
inputs of party~\(2\). The two exchanges cancel: \(W_3\) written in the
coordinates \eqref{eq:canonical} is his operator written in his
coordinates, at the same angles, and his target, the state attaining
the maximum at the ideal angles, is \(\ket{\Phi_3}\). His extraction channel
\cite[Eq.~(9)]{Kaniewski2016} is \eqref{eq:channel}, and
\cite[Proposition~2]{Kaniewski2016} states that
\(K\geq\bar s_3W_3+\bar\mu I\) for all angles in \([0,\pi/2]^3\), his
constants being \((2+\sqrt2)/8=\bar s_3\) and \(-1/\sqrt2=\bar\mu\); this
is the inequality \(T\geq0\) on the whole angle
cube. The same assembly therefore gives the bound
for finite-dimensional projective strategies
(cf.\ \cite[Eq.~(12)]{Kaniewski2016}), and
constant local outputs supply its globally valid form below \(\beta_B\).
His Appendix~A sketches the reduction from arbitrary binary observables
to qubit observables and makes binary POVMs projective with local
ancillas; the dilation and the approximation of Appendix~\ref{sec:lift}
do not depend on \(n\) and apply unchanged.

For normal states of arbitrary local dimension, the finite-support
approximation in Appendix~\ref{sec:lift} preserves the score in the
limit and changes any extracted overlap by at most the trace-norm
error. This proves the lower bound of Theorem~\ref{thm:main} on its
full stated domain. The construction using one qutrit in
Section~\ref{sec:dimension} proves attainment; Appendix~\ref{app:flags}
also gives the earlier construction using local classical registers.

\section{Qubit strategies and the attaining mixture}
\label{sec:dimension}

This section proves Theorem~\ref{thm:qubit-profile}, constructs the
mixture on one qutrit and qubits that attains the unrestricted bound,
which completes the proof of Theorem~\ref{thm:main}, and proves
Theorem~\ref{thm:minimal-dimension}.
Recall \(b=\beta/\beta_Q\) and the affine
functions \(L(\beta)\), \(\mathcal L(b)\) from
\eqref{eq:affine-notation}; in the proof notation,
\(L(\beta)=\bar s_n\beta-1/s\) and \(\mathcal L(b)=ab-1/s\).

\subsection{Qubit strategies}
\label{sec:qubit}

The GHZ eigenbasis and the sum of squared eigenvalues used below were
obtained by Scarani and Gisin \cite{ScaraniGisin2001} for every number
of parties; we rederive both from \eqref{eq:Wblock} in our normalization.
They also bound the expectation in terms of the overlap with a top
eigenvector, for every number of qubits, and conclude that an expectation
above the biseparable bound requires overlap larger than \(1/2\) and a
nondegenerate top eigenvalue
\cite[arXiv:quant-ph/0103068v1, Appendix~B]{ScaraniGisin2001}. Bardyn et
al.~\cite{Bardyn2009} combined this overlap bound with the sum of
squares to obtain the tight two-qubit curve. We give the deductions and
the multipartite attaining families explicitly.
The next lemma collects the spectral facts used in this section and in
the appendices. Its last assertion is the case of a vanishing effective
commutator at one party in \cite[Sec.~II.C]{Kaniewski2017}, whose Bell
operator is \(W_n/2\) up to local relabelings.

\begin{lemma}
\label{lem:bell-spectrum}
Let every party hold a qubit with a canonical pair \eqref{eq:canonical},
with angles \(\theta\in[0,\pi/2]^n\), and let \(A_j,b_j\) be given by
\eqref{eq:variables}. For every partition \(H\sqcup D=[n]\), the operator
\(W_n\) has the eigenvalues \(\pm\mu_H\) on the plane spanned by the
corresponding vectors \(\ket{\sigma}_Y,\ket{-\sigma}_Y\), where
\[
 \mu_H\coloneqq\beta_Q\sqrt{B_HA_D+A_HB_D}.
\]
If \(\mu_H>0\), each of its two eigenvectors in that plane is a balanced
superposition of \(\ket{\sigma}_Y\) and \(\ket{-\sigma}_Y\) and is
therefore locally unitarily equivalent to \(\ket{\GHZ_n}\). The sum over
the \(2^{n-1}\) complementary pairs is
\begin{equation}
 \sum_{\{H,D\}}\mu_H^2
 =\beta_Q^2\prod_{i=1}^n(A_i+b_i)=\beta_Q^2.
 \label{eq:spectral-sum}
\end{equation}
If \(\theta_i\in\{0,\pi/2\}\) for some party \(i\), then
\(\mu_H\leq\beta_B\) for every \(H\). Consequently, a finite-dimensional
projective strategy in which the two observables of some party commute
has score at most \(\beta_B\).
\end{lemma}

\begin{proof}
By Appendix~\ref{app:covariance}, \eqref{eq:Wblock} holds on the full
angle cube when the factors \(\sqrt{b_j}\) in \(x,y\) are read as the
signed quantities \((\cos\theta_j-\sin\theta_j)/s\). The block of
\(W_n\) on the plane is therefore \(\beta_Q\) times a Hermitian matrix
with zero diagonal whose off-diagonal entry has modulus
\(\sqrt{x^2+y^2}\), and \(x^2+y^2=B_HA_D+A_HB_D\) for either sign. This
gives the eigenvalues and, when \(\mu_H>0\), eigenvectors whose two
coefficients have equal modulus. The vectors \(\ket{\sigma}_Y\) and
\(\ket{-\sigma}_Y\) are products of local vectors that are orthogonal at
every party, so such a superposition is locally unitarily equivalent to
\(\ket{\GHZ_n}\). Each of the \(2^n\) product atoms of the distribution
with local probabilities \(A_j,b_j\) occurs once in the sum
\eqref{eq:spectral-sum}.

If \(\theta_i\in\{0,\pi/2\}\), then \(A_i=b_i=1/2\), so
\(\mu_H^2/\beta_Q^2\) is one half of the sum of a product atom of the
remaining parties and its complementary atom. These are distinct
outcomes of the product distribution with local weights \(A_j,b_j\),
\(j\ne i\), so their sum is at most one and \(\mu_H\leq\beta_B\).
For the last assertion, decompose every local space into local sectors
as in Section~\ref{sec:bound}, splitting a two-dimensional simultaneous
eigenspace into lines, and pinch. The score becomes an average over
joint sectors, and in each of them the commuting party sits at a
commuting endpoint, so the conditional state, embedded in canonical
qubits, has score at most \(\max_H\mu_H\leq\beta_B\).
\end{proof}

The binary-POVM statements of this section and of
Section~\ref{sec:original} use the following
replacement of the observables, which keeps the state and the local
spaces; the replacement by extreme points is the argument of
\cite[Sec.~III]{WernerWolf2000}.

\begin{lemma}
\label{lem:projective-replacement}
Let a binary-POVM strategy on finite-dimensional local spaces have
state \(\rho\) and score \(\beta\). Its observables can be replaced by
reflections on the same spaces so that the score of \(\rho\) becomes
the largest value \(\beta_*\) attainable there by binary POVMs; thus
\(\beta\leq\beta_*\leq\beta_Q\), and \(\rho\) and \(\Xi_n(\rho)\) are
unchanged. If moreover \(\beta\geq\beta_B\) and
\(\Xi_n(\rho)\leq L(\beta)+\varepsilon\) with \(\varepsilon\geq0\),
then \(t_*\coloneqq(\beta_*-\beta_B)/(\beta_Q-\beta_B)\) satisfies
\(0\leq t_*-t\leq2\varepsilon\) and
\(0\leq\Xi_n(\rho)-L(\beta_*)\leq\varepsilon\); in particular
\(\beta_*=\beta\) if \(\varepsilon=0\).
\end{lemma}

\begin{proof}
The observables of binary POVMs on a finite-dimensional space are
exactly its Hermitian contractions, which form a compact convex set,
and the score is affine in each of the \(2n\) observables. Starting at
a global maximizer, replace the observables one at a time by extreme
maximizers of the corresponding affine functional; each step retains
the global maximum. The extreme Hermitian contractions are reflections:
an eigenvalue in \((-1,1)\) permits opposite perturbations, whereas
eigenvalues in \(\{+1,-1\}\) force both terms of any convex
decomposition to agree. The quantum bound gives \(\beta_*\leq\beta_Q\).
Under the additional hypotheses, \(\beta_B\leq\beta_*\leq\beta_Q\), so
Theorem~\ref{thm:main} gives \(\Xi_n(\rho)\geq L(\beta_*)\). Since
\(L\) is increasing and \(L(\beta_*)-L(\beta)=(t_*-t)/2\), the remaining
inequalities follow from
\(L(\beta)\leq L(\beta_*)\leq\Xi_n(\rho)\leq L(\beta)+\varepsilon\).
\end{proof}

Several upper bounds on extractability below come from the negativity
method of Vidal and Werner \cite{VidalWerner2002}, in the following form.

\begin{lemma}
\label{lem:negativity}
Let \(\rho\) be a state on finite-dimensional local spaces, let \(i\) be
a party, and let \({}^{T_i}\) denote partial transposition at
party~\(i\). Every product \(\Lambda\coloneqq\bigotimes_j\Lambda_j\) of
local channels to qubits satisfies
\begin{equation}
 \tr(G_n\Lambda(\rho))
 \leq\frac12\|(\Lambda(\rho))^{T_i}\|_1
 \leq\frac12\|\rho^{T_i}\|_1.
 \label{eq:partial-transpose-bound}
\end{equation}
Consequently \(\Xi_n(\rho)\leq\frac12\|\rho^{T_i}\|_1\).
\end{lemma}

\begin{proof}
Since \(\|G_n^{T_i}\|_\infty=1/2\) (its nonzero eigenvalues are \(\pm1/2\)), Hölder's inequality gives the first
inequality. For the second, \((\Lambda(\rho))^{T_i}\) is the image of
\(\rho^{T_i}\) under the product map in which \(\Lambda_i\) is replaced
by its conjugation with the transposition. This map has the conjugate
Kraus operators and is CPTP, and the trace norm of a Hermitian operator
contracts under CPTP maps.
\end{proof}

\begin{proof}[Proof of Theorem~\ref{thm:qubit-profile}]
First take noncommuting projective qubit pairs and put them in the
canonical form \eqref{eq:canonical}. By Lemma~\ref{lem:bell-spectrum},
\(W_n\) has the eigenvalue pairs \(\pm\mu_H\), one for each
complementary pair, the eigenvector for every \(\mu_H>0\) is
locally unitarily equivalent to GHZ, and
\(\sum_{\{H,D\}}\mu_H^2=\beta_Q^2\).
Order the nonnegative eigenvalues as
\(\mu_1\geq\mu_2\geq\cdots\), choose a top eigenvector, and let
\(p\) be its overlap with \(\rho\). Then
\[
 \beta\leq p\mu_1+(1-p)\mu_2,\qquad
 \mu_1^2+\mu_2^2\leq\beta_Q^2.
\]
If \(p\leq1/2\), the first expression is at most
\((\mu_1+\mu_2)/2\leq\beta_Q/s\).
Thus a score above \(\beta_B\) forces \(p>1/2\); it also excludes a
degenerate top eigenvalue. Cauchy--Schwarz gives
\[
 b^2\leq p^2+(1-p)^2,\qquad
 \Xi_n(\rho)\geq p\geq F(b).
\]
The inequality \(\Xi_n(\rho)\geq p\) uses the fixed product unitary mapping the
chosen eigenvector to GHZ. At \(b=1/s\), constant local outputs already give the required lower bound \(1/2\).

For binary qubit POVMs, Lemma~\ref{lem:projective-replacement}
replaces the observables by reflections on the same qubits, with the
same state and at least the same score.
A qubit reflection is either \(\pm I\) or \(r_xX+r_yY+r_zZ\) with \(r_x^2+r_y^2+r_z^2=1\).

A commuting pair at any party limits the score to \(\beta_B\), by
Lemma~\ref{lem:bell-spectrum}. Thus, if the original score exceeds
\(\beta_B\), the maximizing reflection strategy has a noncommuting
pair at every party. Its score is at least the original one, its state and
extractability are unchanged, and \(F\) is increasing. The preceding
projective argument proves the POVM lower bound directly.

To attain the bound, abbreviate \(\ket{\Phi_n}\) to \(\ket{\Phi}\) and
take
\begin{equation}
 \rho_p\coloneqq p\proj{\Phi}+(1-p)X_1\proj{\Phi}X_1,\qquad
 \frac12\leq p\leq1.
 \label{eq:qubit-attainer}
\end{equation}
The vectors \(\ket{\Phi}\) and \(X_1\ket{\Phi}\) are orthogonal.
Set \(\theta_j\coloneqq\pi/4\) for \(j\geq2\), and set
\[
 \theta_1\coloneqq\frac{\pi}{4}-\varphi,\qquad
 \tan\varphi=\frac{1-p}{p},\qquad 0\leq\varphi\leq\frac{\pi}{4}.
\]
The two vectors in \eqref{eq:qubit-attainer} are eigenvectors of \(W_n\) with eigenvalues \(\beta_Q\cos\varphi\) and
\(\beta_Q\sin\varphi\), respectively, by \eqref{eq:Wblock}.
Hence
\[
 \frac{\beta}{\beta_Q}
 =p\cos\varphi+(1-p)\sin\varphi=\sqrt{p^2+(1-p)^2}.
\]
The identity channel, followed by the fixed output unitaries, has overlap \(p\).

For the reverse bound, apply partial transposition across
party~\(1\). Across the cut between party~\(1\) and the others,
\(\rho_p\) is supported on a two-qubit subspace, on which it is a mixture
of two orthogonal Bell states.
Its partial-transpose eigenvalues are
\[
 \frac12,\quad \frac12,\quad \frac{2p-1}{2},\quad
 -\frac{2p-1}{2},
\]
with possible additional zeros, so \(\|\rho_p^{T_1}\|_1=2p\), and
Lemma~\ref{lem:negativity} gives \(\Xi_n(\rho_p)\leq p\).
Thus \(\Xi_n(\rho_p)=p\), proving attainment throughout the interval.
\end{proof}

\subsection{The attaining mixture and the local ranks}
\label{sec:attainer}

By \eqref{eq:strict-gap}, attaining the affine line at an interior score
requires a local dimension above two at some party. The next construction modifies the mixture with local classical registers of \cite{Kaniewski2016}: only the first party records the branch, that is, which of the two components of the mixture is present, and every other party uses the same measurement pair in both branches.

\begin{theorem}
\label{thm:qutrit-attainer}
For every \(n\geq3\) and \(\beta\in[\beta_B,\beta_Q]\), the affine
bound \(L(\beta)\) is attained on
\(\mathbb C^3\otimes(\mathbb C^2)^{\otimes(n-1)}\), with one fixed
collection of observables as \(\beta\) varies.
\end{theorem}

\begin{proof}
At party \(1\), use the ideal pair \((X+Z)/s,(X-Z)/s\) on
\(\operatorname{span}\{\ket0,\ket1\}\), and give both observables
the scalar value \(+1\) on \(\ket2\). All other parties use their
ideal canonical pairs. Let \(J\) embed the first qubit into this
qutrit, acting identically elsewhere, and define
\[
 \ket{x+}\coloneqq\frac{\ket0+\ket1}{s},\qquad
 \ket{\psi}\coloneqq\sqrt2\,{}_1\!\braket{x+}{\Phi_n}.
\]
The vector \(\ket{\psi}\) is normalized because the first marginal
of \(\proj{\Phi_n}\) is \(I/2\). For
\(t=(\beta-\beta_B)/(\beta_Q-\beta_B)\), take
\begin{equation}
 \omega_t
 \coloneqq tJ\proj{\Phi_n}J^*
 +(1-t)\proj{2}\otimes\proj{\psi}.
 \label{eq:qutrit-attainer}
\end{equation}
The first branch has score \(\beta_Q\).
At the commuting endpoint \(\theta_1=0\), with all other angles ideal,
both \(\ket{\Phi_n}\) and \(X_1\ket{\Phi_n}\) are eigenvectors of \(W_n\) with eigenvalue \(\beta_B\), and
\[
 \ket{x+}\otimes\ket{\psi}
 =\frac{\ket{\Phi_n}+X_1\ket{\Phi_n}}s.
\]
The first observables are both \(X\) at this endpoint, and therefore
act as \(+1,+1\) on this vector. Replacing that first factor by
\(\ket2\) and its observables by the scalar values \(+1,+1\)
preserves the score. The second branch of
\eqref{eq:qutrit-attainer} consequently has score \(\beta_B\).

The first local channel acts as the identity on the subspace
\(\operatorname{span}\{\ket0,\ket1\}\), maps \(\ket2\) to \(\ket{x+}\), and discards coherences
between these subspaces. This is CPTP. With the identity at every
other party and the fixed output unitaries converting \(\Phi_n\)
to GHZ, the resulting product channel has overlap \(t+(1-t)/2\).
Conversely, for any product channel the first branch contributes
at most one, and the second branch remains separable across the
first-party cut and contributes at most \(1/2\), the largest squared
Schmidt coefficient of \(\ket{\GHZ_n}\) across that cut. Linearity gives
\[
 \Xi_n(\omega_t)=\frac{1+t}{2}=L(\beta).
\]
Every other party remains a qubit and uses the same observables in
both branches.
\end{proof}

The same two branches can carry coherence. The following
calculation, used in Theorem~\ref{thm:minimal-equality-state} to pass
from the pinched state to the original one and in
Proposition~\ref{prop:stability-exponent}, shows that coherence
increases extractability while leaving the score of the fixed
measurements unchanged.

\begin{proposition}
\label{prop:coherent-qutrit}
Let \(n\geq3\), and use the vectors and embedding of
\eqref{eq:qutrit-attainer}. Set
\[
 \ket G\coloneqq J\ket{\Phi_n},\qquad
 \ket B\coloneqq\ket2\otimes\ket\psi.
\]
For \(0\leq t\leq1\) and \(c\in\mathbb C\) with
\(|c|^2\leq t(1-t)\), define
\begin{equation}
 \omega_{t,c}
 \coloneqq t\proj G+(1-t)\proj B
   +c\ket G\!\bra B+\overline c\ket B\!\bra G.
 \label{eq:coherent-qutrit}
\end{equation}
Then \(\omega_{t,0}=\omega_t\), and the fixed observables in
Theorem~\ref{thm:qutrit-attainer} give score
\(t\beta_Q+(1-t)\beta_B\) for every allowed \(c\). The exact extractability of \(\omega_{t,c}\) is
\begin{equation}
 \Xi_n(\omega_{t,c})
 =\frac{1+\sqrt{t^2+2|c|^2}}2.
 \label{eq:coherent-extractability}
\end{equation}
Thus nonzero coherence strictly raises extractability above
\((1+t)/2\).
\end{proposition}

\begin{proof}
The two branch vectors are orthonormal, so the stated condition on
\(c\) is precisely positivity of their two-by-two density matrix.
The Bell operator preserves the first party's two-dimensional and
one-dimensional subspaces. Its cross terms between \(G\) and \(B\)
therefore vanish, proving the score assertion.

To compute the partial transpose, use the first-party basis
\(\ket{e_0}\coloneqq\ket{x+}\),
\(\ket{e_1}\coloneqq\ket{x-}\coloneqq(\ket0-\ket1)/s\), and
\(\ket{e_2}\coloneqq\ket2\). On the remaining parties put
\[
 \ket{f_0}\coloneqq\ket\psi,\qquad
 \ket{f_1}\coloneqq\sqrt2\,{}_1\!\braket{x-}{\Phi_n}.
\]
The marginal of \(\Phi_n\) is \(I/2\), so \(f_0,f_1\) are
orthonormal and
\[
 \ket G=\frac{\ket{e_0f_0}+\ket{e_1f_1}}s,
 \qquad \ket B=\ket{e_2f_0}.
\]
The basis on the remaining parties serves only to write the matrix.
Write \(d\coloneqq\sqrt{t^2+2|c|^2}\). Up to a permutation of this
six-dimensional basis, \(\omega_{t,c}^{T_1}\) has blocks
\[
 \begin{pmatrix}
 t/2&\overline c/s\\ c/s&1-t
 \end{pmatrix},\qquad
 \begin{pmatrix}
 0&t/2&0\\ t/2&0&\overline c/s\\ 0&c/s&0
 \end{pmatrix},\qquad (t/2).
\]
The first block is positive, since its determinant is
\([t(1-t)-|c|^2]/2\), and its trace is \(1-t/2\).
The second block has eigenvalues \(0,d/2,-d/2\). Consequently
\[
 \|\omega_{t,c}^{T_1}\|_1=1+d.
\]
Lemma~\ref{lem:negativity} therefore bounds every allowed
product-channel overlap by \((1+d)/2\).

For attainment when \(d>0\), put
\(\alpha\coloneqq t/d\) and \(\zeta\coloneqq s\overline c/d\); here \(\alpha\) and \(\zeta\) are channel parameters, and \(\zeta\) is unrelated to the phase \(\zeta\) of Lemma~\ref{lem:blocks}. The first party
uses the channel from the qutrit to the output qubit with Kraus
operators
\begin{equation}
 M_0\coloneqq\ket{e_0}\!\bra{e_0}
       +\alpha\ket{e_1}\!\bra{e_1},\qquad
 M_1\coloneqq\zeta\ket{e_1}\!\bra{e_1}
       +\ket{e_0}\!\bra{e_2}.
 \label{eq:coherent-channel}
\end{equation}
Here the output vectors \(e_0,e_1\) mean \(x+,x-\) in the output
qubit. Since \(\alpha^2+|\zeta|^2=1\) and the two terms of
\(M_1\) have orthogonal output vectors,
\(M_0^*M_0+M_1^*M_1=I_3\).
Every other party uses the identity channel, followed by the fixed
product of output unitaries converting \(\Phi_n\) to GHZ.
Before these final unitaries, the amplitudes against \(\Phi_n\)
from the two branches are
\[
 \frac{1+\alpha}{2},\quad\frac{\zeta}{2}
 \quad\text{from }G,\qquad
 0,\quad\frac1s\quad\text{from }B,
\]
in the order \(M_0,M_1\). The resulting overlap is
\[
 \frac{t(1+\alpha)}2+\frac{1-t}{2}
       +\frac{\operatorname{Re}(c\zeta)}s
 =\frac12+\frac{t^2+2|c|^2}{2d}
 =\frac{1+d}{2}.
\]
If \(d=0\), then \(t=c=0\); setting \(\alpha\coloneqq1\) and
\(\zeta\coloneqq0\) in the same Kraus operators attains \(1/2\).
\end{proof}

Within this family, the extractability excess gives an exact
distance from the incoherent mixture. With \(\delta\coloneqq\Xi_n(\omega_{t,c})-(1+t)/2\), we have
\(\sqrt{t^2+2|c|^2}=t+2\delta\), and hence
\begin{equation}
 |c|^2=2t\delta+2\delta^2,\qquad
 \|\omega_{t,c}-\omega_{t,0}\|_1
   =2\sqrt{2t\delta+2\delta^2}.
 \label{eq:coherent-distance}
\end{equation}
Indeed, the difference has eigenvalues \(|c|,-|c|\) on the span
of \(G,B\). In particular, \(\delta\leq\varepsilon\) implies
\(\|\omega_{t,c}-\omega_{t,0}\|_1
\leq2\sqrt{2\varepsilon(t+\varepsilon)}\).

With the qubit optimum and the construction \eqref{eq:qutrit-attainer}
available, the bound on the local ranks follows.

\begin{proof}[Proof of Theorem~\ref{thm:minimal-dimension}]
By \eqref{eq:strict-gap}, \(F(b)>\mathcal L(b)\).
A rank-one local marginal makes the state product across that cut,
so \(\Xi_n(\rho)\leq1/2<\mathcal L(b)\), a contradiction.
If all local ranks were two, then \(\rho\) would lie in the tensor
product of the local supports, since each local complementary projection
has zero expectation and therefore annihilates \(\rho\); compressing each
POVM effect to its local support would then give an all-qubit strategy
with unchanged score and extractability. Extractability is unchanged because
channels restrict to the support, and every channel on the support
extends by a fixed output on its orthogonal complement.
Theorem~\ref{thm:qubit-profile} then contradicts
\(\Xi_n(\rho)=\mathcal L(b)\).
This proves \eqref{eq:minimal-dimension}.
For \(0<t<1\), the first marginal of \(\omega_t\) has a positive
\(tI_2/2\) block and a positive \((1-t)\proj2\) block.
Every other marginal has a positive \(tI_2/2\) contribution.
The local ranks of \(\omega_t\) are therefore exactly \(3,2,\ldots,2\).
\end{proof}

\section{Near-equality and equality after local pinching}
\label{sec:stability}

The affine bound is attained by combining an ideal strategy with
a strategy having one deterministic party. We now ask what structure
is forced when actual extractability is at or close to this minimum.
Recall the operator \(T\) of Section~\ref{sec:channels}, whose positivity
on every joint sector gave the lower bound, and the variables \(v_i\) of
\eqref{eq:variables}, which vanish at the ideal angle \(\pi/4\) and equal
\(\kappa\) at the commuting endpoints.
The results in this section assume finite-dimensional projective strategies and \(n\geq4\). The restriction on \(n\) is inherited from the estimates that Theorem~\ref{thm:uniform-gap} strengthens (see the end of Section~\ref{sec:positivity}); the results of this section are not established for \(n=3\). For binary POVMs, they hold for the Naimark dilation
of Appendix~\ref{sec:lift}, which is given by an explicit formula in the
observables and enlarges each local space. The binary-POVM statements of
Sections~\ref{sec:dimension} and~\ref{sec:original} hold instead on the
given spaces: their proofs keep the state and the local spaces and
replace only the observables, by score-maximizing reflections
(Lemma~\ref{lem:projective-replacement}).

\subsection{A spectral lower bound and its consequence for extractability}

For canonical angles \(\theta\in[0,\pi/2]^n\), set
\begin{equation}
 x_i\coloneqq\frac{v_i}{\kappa},\qquad
 R(x)\coloneqq\sum_i x_i-\max_i x_i.
 \label{eq:axis-distance}
\end{equation}
Thus \(R\) is the \(\ell^1\) distance to the union of the coordinate
axes in \([0,1]^n\); these coordinates are unrelated to the block entries \(x,y\) of Lemma~\ref{lem:blocks}. Let \(N_i\) be auxiliary independent Bernoulli
variables with probabilities \(x_i\), and put
\begin{equation}
 \mathsf P_2(x)\coloneqq\Pr\!\left\{\sum_i N_i\geq2\right\},\qquad
 c_*\coloneqq\frac{31-21s}{128k}=\frac{52s-73}{128}\approx0.0042>0.
 \label{eq:uniform-constant}
\end{equation}
Like \(R\), the probability \(\mathsf P_2(x)\) vanishes exactly on the
axes, where \(T\ket{\Phi_n}=0\) (Proposition~\ref{prop:kernel}); a
lower bound on \(T\) must therefore vanish there.
The auxiliary probability is a polynomial in the normalized angle parameters \(x_i\). The physical probabilities of the joint sectors enter only when the operator estimate is averaged over sectors.

\begin{theorem}
\label{thm:uniform-gap}
For every \(n\geq4\) and all canonical angles \(\theta\in[0,\pi/2]^n\),
\begin{equation}
 T\geq c_*\mathsf P_2(x)I
 \geq c_*\left(\frac{R(x)}{R(x)+2}\right)^2 I.
 \label{eq:uniform-gap}
\end{equation}
With \(x_{(2)}\) the second largest coordinate, also
\(T\geq c_*x_{(2)}^2I\).
The constant is independent of \(n\).
\end{theorem}

The proof is given at the end of this subsection.
Here \(x_{(2)}\) is the \(\ell^\infty\) distance from \(x\) to the union of the axes, just as \(R\) is the \(\ell^1\) distance.
The lower bound on \(T\) controls extractability because the overlap
of the constructed extraction is bounded above by \(\Xi_n\). Thus a
small excess \(\Xi_n(\rho)-L(\beta)\) forces the measurement angles to concentrate near the axes, where at most one angle differs from \(\pi/4\).

\begin{theorem}
\label{thm:actual-stability}
Let \(n\geq4\), and let a finite-dimensional projective strategy
have \(\beta\in[\beta_B,\beta_Q]\) and
\[
 0\leq\Xi_n(\rho)-L(\beta)\leq\varepsilon.
\]
Locally pinch into local sectors. Write \(\ell\) for the joint
sector label, \(\mathsf p_\ell\) for its probability, and \(x^{(\ell)}\)
for its normalized angle parameters, using the canonical
embedding for one-dimensional sectors. Put
\(R_\ell\coloneqq R(x^{(\ell)})\). Then
\begin{align}
 \sum_\ell\mathsf p_\ell\,\mathsf P_2(x^{(\ell)})
 &\leq\frac{\varepsilon}{c_*}, \label{eq:average-stability}\\
 \sum_\ell\mathsf p_\ell
       \left(\frac{R_\ell}{R_\ell+2}\right)^2
 &\leq\frac{\varepsilon}{c_*}, \nonumber\\
 \Pr_{\mathsf p}\{R_\ell\geq r_0\}
 &\leq\min\left\{1,\,
       \frac{\varepsilon}
       {c_*(r_0/(r_0+2))^2}\right\}
       \quad(r_0>0). \label{eq:tail-stability}
\end{align}
The same statements hold for finite-dimensional binary POVMs in
the Naimark dilation of Appendix~\ref{sec:lift}; the angle distribution is then defined by that dilation and a Jordan decomposition of the dilated observables.
\end{theorem}

\begin{proof}
Let \(\rho_\ell\) be the normalized conditional state in canonical
coordinates. The explicit local extraction of Section~\ref{sec:bound}
has overlap \(F_{\mathrm{ext}}\) satisfying
\[
 0\leq F_{\mathrm{ext}}-L(\beta)
 =\sum_\ell\mathsf p_\ell\tr(\rho_\ell T_\ell)
 \leq\Xi_n(\rho)-L(\beta)\leq\varepsilon.
\]
This uses one channel per local sector label. Theorem~\ref{thm:uniform-gap} gives
\eqref{eq:average-stability}, and Markov's inequality gives
\eqref{eq:tail-stability}. The physical distribution \(\mathsf p\) can
be arbitrarily correlated.

For a binary POVM, use the fixed local embedding into the Naimark dilation of Appendix~\ref{sec:lift}. Extractability is invariant under this
embedding: channels on the enlarged space restrict by composing
with the embedding, while channels on the original space extend
by any fixed output off its embedded support. The score is also
unchanged. Applying the projective statement proves the claimed
dilated version.
\end{proof}

To express Theorem~\ref{thm:actual-stability} in terms of measurement
angles, consider a strategy consisting of one joint qubit sector. By Theorem~\ref{thm:qubit-profile} and \eqref{eq:strict-gap}, its excess is at least \(F(b)-\mathcal L(b)\), which is positive at interior scores, so the condition \(\varepsilon<c_*\) below can hold only near the endpoint scores. Whenever \(\varepsilon<c_*\), Theorem~\ref{thm:actual-stability} gives
\begin{equation}
 R\leq\frac{2\sqrt{\varepsilon/c_*}}
              {1-\sqrt{\varepsilon/c_*}}.
 \label{eq:single-block-stability}
\end{equation}
In terms of angles, set
\(\varphi_i\coloneqq|\theta_i-\pi/4|\), and choose \(i_*\) maximizing
\(\varphi_i\). Since \(v_i=1-\cos\varphi_i\) and
\[
 \frac{1-\cos\varphi}{\kappa}
 \geq\frac{16\varphi^2}{\pi^2}
 \qquad(0\leq\varphi\leq\pi/4),
\]
we have
\[
 \sum_{i\ne i_*}\varphi_i^2\leq\frac{\pi^2}{16}R.
\]
The scalar inequality follows from the decrease of
\(\sin u/u\) on \([0,\pi/8]\), with equality at the upper endpoint.
Thus the Euclidean angle distance to the union of axes is
\(O(\varepsilon^{1/4})\) in a single sector, with constants
independent of \(n\). Since \(\sum_{i\ne i_*}\varphi_i^2\leq\pi^2R/16\) in every joint sector, \eqref{eq:tail-stability} with \(r_0=16\varphi_0^2/\pi^2\) also shows that, for \(0<\varphi_0\leq1\), so that \(r_0<2\), the joint sectors of Theorem~\ref{thm:actual-stability} whose Euclidean angle distance from the union of axes is at least \(\varphi_0\) have total probability at most \(\pi^4\varepsilon/(16c_*\varphi_0^4)\). These include every joint sector
with \(\varphi_i\geq\varphi_0\) at two or more parties, that is, with the
angle \(2\theta_i\) between the two observables of \eqref{eq:canonical} at
least \(2\varphi_0\) away from a right angle there (commuting observables
have angle \(0\) or \(\pi\)). Appendix~\ref{sec:axis-formulas} shows that this fourth-root angular scaling
cannot be improved for spectral bounds on \(T\).
For the strategies of Theorem~\ref{thm:original-state-stability},
Appendix~\ref{app:original-state-stability} improves the fourth-root
rate to the order \(\sqrt\varepsilon\).

\begin{proof}[Proof of Theorem~\ref{thm:uniform-gap}]
We strengthen the determinant bounds from
Section~\ref{sec:bound} and Appendix~\ref{app:positivity}.
Complement symmetry allows \(h\leq d\). Put
\[
 \alpha\coloneqq8s-\frac{179}{16},\qquad
 \gamma\coloneqq\frac{30-15s}{56},\qquad
 e_*\coloneqq\frac{31-21s}{64}.
\]
For \(h=0\), the polynomial \(\mathcal P\) in
\eqref{eq:lowerpolynomial} is separately concave.
Iterating its one-variable chord inequalities bounds it below
by the auxiliary Bernoulli average of its vertex values.
The vertices with zero or one nonzero coordinate have value zero.
With \(2\leq j<n\) nonzero coordinates, their values are
\[
 1-(3+s)2^{-j}+s k\,2^{-3j/2}\geq1-(3+s)2^{-2}+s k\,2^{-3}=\frac{4-s}{8}>\alpha.
\]
The vertex \((\kappa,\ldots,\kappa)\) has value at least \(\alpha>0\), by
\eqref{eq:lastvertex}. Hence
\[
 \lambda_{\varnothing,[n]}\geq\mathcal P
 \geq\alpha\mathsf P_2(x).
\]

For \(h=1\), the proof of \eqref{eq:f1} and
\eqref{eq:interpolation} gives
\(\lambda_{H,D}\geq\gamma(1-A_D)\), where \(\gamma\) is the smaller endpoint
value of the affine lower bound for \(g_{1,d}\) in
Section~\ref{sec:positivity} and the term \(A_Df_1\geq0\) is dropped.
Since \(b_i=\kappa x_i(2-\kappa x_i)\geq x_i/2\),
\[
 1-A_D\geq1-\prod_{i\in D}(1-x_i/2)
 \geq\frac12\left(1-\prod_{i\in D}(1-x_i)\right)
 \geq\frac12\mathsf P_2(x).
\]
For the middle inequality, retain each success in \(D\) independently
with probability \(1/2\): conditional on at least one success, the
probability that some success survives is at least \(1/2\).
The last inequality holds because two successes overall require at
least one in \(D\), as \(H\) is a singleton.

For \(h\geq2\), Appendix~\ref{app:two-indices} gives
\(\lambda_{H,D}\geq e_*\) for every \(d\geq h\).
Since \(\alpha,\gamma/2>e_*\) and \(\mathsf P_2\leq1\),
every block satisfies
\[
 \lambda_{H,D}\geq e_*\mathsf P_2(x).
\]
The quantities \(w,z\) in Lemma~\ref{lem:blocks} are complementary
product atoms, so \(w+z\leq1\) and \(\tr T_H\leq1+s=k\).
For a positive semidefinite two-by-two matrix,
\(\lambda_{\min}\geq\det/\tr\). Consequently
\[
 \lambda_{\min}(T_H)
 \geq\frac{e_*}{2k}\mathsf P_2(x)=c_*\mathsf P_2(x).
\]

It remains to compare \(\mathsf P_2\) and \(R\).
Greedily place each \(x_i\) in the currently lighter of two groups.
Their final load difference is at most \(\max_i x_i\), so both
loads are at least \(R/2\). The probability of a success in a
group of load \(u\) is at least \(1-e^{-u}\).
Independence between the two auxiliary groups therefore gives
\[
 \mathsf P_2(x)\geq(1-e^{-R/2})^2
 \geq\left(\frac{R}{R+2}\right)^2,
\]
where \(e^t\geq1+t\) proves the last inequality.
The event that the two largest coordinates both succeed also
gives \(\mathsf P_2(x)\geq x_{(1)}x_{(2)}\geq x_{(2)}^2\).
Covariance under angle reflections extends the conclusion to all canonical
angles.
\end{proof}

\subsection{Equality after local pinching}
\label{sec:equality}

At exact equality, every joint sector of positive probability satisfies
\(\tr(\rho_\ell T_\ell)=0\). We first determine the possible supports
and then improve the local channels to exclude interior-axis sectors, in which one angle lies strictly between \(\pi/4\) and a commuting endpoint and every other angle is ideal.

\begin{proposition}
\label{prop:kernel}
Let \(n\geq4\) and let \(T=T(\theta)\) be the operator of
Section~\ref{sec:channels} for canonical angles \(\theta\in[0,\pi/2]^n\).
If at least two \(v_i\) are positive, then \(T\) is positive definite.
If all \(v_i=0\), or exactly one \(v_i\) is nonzero and it lies in
\((0,\kappa)\), then
\[
 \ker T=\operatorname{span}\{\ket{\Phi_n}\}.
\]
If exactly one coordinate satisfies \(v_i=\kappa\) and the others
are zero, then
\begin{equation}
 \ker T=\operatorname{span}
       \{\ket{\Phi_n},\Gamma_i\ket{\Phi_n}\},
 \quad
 \Gamma_i\coloneqq
 \begin{cases}
 X_i,&\theta_i=0,\\
 Z_i,&\theta_i=\pi/2.
 \end{cases}
 \label{eq:endpoint-kernel}
\end{equation}
\end{proposition}

\begin{proof}
Off the axes, \(\mathsf P_2>0\), so the first claim follows from
Theorem~\ref{thm:uniform-gap}. Work on a reduced axis, with only
\(v_i=v\) potentially nonzero. In the \(h=0\) block, the entries of Lemma~\ref{lem:blocks} are
\(w=q=1-av,z=0,x=1-v,y=0\). Its diagonal is \(q/2+1/s\), and its
upper off-diagonal is
\[
 e^{-\mathrm i\delta_n}
       \left(\frac q2-a(1-v)\right)
 =-e^{-\mathrm i\delta_n}\left(\frac q2+\frac1s\right).
\]
By \eqref{eq:target}, it has exactly the one-dimensional kernel
\(\operatorname{span}\{\ket{\Phi_n}\}\).
For \(h=1\), \eqref{eq:interpolation} with \(g_{1,d}>0\) gives a strictly positive determinant
unless \(A_D=1\) and \(f_1(v_H)=0\).
Equation \eqref{eq:f1} shows that this occurs exactly for
\(H=\{i\}\), \(v_i=\kappa\), and all \(v_j=0\) in \(D\).
Its kernel vector is \(X_i\ket{\Phi_n}\), by
\eqref{eq:Kblock}--\eqref{eq:Wblock}. All blocks with \(h\geq2\)
have a strictly positive determinant by the bound of Appendix~\ref{app:two-indices} used
above. The covariance under angle reflections carries the endpoint vector to
\(Z_i\ket{\Phi_n}\), up to phase, at \(\theta_i=\pi/2\), while
preserving \(\ket{\Phi_n}\). This proves the classification.
\end{proof}

For any state, \(\tr(\rho T)=0\) is equivalent to its support being
contained in \(\ker T\). This follows from \(T\geq0\) and
\(\tr(T^{1/2}\rho T^{1/2})=0\).

\begin{theorem}
\label{thm:necessary-equality}
Let \(n\geq4\) and consider a finite-dimensional projective strategy
of score \(\beta\in[\beta_B,\beta_Q]\) satisfying
\(\Xi_n(\rho)=L(\beta)\).
After local pinching into local sectors and canonical embedding, every joint
sector of positive probability has one of the following forms:
\begin{enumerate}
 \item All angles are ideal, and its state is \(\proj{\Phi_n}\).
       Its score is \(\beta_Q\).
 \item Exactly one angle is a commuting endpoint \(0\) or \(\pi/2\),
       every other angle is ideal, and its state is supported on
       the two-dimensional space in \eqref{eq:endpoint-kernel}.
       Its score is \(\beta_B\).
\end{enumerate}
Their total probabilities are respectively
\(t=(\beta-\beta_B)/(\beta_Q-\beta_B)\) and \(1-t\).
\end{theorem}

\begin{proof}
The extraction used in Theorem~\ref{thm:actual-stability} must satisfy
\(F_{\mathrm{ext}}=L(\beta)\). Every sector of positive probability
is therefore supported on \(\ker T_\ell\).
By Proposition~\ref{prop:kernel}, every such sector lies on an axis, and its state is the pure target except at the commuting endpoints.

Now modify the extraction locally: at every local angle
\(0<\theta<\pi/2\), use the identity channel, and at the two
endpoints keep the original dephasing channel.
This choice depends only on the party's own local sector.
In an interior-axis sector, the conditional state is
\(\proj{\Phi_n}\), and the old overlap is \(q_i=1-av_i\);
the modified overlap is one. The gain is \(av_i>0\).
In an all-ideal sector nothing changes. In an endpoint sector
the chosen channels coincide with the original ones, and the
overlap remains \(1/2\). Thus any positive weight on interior-axis
sectors would strictly improve the extraction beyond \(L(\beta)\),
a contradiction. Only the two listed types remain.

For the first type, \eqref{eq:Wblock} gives score \(\beta_Q\).
For the second, both spanning vectors are eigenvectors of \(W_n\) with eigenvalue \(\beta_B\), so every supported state has that score.
Averaging the scores fixes the two total probabilities.
\end{proof}

Theorem~\ref{thm:necessary-equality} gives only a necessary condition:
support on the two-dimensional space \eqref{eq:endpoint-kernel} does not
by itself give \(\Xi_n=L(\beta)\), since \(\ket{\Phi_n}\) lies in that
space, with score \(\beta_B\) and extractability \(1\).
Proposition~\ref{prop:endpoint-extractability} gives the exact
extractability of every state on that space.

A more concrete conclusion is available after refining the local
Jordan decomposition into irreducible sectors: every commuting
two-dimensional sector is split into its common one-dimensional
eigenspaces. This refinement yields a pure-state description in each sector of
the locally pinched state.

\begin{theorem}
\label{thm:refined-equality}
Under the hypotheses of Theorem~\ref{thm:necessary-equality}, use
the refined decomposition just described.
In the pinched state, the conditional state of each joint sector of positive probability is pure and is one of the following:
\begin{enumerate}
 \item A locally embedded \(n\)-party GHZ state, with every local
       sector two-dimensional and ideal; the sector has score \(\beta_Q\).
 \item A unit vector of a one-dimensional local sector at one party, the deterministic party, tensored with a locally embedded \((n-1)\)-party GHZ state, with the other
       local sectors two-dimensional and ideal; the sector has
       score \(\beta_B\).
\end{enumerate}
The two classes have total weights \(t\) and \(1-t\), as above;
the deterministic party can vary between joint sectors.

Conversely, a strategy that is already block diagonal in local sectors, with every joint sector of positive probability of one of these two forms and with these scores, has \(\Xi_n=L(\beta)\).
For finite-dimensional binary POVMs, the corresponding necessary statement holds for the Naimark dilation of Appendix~\ref{sec:lift}, after refinement and local pinching.
\end{theorem}

\begin{proof}
Local pinching preserves the score. The pinched state has extractability at most that of the original state, because any extraction after
pinching is also an allowed local extraction before pinching.
The main lower bound thus forces equality for the pinched state.
Theorem~\ref{thm:necessary-equality} applies to its conditional
sectors. An ideal noncommuting pair needs a two-dimensional
irreducible sector, so the first class has the asserted form.

In an endpoint sector let the scalar observables at party \(i\)
be \(a_0,a_1\in\{+1,-1\}\).
For \(a_0=a_1\), embed the scalar sector in the
\(X\)-eigenspace of eigenvalue \(e\coloneqq a_0\), using \(\theta_i\coloneqq0\).
For \(a_1=-a_0\), use \(\theta_i\coloneqq\pi/2\) and the
\(Z\)-eigenspace of eigenvalue
\(e\coloneqq(-1)^{i-1}a_0\), since the canonical first observable is
\((-1)^{i-1}Z\).

Let \(\ket{g_e}\) be the indicated eigenvector of \(\Gamma_i\).
The intersection of the kernel \eqref{eq:endpoint-kernel} with
the physical local eigenspace is one-dimensional, spanned by
\begin{equation}
 \ket{\psi_e}
 \coloneqq\frac{\ket{\Phi_n}+e\Gamma_i\ket{\Phi_n}}s
 =\ket{g_e}_i\otimes
     \sqrt2\,{}_i\!\braket{g_e}{\Phi_n}.
 \label{eq:refined-endpoint}
\end{equation}
The marginal of \(\Phi_n\) at party \(i\) is \(I/2\), which proves normalization.
An \(X\)- or \(Z\)-eigenvector has inner products of equal modulus with
the two \(Y\)-eigenvectors in \eqref{eq:target}.
The remaining factor in \eqref{eq:refined-endpoint} is therefore
a balanced superposition of two locally orthogonal product
strings, hence is locally equivalent to \(\GHZ_{n-1}\).
This proves necessity and purity in the second class.

For the converse, take a block-diagonal strategy whose joint sectors of positive probability have
these forms and scores.
The sectors of the first type have score \(\beta_Q\) and extractability one.
Each sector of the second type has score \(\beta_B\) and is separable across
its deterministic-party cut. Product channels keep it separable across
that cut, and the GHZ state has two equal Schmidt coefficients across
every cut, so every product extraction gives
overlap at most \(1/2\) on that sector.
For any one fixed product extraction, linearity bounds the total
overlap by \(t+(1-t)/2\), even if the cuts vary.
Taking the supremum preserves this upper bound, while
Theorem~\ref{thm:main} gives the reverse inequality.
This proves equality.
The dilation argument from Theorem~\ref{thm:actual-stability}
proves the last assertion.
\end{proof}

\section{The original state: uniqueness and stability}
\label{sec:original}

The results of Section~\ref{sec:stability} describe the state after local
pinching. This section concerns the original state: its uniqueness at
the minimum product of local ranks
(Theorem~\ref{thm:minimal-equality-state}) and its stability on one
qutrit and qubits (Theorem~\ref{thm:original-state-stability} and
Propositions~\ref{prop:high-score-stability}--\ref{prop:stability-exponent}).

\subsection{Uniqueness at the minimum product of local ranks}
\label{sec:uniqueness}

Passing from the pinched state to the original one requires the coherence calculation of Proposition~\ref{prop:coherent-qutrit} and one system of local coordinates for both branches; the following lemma provides the coordinates, and Theorem~\ref{thm:original-state-stability} uses it again.

\begin{lemma}
\label{lem:branch-alignment}
Let \(n\geq3\), let party \(i\) have local space \(\mathbb C^3\) and
every other party \(\mathbb C^2\), and let \(J'\) embed the qubit of
party \(i\) isometrically into \(\mathbb C^3\), orthogonally to a unit
vector \(\ket{e_2}\). If \(\ket g\in\mathbb C^2\) has inner products of modulus
\(1/\sqrt2\) with both \(Y\)-eigenvectors, as every \(X\)- or
\(Z\)-eigenvector does, then local unitaries, followed by the
reordering of the parties that puts party \(i\) first, map
\(J'\ket{\Phi_n}\) to \(\ket G\) and
\(\ket{e_2}_i\otimes\sqrt2\,{}_i\!\braket{g}{\Phi_n}\) to \(\ket B\),
the vectors of Proposition~\ref{prop:coherent-qutrit}.
\end{lemma}

\begin{proof}
By \eqref{eq:target}, the local Schmidt vectors of \(\Phi_n\) are the
\(Y\)-eigenvectors \(\ket{\pm\tau_j}_Y\). Local unitaries sending these
vectors (at party \(i\), their images under \(J'\), with suitable
phases) to \(\ket0,\ket1\), and \(\ket{e_2}\) to \(\ket2\), map
\(J'\ket{\Phi_n}\) to
\((\ket{0\cdots0}+\ket{1\cdots1})/\sqrt2\). Since \(\ket g\) has
inner products of equal modulus with \(\ket{\pm\tau_i}_Y\), they map the
second vector to
\(\ket2_i\otimes(\ket{0\cdots0}+e^{\mathrm i\vartheta}\ket{1\cdots1})/\sqrt2\)
up to an overall phase. Then
\(\operatorname{diag}(1,e^{\mathrm i\vartheta})\) on
\(\operatorname{span}\{\ket0,\ket1\}\) at party \(i\),
\(\operatorname{diag}(1,e^{-\mathrm i\vartheta})\) at one other party
and a phase on \(\ket2\) fix the first vector and remove both phases
of the second. The two resulting vectors keep their form when party
\(i\) is moved to the first position. The pair \(G,B\) is the case
\(i=1\), \(J'=J\), \(\ket{e_2}=\ket2\), \(\ket g=\ket{x+}\), so composing
with the inverse steps for that case proves the lemma.
\end{proof}

\begin{theorem}
\label{thm:minimal-equality-state}
Let \(n\geq4\), and let a binary-POVM strategy of score
\(\beta_B<\beta<\beta_Q\) satisfy \(\Xi_n(\rho)=L(\beta)\).
Suppose its local ranks satisfy
\(\prod_i r_i=3\cdot2^{n-1}\).
After a permutation of the parties and local isometries identifying
their supports with \(\mathbb C^3,\mathbb C^2,\ldots,\mathbb C^2\),
the original state is exactly \(\omega_t\) in
\eqref{eq:qutrit-attainer}, where
\(t=(\beta-\beta_B)/(\beta_Q-\beta_B)\).
Conversely, this state attains equality with those local ranks.
\end{theorem}

\begin{proof}
The finite product makes every local support finite-dimensional.
As in the proof of Theorem~\ref{thm:minimal-dimension}, \(\rho\) lies in
the tensor product of these supports. Compressing the effects to the
supports preserves the score and extractability, by the restriction
and extension argument in the proof of
Theorem~\ref{thm:minimal-dimension}.
That theorem implies that every local rank is at least two and that at
least one exceeds two. Equality in their product then forces the
local ranks to be \(3,2,\ldots,2\), after a party permutation.

For this fixed state on the same supports,
Lemma~\ref{lem:projective-replacement} with \(\varepsilon=0\) gives
reflections on these supports with the same score \(\beta\).

Apply Theorem~\ref{thm:refined-equality} to these projective
observables. Since \(0<t<1\), both of its classes have positive
weight. Every qubit support must contain an ideal two-dimensional irreducible sector, which exhausts that support. The qutrit support contains an ideal two-dimensional sector and one common one-dimensional eigenspace. There are consequently exactly two
joint sectors, and the deterministic party in the second sector
is the qutrit party. Their conditional states are pure, with
weights \(t\) and \(1-t\).
Every vector orthogonal to those pure states inside either sector has zero expectation in \(\rho\); since \(\rho\geq0\), it lies in the kernel of \(\rho\), so all matrix elements of \(\rho\) involving such a vector vanish. The original state is therefore supported on
the span of the two branch vectors, with a single complex
off-diagonal entry \(c\) satisfying \(|c|^2\leq t(1-t)\).

It remains to choose one system of local coordinates for both
branches. In the canonical coordinates of the two sectors, the ideal
branch is the embedded \(\ket{\Phi_n}\), and by
\eqref{eq:refined-endpoint} the other branch is a unit vector of the
scalar sector at the qutrit party \(i\) tensored with
\(\sqrt2\,{}_i\!\braket{g_e}{\Phi_n}\), where \(\ket{g_e}\) is an
\(X\)- or \(Z\)-eigenvector of the canonical qubit in which that scalar
sector is embedded. The two joint sectors differ only at party \(i\), so these coordinates agree at the qubit parties. Lemma~\ref{lem:branch-alignment}, applied with \(\ket g=\ket{g_e}\) and with \(\ket{e_2}\) a unit vector of the scalar sector, maps
these branches to \(G,B\) of \eqref{eq:coherent-qutrit} by local
unitaries and a reordering of the parties, which do not change
\(\Xi_n\), and so turns \(\rho\) into \(\omega_{t,c}\).

Proposition~\ref{prop:coherent-qutrit} now gives
\[
 \Xi_n(\rho)=\frac{1+\sqrt{t^2+2|c|^2}}2.
\]
Equality with \(L(\beta)=(1+t)/2\) forces \(c=0\).
The converse and the stated local ranks follow from
Theorems~\ref{thm:qutrit-attainer} and~\ref{thm:minimal-dimension}.
\end{proof}

In particular, the states classified in
Theorem~\ref{thm:minimal-equality-state} are mixed states of rank two,
with nonzero eigenvalues \(t\) and \(1-t\).

\subsection{Stability}
\label{sec:original-stability}

The equality classification has a quantitative form on the same
local dimensions.
For \(n\geq3\), a state \(\rho\) on
\(\mathbb C^3\otimes(\mathbb C^2)^{\otimes(n-1)}\), and
\(0\leq t\leq1\), define
\begin{equation}
 D_t(\rho)\coloneqq
 \inf_{\substack{U_1\in\mathrm U(3)\\
                 U_j\in\mathrm U(2),\ 2\leq j\leq n}}
 \left\|U\rho U^*-\omega_t\right\|_1,
 \qquad U\coloneqq\bigotimes_{j=1}^n U_j.
 \label{eq:local-unitary-distance}
\end{equation}
Here \(\omega_t\) is the state in \eqref{eq:qutrit-attainer}.

\begin{theorem}
\label{thm:original-state-stability}
For every \(0<\eta<1/2\), there are constants
\(C_\eta,\varepsilon_\eta>0\), independent of \(n\), with the
following property.
Let \(n\geq4\), and let a binary-POVM strategy on
\(\mathbb C^3\otimes(\mathbb C^2)^{\otimes(n-1)}\) have state
\(\rho\) and score \(\beta\). If
\[
 \eta\leq t=
 \frac{\beta-\beta_B}{\beta_Q-\beta_B}\leq1-\eta,
 \qquad
 0\leq\Xi_n(\rho)-L(\beta)
 \leq\varepsilon\leq\varepsilon_\eta,
\]
then
\begin{equation}
 D_t(\rho)\leq C_\eta\sqrt{\varepsilon}.
 \label{eq:original-state-stability}
\end{equation}
\end{theorem}

The proof is given in Appendix~\ref{app:original-state-stability}.
The strict gap between the qubit curve and the affine line forces
a positive weight on the scalar sector at the qutrit party.
In that sector the qutrit party sits at a commuting endpoint (the canonical embedding of scalar sectors in Section~\ref{sec:bound}), so \(x_1=1\) and \(\mathsf P_2\) is of first order in the remaining \(x_j\). With this weight, the spectral bound of Theorem~\ref{thm:uniform-gap} becomes a bound on their sum that is linear in \(\varepsilon\). Without a coordinate equal to \(1\), \(\mathsf P_2\) can be of second order; this is why the single-sector angle bound after \eqref{eq:single-block-stability} has only fourth-root order. Controlling this sum
bounds the weight outside the two attaining branches uniformly
in \(n\). A change to the first party's extraction channel then
improves the overlap directly in the presence of coherence,
and hence bounds that coherence in the original state.
For \(\varepsilon\) below an explicit threshold, the proof also bounds
\((\sum_{j=1}^n|\theta_j-\pi/4|^2)^{1/2}\) by a constant times
\(\sqrt\varepsilon\), where \(\theta_1\) is the canonical angle of the
qutrit's two-dimensional sector and \(\theta_j\), \(j\geq2\), are the
angles of the other parties; the constant is independent of \(n\)
(Appendix~\ref{app:original-state-stability}, paragraph on measurement
angles).
For binary POVMs, these are the angles of any projective replacement as in
Lemma~\ref{lem:projective-replacement}.

Under the same hypotheses, the proof also gives the more precise
estimate
\begin{equation}
 D_t(\rho)\leq
 \min\{2,\,A_\eta\sqrt{\varepsilon}+B_\eta\varepsilon\},
 \label{eq:original-state-stability-separated}
\end{equation}
with explicit constants \(A_\eta,B_\eta>0\) independent of \(n\).
Their values are given at the end of the proof in
Appendix~\ref{app:original-state-stability}.
For example, at \(\eta=1/4\)
one has \(\varepsilon_\eta\approx8\times10^{-7}\),
\(A_\eta\approx1\times10^{5}\) and \(B_\eta\approx3\times10^{9}\), so
\eqref{eq:original-state-stability-separated} is informative for
\(\varepsilon\) below about \(2\times10^{-10}\). The content of the
theorem is the exponent \(1/2\) and the uniformity in \(n\).

\begin{remark}
\label{rem:qutrit-position}
Let \(\overline W_n\) denote \(W_n\) with the two inputs interchanged at
every party. By \eqref{eq:bell},
\(\overline W_n=\tfrac12\overline W_{n-1}\otimes(A_{n,0}+A_{n,1})
+\tfrac12W_{n-1}\otimes(A_{n,0}-A_{n,1})\), and induction gives
\[
 W_n+\mathrm i\overline W_n
 =2\Bigl(\frac{1+\mathrm i}2\Bigr)^{n-1}(A_{1,0}+\mathrm iA_{1,1})
  \otimes\bigotimes_{j=2}^n(A_{j,0}-\mathrm iA_{j,1}).
\]
Up to normalization and local relabelings, Jha, Singh and Pan write the
MABK operator in the same form, as the Hermitian part of a tensor
product of the operators \(A_{j,0}\pm\mathrm iA_{j,1}\)
\cite[Eq.~(4)]{Jha2026}, following \cite{PanwarPandyaWiesniak2023}.
Since \(W_n\) and \(\overline W_n\) are Hermitian, \(W_n\) is the
Hermitian part of this product. Exchanging the two inputs of a party
turns \(A_{j,0}-\mathrm iA_{j,1}\) into
\(-\mathrm i(A_{j,0}+\mathrm iA_{j,1})\), and
\(A_{j,0}+\mathrm iA_{j,1}\) into \(\mathrm i(A_{j,0}-\mathrm iA_{j,1})\).
Hence a permutation of the parties, combined with exchanging the inputs
of the parties moved to and from the first position, leaves \(W_n\)
unchanged. Such relabelings leave the score unchanged, act on the state
only by reordering its tensor factors, and leave \(\Xi_n\) unchanged, since
\(G_n\) and the class of product channels are invariant under permutations
of the parties. Theorem~\ref{thm:original-state-stability} therefore holds
with the qutrit at any party, with \(\omega_t\) and \(D_t\) permuted
accordingly.
\end{remark}

For scores near the quantum maximum, the original state admits an
estimate with explicit constants that holds for every excess
\(\varepsilon\geq0\).

\begin{proposition}
\label{prop:high-score-stability}
Let \(n\geq3\), and let a binary-POVM strategy on
\(\mathbb C^3\otimes(\mathbb C^2)^{\otimes(n-1)}\) have state
\(\rho\) and score \(\beta\). With \(b=\beta/\beta_Q\) and
\(t=(\beta-\beta_B)/(\beta_Q-\beta_B)\) as in Section~\ref{sec:model}, put
\[
 b_0\coloneqq\sqrt{1-(\sqrt2-1)^2},
\]
and, for \(b>b_0\), define
\[
 d(b)\coloneqq\frac1{\sqrt2}
 -\left(1+\frac1{\sqrt2}\right)\sqrt{1-b^2}.
\]
If \(b_0<b\leq1\) and
\(0\leq\Xi_n(\rho)-L(\beta)\leq\varepsilon\), then
\begin{equation}
 D_t(\rho)\leq\min\left\{2,
 \sqrt{\left[\frac{4(1+\sqrt{1-t})^2}{d(b)}
                   +8(1+\varepsilon)\right]\varepsilon}
 +\left(\frac3{d(b)}+4\right)\varepsilon\right\}.
 \label{eq:high-score-stability}
\end{equation}
\end{proposition}

The threshold is \(b_0\approx0.91018\), or \(t\approx0.69333\).
For example, at \(t=0.9\) and \(\varepsilon=0.01\), the bound
in \eqref{eq:high-score-stability} is below \(0.702\).
The proof is given in Appendix~\ref{app:further-stability-certification}.

Combining Proposition~\ref{prop:high-score-stability} with
Theorem~\ref{thm:original-state-stability} removes the upper cutoff
from the latter's square-root estimate.

\begin{proposition}
\label{prop:stability-to-maximum}
For every \(0<\eta<1/2\), there are constants
\(C_\eta^+,\varepsilon_\eta^+>0\), independent of \(n\), with the
following property. Let \(n\geq4\), and let a binary-POVM strategy on
\(\mathbb C^3\otimes(\mathbb C^2)^{\otimes(n-1)}\) have state
\(\rho\) and score \(\beta\). If
\[
 \eta\leq t\leq1,\qquad
 0\leq\Xi_n(\rho)-L(\beta)\leq\varepsilon\leq\varepsilon_\eta^+,
\]
then \(D_t(\rho)\leq C_\eta^+\sqrt\varepsilon\).
\end{proposition}

The proof, with explicit constants, is given in
Appendix~\ref{app:further-stability-certification}.

The coherent family \(\omega_{t,c}\) of Proposition~\ref{prop:coherent-qutrit} also determines the best possible exponent.
At the equal mixture, its spectrum gives the distance after
optimizing over local unitaries.

\begin{proposition}
\label{prop:stability-exponent}
Let \(n\geq3\) and \(c\in\mathbb C\) with \(|c|\leq1/2\).
For the state \(\omega_{1/2,c}\) in
\eqref{eq:coherent-qutrit}, put
\[
 \delta_c\coloneqq\Xi_n(\omega_{1/2,c})-\frac34.
\]
Then
\begin{equation}
 D_{1/2}(\omega_{1/2,c})
 =2|c|
 =2\sqrt{\delta_c+2\delta_c^2}.
 \label{eq:stability-exponent}
\end{equation}
For every \(0<\eta<1/2\), the exponent \(1/2\) in
\eqref{eq:original-state-stability} cannot be replaced by any
larger exponent, even with \(n\) fixed.
\end{proposition}

\begin{proof}
The eigenvalues of \(\omega_{1/2,c}\) on the span of \(G,B\)
are \(1/2+|c|\) and \(1/2-|c|\), and the state vanishes
on its orthogonal complement. Also,
\(\|\omega_{1/2}\|_\infty=1/2\).
Let \(U\) be any product unitary, set
\(\sigma\coloneqq U\omega_{1/2,c}U^*\), and let \(Q\) be a rank-one
projection onto an eigenvector of \(\sigma\) with eigenvalue
\(1/2+|c|\). Then
\[
 \tr\!\left[Q(\sigma-\omega_{1/2})\right]\geq |c|.
\]
For a Hermitian trace-zero operator \(\Delta\), with positive part \(\Delta_+\),
\(\|\Delta\|_1=2\tr(\Delta_+)\geq2\tr(Q\Delta)\). Hence
\[
 \|\sigma-\omega_{1/2}\|_1\geq2|c|.
\]
The identity unitaries attain \(2|c|\), since the difference has
eigenvalues \(|c|,-|c|\) on the span of \(G,B\). This proves
the first equality. Equation~\eqref{eq:coherent-extractability}
gives
\[
 \delta_c=
 \frac{\sqrt{1/4+2|c|^2}-1/2}{2},
 \qquad
 |c|^2=\delta_c+2\delta_c^2,
\]
which proves the second; the identity for \(|c|^2\) is the case
\(t=1/2\) of \eqref{eq:coherent-distance}.
The parameter \(t=1/2\) belongs to every interval
\([\eta,1-\eta]\) under consideration. As nonzero \(c\) tends
to zero, \(\delta_c\) tends to zero and, for any exponent \(\alpha>1/2\),
\[
 \frac{D_{1/2}(\omega_{1/2,c})}{\delta_c^\alpha}
 =2\delta_c^{1/2-\alpha}\sqrt{1+2\delta_c}
 \longrightarrow\infty.
\]
This excludes every exponent \(\alpha>1/2\).
\end{proof}

The lower cutoff \(t\geq\eta\) in Theorem~\ref{thm:original-state-stability} is necessary even if the endpoint itself is excluded.
For \(0\leq t<1/2\), put \(b\coloneqq1/s+\kappa t\) and \(p\coloneqq F(b)\),
and embed the qubit attainer \(\rho_p\) of \eqref{eq:qubit-attainer},
with its observables, at the first party, giving both observables the
value \(+1\) on \(\ket2\). By the proof of Theorem~\ref{thm:qubit-profile},
its score is \(\beta_Qb\), and by the restriction and extension argument
in the proof of Theorem~\ref{thm:minimal-dimension} its extractability is
\(p\). Hence its excess \(F(b)-(1+t)/2\) vanishes at \(t=0\) and tends to
zero as \(t\downarrow0\).
Its first marginal has operator norm \(1/2\), whereas
\(\omega_t\) has first-party weight \(1-t\) on \(\ket2\).
Testing the first-party projection onto \(\ket2\) after any local unitary, as in the proof of
Proposition~\ref{prop:stability-exponent}, gives
\[
 D_t(\rho_p)\geq 2\bigl((1-t)-1/2\bigr)=1-2t.
\]
Thus the conclusion fails at \(t=0\), and no fixed positive constants
\(C,\varepsilon_0\) give \(D_t(\rho)\leq C\sqrt\varepsilon\) for all
\(0<t\leq1-\eta\), even with \(n\) fixed.
With the observables of Theorem~\ref{thm:qutrit-attainer} held fixed, Proposition~\ref{prop:fixed-observable-stability} in Appendix~\ref{app:original-state-stability} bounds \(\|\rho-\omega_t\|_1\) itself for every \(n\geq3\) and \(0\leq t\leq1\).

The near-minimality hypothesis concerns actual extractability.
For the fixed observables of Theorem~\ref{thm:qutrit-attainer},
all outcome probabilities of \(\omega_{t,c}\) are independent
of \(c\). Every first-party measurement effect is block diagonal
in the decomposition into its qubit and scalar sectors, so
every joint outcome effect has zero matrix element between
\(G\) and \(B\). Proposition~\ref{prop:coherent-qutrit} shows
that extractability nevertheless increases with \(|c|\).
These statistics alone therefore do not determine the
extractability excess assumed in
Theorem~\ref{thm:original-state-stability}.
A stronger hypothesis, computable from a description of the state, is a
bound on the excess of the negativity across the qutrit party
(Remark~\ref{rem:negativity-certificate}).

\section{Certification from additional local measurements}
\label{sec:local-certification}

The actual-extractability hypothesis can be replaced by measured
quantities when the local dimensions and the qutrit's two sectors are
known and the additional qutrit measurements are calibrated. The score and the sector weight bound the weight of the state outside a
two-dimensional subspace that a product of local unitaries maps onto
the span of the two attaining branches, while local qutrit coherences
control the coherence between the branches. The first bound uses the spectral-gap approach of GHZ
verification~\cite{LiHanZhu2020} and subspace
verification~\cite{ChenEtAl2025}.

\begin{theorem}
\label{thm:local-certification}
Let \(n\geq3\), and let a binary-POVM strategy on
\(\mathbb C^3\otimes(\mathbb C^2)^{\otimes(n-1)}\) have state
\(\rho\) and score \(\beta_B<\beta\leq\beta_Q\).
Suppose that a known rank-two orthogonal projector \(Q_1\) satisfies
\([A_{1,x},Q_1]=0\) for \(x=0,1\), and choose a unit vector
\(\ket{e_2}\) spanning its orthogonal complement. Put
\[
 \rho_1\coloneqq\tr_{2,\ldots,n}\rho,\qquad
 q\coloneqq\tr(Q_1\rho_1),\qquad
 v\coloneqq\|Q_1\rho_1\ket{e_2}\|_2,\qquad
 \delta\coloneqq q-t,
\]
where \(b=\beta/\beta_Q\), \(t=(b-1/s)/\kappa\),
\(s=\sqrt2\), and \(\kappa=1-1/s\). Then \(\delta\geq0\) and
\begin{equation}
 \begin{split}
 D_t(\rho)\leq\min\bigl\{2,\,
 &2(1+s)\sqrt{r(b)\delta}
 +(2+r(b))\delta+2s v\bigr\},\\
 &r(b)\coloneqq\frac{s-1}{1-\sqrt{1-b^2}}.
 \end{split}
 \label{eq:local-certification}
\end{equation}
Here \(D_t\) and \(\omega_t\) are given by
\eqref{eq:local-unitary-distance} and \eqref{eq:qutrit-attainer}.
In particular, \(q=t\) and \(v=0\) imply \(D_t(\rho)=0\).
\end{theorem}

The measured inputs are the original MABK score and five additional
binary settings at party~\(1\). The projective measurement
\(\{Q_1,I-Q_1\}\) gives \(q\). For a known orthonormal basis
\(e_0,e_1\) of \(\operatorname{ran}Q_1\), the four observables
\[
 X_{k2}\coloneqq\ket{e_k}\!\bra{e_2}+\ket{e_2}\!\bra{e_k},\qquad
 Y_{k2}\coloneqq-\mathrm i\ket{e_k}\!\bra{e_2}
                   +\mathrm i\ket{e_2}\!\bra{e_k},\qquad k=0,1,
\]
implemented with effects \((I\pm X_{k2})/2\) and
\((I\pm Y_{k2})/2\), determine
\[
 4v^2=\sum_{k=0}^1
 \bigl(\langle X_{k2}\rangle_{\rho_1}^2
       +\langle Y_{k2}\rangle_{\rho_1}^2\bigr).
\]
The local dimensions and the preservation of \(\operatorname{ran}Q_1\)
by the qutrit Bell observables are assumed, and these five measurements
are calibrated; the Bell observables may otherwise be unknown. Since
\(r(b)\leq\sqrt2\), the constants are uniformly bounded throughout
the stated score interval and for all \(n\geq3\).
Appendix~\ref{app:local-certification} proves the estimate for the
original state, including its coherence between the two sectors.
For binary POVMs, it applies the replacement of
Lemma~\ref{lem:projective-replacement} with the qutrit Bell observables
restricted to commute with \(Q_1\), preserving the state and the
quantities \(q\) and \(v\).

For the fixed calibrated observables of
Theorem~\ref{thm:qutrit-attainer}, a sharper bound uses three
additional settings. In its coordinates take
\(Q_1\coloneqq\proj{0}+\proj{1}\), \(\ket{e_0}\coloneqq\ket{x+}\),
\(\ket{e_2}\coloneqq\ket2\), and
\(m\coloneqq\bra{e_0}\rho_1\ket{e_2}\).
For every \(\beta_B\leq\beta\leq\beta_Q\), with the same
\(t,q,\delta=q-t\), one has \(\delta\geq0\) and
\begin{equation}
 \|\rho-\omega_t\|_1\leq
 \min\bigl\{2,\,
 2\sqrt{1+s}\sqrt\delta+(1+s)\delta+2s|m|\bigr\}.
 \label{eq:fixed-local-certification}
\end{equation}
The measurement of \(Q_1\) and the two observables \(X_{02},Y_{02}\)
supply these data, since
\(2|m|=(\langle X_{02}\rangle_{\rho_1}^2+
\langle Y_{02}\rangle_{\rho_1}^2)^{1/2}\).
Thus \(q=t\) and \(m=0\) certify \(\rho=\omega_t\) in the fixed
coordinates, including both score endpoints.
Appendix~\ref{app:local-certification} also proves
\eqref{eq:fixed-local-certification}.

\section{Discussion}
\label{sec:discussion}

We have proved that, between the biseparable bound and the quantum maximum, the least GHZ
extractability consistent with an MABK score is affine for every
\(n\geq3\).
The affine bound holds for density operators on tensor products of local
spaces of any finite or infinite dimension and for binary POVMs, and it
is attained.
For three parties the bound extends Kaniewski's finite-dimensional
bound~\cite{Kaniewski2016}; from four parties onward the proof gives one
analytic argument, which covers the four- and five-party bounds stated by
Jha, Singh and Pan~\cite{Jha2026}. For \(n\geq6\) it covers the two
scalar cases left open by Cao et~al.~\cite{Cao2026} and proved in
Lean~\cite{ElonSG2026}.

Qubit strategies have a strictly larger minimum at every interior
score, whereas one additional level at a single party suffices: one
qutrit and qubits elsewhere attain the affine line with the minimum
product of the local ranks.
The derivation of the qubit minimum combines the spectral decomposition
and overlap bound of Scarani and Gisin~\cite{ScaraniGisin2001}, as in the
two-qubit argument of Bardyn et~al.~\cite{Bardyn2009}, with explicit
multipartite attaining families.
The exact extractability calculation for the coherent family shows
that every nonzero coherence between the two branches of the attaining
mixture raises
the extractability above the affine minimum while leaving the score of
the fixed measurements unchanged.
It also gives an exact trace-norm distance to the incoherent mixture
within this coherent family.

For \(n\geq4\), we strengthen the operator inequality behind the
lower bound to a spectral bound whose constant is independent of
\(n\) (Theorem~\ref{thm:uniform-gap}).
For finite-dimensional projective strategies with \(n\geq4\) whose
actual extractability is close to the affine minimum, our spectral bound
controls the distribution of the local measurement angles, with
constants independent of \(n\).
At exact equality, after local pinching into irreducible sectors, every joint sector of positive probability is of one of two kinds.
In the first, the conditional state is a locally embedded GHZ state
and the local measurements are ideal.
In the second, one party is deterministic and the other \(n-1\)
parties share a locally embedded GHZ state with ideal measurements.
The score determines the total weights of the two kinds.
For finite-dimensional binary POVMs, these necessary statements about
the pinched sectors concern the Naimark dilation of
Appendix~\ref{sec:lift}.
This sector classification describes the locally pinched state.

The original state itself is fixed at exact equality when \(n\geq4\), the
score is interior and the product of the local ranks is minimal: up to
local isometries and a party permutation, it is the incoherent attaining
mixture constructed with one qutrit.
On one qutrit and qubits elsewhere, again for \(n\geq4\), suppose that
the score is at least a fixed fraction of the way from the biseparable
bound to the quantum maximum.
For a sufficiently small excess of the actual extractability above the
affine minimum, we bound the original
state's trace-norm distance from the attaining mixture of the same
score, up to local unitaries, by an explicit constant
times the square root of that excess. The constant is independent
of \(n\), and the score interval includes the quantum maximum.
This lower cutoff on the score cannot be removed when the measurements
on these spaces are arbitrary; with the measurements of the attaining
strategy held fixed, Proposition~\ref{prop:fixed-observable-stability}
bounds the trace-norm distance from the attaining mixture by a constant
times the square root of the extractability excess, for every
\(n\geq3\) and every score from the biseparable bound to the quantum
maximum.
The coherent family shows that the exponent \(1/2\) is optimal
(Proposition~\ref{prop:stability-exponent}).
For arbitrary measurements, Proposition~\ref{prop:high-score-stability}
gives an explicit estimate at high scores, valid for every excess,
already for \(n\geq3\).

Theorem~\ref{thm:local-certification} supplies a sufficient condition
in measured quantities for proximity to the same attaining state.
Given the local dimensions and a known two-dimensional qutrit
subspace invariant under both qutrit Bell observables, it uses the
MABK score and five calibrated binary settings
at the qutrit, with constants uniform in \(n\) and in the score
above the biseparable bound.
The Bell observables may otherwise be unknown.
Bell statistics alone cannot replace these calibrated additional
measurements: they are identical along the coherent family even as its
extractability changes.

Open problems include classifying the original state and its coherences
at larger local ranks and extending stability to those ranks.
For three parties, the concentration estimate and equality classification
after local pinching remain open, as do uniqueness and stability of the
original state throughout the interior score range with arbitrary
measurements. The bound after \eqref{eq:single-block-stability} controls
the Euclidean norm of the deviations from a right angle in a single joint
qubit sector, omitting the largest deviation.
Appendix~\ref{sec:axis-formulas} shows that the exponent \(1/4\) is optimal
for spectral bounds on the operator~\(T\) of Section~\ref{sec:channels}.
Whether it is optimal for the actual excess \(\Xi_n(\rho)-L(\beta)\)
remains open.
For the certificate from additional local measurements, natural extensions are
quantitative tolerance to an imperfectly invariant subspace and to
imperfect calibration of the additional measurements, and confidence
bounds from finitely many samples.

\appendix

\section{Canonical target and local reductions}
\label{app:canonical}

The proof in Section~\ref{sec:bound} uses one fixed GHZ target throughout the angle cube. We verify its relation to the recursive target of \cite{Cao2026}, prove the covariance used to reduce the angles, and supply the local reductions needed for arbitrary local dimension.

\subsection{The target phase}

The target \eqref{eq:target} is the recursive target of \cite{Cao2026};
the superscripted vectors \(\ket{\Psi_n^+}\) and
\(\ket{\Phi_n^-}\) below follow its notation.
For completeness, start with
\(\ket{\Psi_1^+}\coloneqq\ket0\), \(\ket{\Phi_1^-}\coloneqq\ket1\), and define
\begin{align*}
 \ket{\Psi_n^+}
 &\coloneqq\frac{\ket{\Psi_{n-1}^+}\ket1+
               \ket{\Phi_{n-1}^-}\ket0}{\sqrt2},\\
 \ket{\Phi_n^-}
 &\coloneqq\frac{\ket{\Psi_{n-1}^+}\ket0-
               \ket{\Phi_{n-1}^-}\ket1}{\sqrt2}.
\end{align*}
The coefficient pairs of \((\ket{\Psi_n^+},\ket{\Phi_n^-})\) on \(\ket{\tau}_Y\) are
\((1,-\mathrm i)/\sqrt2\) for odd \(n\), and
\((-\mathrm i,1)/\sqrt2\) for even \(n\); the coefficients on
\(\ket{-\tau}_Y\) are their conjugates.
For the step from \(n-1\) to \(n\), write
\(\ket0=(\ket{\tau_n}_Y+\ket{-\tau_n}_Y)/\sqrt2\) and
\(\ket1=-\mathrm i\tau_n(\ket{\tau_n}_Y-\ket{-\tau_n}_Y)/\sqrt2\),
and let \(\tau'\coloneqq(\tau_1,\ldots,\tau_{n-1})\). Suppose that
\(\ket{\Psi_{n-1}^+}=\alpha\ket{\tau'}_Y+\overline\alpha\ket{-\tau'}_Y\) and
\(\ket{\Phi_{n-1}^-}=\gamma\ket{\tau'}_Y+\overline\gamma\ket{-\tau'}_Y\),
where \((\alpha,\gamma)\) is the stated pair for \(n-1\).
For either parity of \(n\), this pair satisfies \(\gamma=-\mathrm i\tau_n\alpha\).
In the recursion, the mixed strings \(\ket{\tau'}_Y\ket{-\tau_n}_Y\) and
\(\ket{-\tau'}_Y\ket{\tau_n}_Y\) therefore have coefficients proportional to
\(\mathrm i\tau_n\alpha+\gamma=0\) in \(\ket{\Psi_n^+}\) and to
\(\alpha-\mathrm i\tau_n\gamma=0\) in \(\ket{\Phi_n^-}\), or to their conjugates.
The coefficients on \(\ket{\tau}_Y\) are
\((-\mathrm i\tau_n\alpha+\gamma)/2=-\mathrm i\tau_n\alpha\) and
\((\alpha+\mathrm i\tau_n\gamma)/2=\alpha\), and those on \(\ket{-\tau}_Y\) are
their conjugates. Since \(\tau_n=(-1)^n\), the new pair
\((-\mathrm i\tau_n\alpha,\alpha)\) is \((-\mathrm i,1)/\sqrt2\) for even \(n\)
and \((1,-\mathrm i)/\sqrt2\) for odd \(n\). The case \(n=1\) is the
two expansions above with \(\tau_1=1\). Thus
\(\cos(\pi/8)\ket{\Psi_n^+}+\sin(\pi/8)\ket{\Phi_n^-}\), which is the target of \cite[Supplemental A.3]{Cao2026},
equals \eqref{eq:target}.

\subsection{Covariance under angle reflections}
\label{app:covariance}

The extraction estimate is proved on the reduced angle cube.
The following conjugations transport it to every canonical measurement
pair while preserving the target and the local form of the channels.
These are the symmetries under angle reflections of
\cite[App.~B]{Kaniewski2016} for \(n\leq3\) and
\cite[Supplemental A.4, Observation 1]{Cao2026} in general, written in
the present phase convention.

For a party \(i\), let \(S_i\) have \(X\) on every site \(j\ne i\)
and \((X+\tau_i Z)/s\) at \(i\). It is a tensor product of Hermitian
unitaries. The identities
\[
 S_i\ket{\tau}_Y=e^{\mathrm i\delta_n}\ket{-\tau}_Y,\qquad
 S_i\ket{-\tau}_Y=e^{-\mathrm i\delta_n}\ket{\tau}_Y
\]
show that it stabilizes \(\ket{\Phi_n}\). Indeed, the first phase is
\(\mathrm i^{\,n-1}(\prod_j\tau_j)e^{\mathrm i\pi/4}
=e^{\mathrm i\delta_n}\), independent of \(i\).

Let \(\theta'\) differ from \(\theta\) only at
\(\theta_i'\coloneqq\pi/2-\theta_i\). The coefficient \(g\) is unchanged.
Conjugation by \((X+\tau_i Z)/s\) exchanges \(X\) and \(\tau_i Z\),
while conjugation by \(X\) on another site preserves its dephasing
Pauli up to sign. Signs cancel in the channel, so
\[
 S_i\Lambda_\theta(B)S_i
 =\Lambda_{\theta'}(S_iBS_i),
 \qquad
 S_iK(\theta)S_i=K(\theta'),
\]
where \(\Lambda_\theta\coloneqq\bigotimes_j\Lambda_{\theta_j}\).
At \(\theta_i=\pi/4\), the channel is the identity, so either choice
of dephasing axis gives the same map.

For the Bell operator, \eqref{eq:Wblock} remains valid on the full cube
if the factors \(\sqrt{b_j}\) in \(x,y\) are interpreted as the signed
quantities \((\cos\theta_j-\sin\theta_j)/s\), since the induction in the
proof of Lemma~\ref{lem:blocks} holds for either sign of
\(\cos\theta_j-\sin\theta_j\).
Reflection changes \(y\) to \(-y\) if \(i\in D\), and \(x\) to \(-x\)
if \(i\in H\). In the corresponding \(Y\) block, in the ordered basis of
Lemma~\ref{lem:blocks}, the lower off-diagonal
entry of \(S_i\) is \(\overline\zeta\) in the first case and
\(\mathrm i\overline\zeta\) in the second. Conjugation gives precisely
those changes. Therefore
\[
 S_iW_n(\theta)S_i=W_n(\theta'),\qquad
 S_iT(\theta)S_i=T(\theta').
\]
Hence repeated reflections extend positivity from \([0,\pi/4]^n\) to the full cube.

\subsection{Local reductions and normal states}
\label{sec:lift}

For completeness, we verify the Naimark dilation with one common
embedding for the two binary measurements and the passage from finite-dimensional states to normal states on tensor products. The finite-dimensional extraction itself was assembled in Section~\ref{sec:bound}.

For a Hermitian contraction \(A_x\) on a local Hilbert space
\(\mathcal H\), define on \(\mathcal H\oplus\mathcal H\)
\[
 \widetilde A_x\coloneqq
 \begin{pmatrix}
 A_x&\sqrt{I-A_x^2}\\
 \sqrt{I-A_x^2}&-A_x
 \end{pmatrix}.
\]
This is a reflection. The same embedding
\(V\psi\coloneqq(\psi,0)\) works for both inputs and satisfies
\(V^\ast\widetilde A_xV=A_x\).
Tensoring the local embeddings preserves every Bell correlator.
Any local extraction on the embedded state composes with this
embedding to give a permitted extraction on the original state.

In finite dimension, Jordan's lemma decomposes the local space into common invariant subspaces of the two reflections, of dimensions at most two. On a two-dimensional sector on which the two reflections do not commute, their sum and difference
anticommute and can be oriented as positive multiples of \(X,Z\).
This gives \eqref{eq:canonical}, with the alternating sign absorbed in
the local orientation of \(Z\).
A one-dimensional sector is identified with a common eigenvector of the
canonical pair at a commuting endpoint: \(\theta=0\) for equal
observable values, and \(\theta=\pi/2\) for opposite values.

To pass to arbitrary local dimension, let \(\rho\) be a normal state on the
tensor product of local spaces of arbitrary dimension. Each local marginal
\(\rho_j\) is trace class, so its support is a separable Hilbert space; let
\(P_{j,m}\) be finite-rank spectral projections of \(\rho_j\) with
\(\tr(P_{j,m}\rho_j)\to1\), and put \(P_m\coloneqq\bigotimes_jP_{j,m}\). Since
\(I-\bigotimes_jP_{j,m}\leq\sum_j(I-P_{j,m})\), we have \(\tr(P_m\rho)\to1\),
and the gentle measurement lemma \cite{Winter1999,OgawaNagaoka2007} in
the form \(\|\rho-P_m\rho P_m\|_1\leq2\sqrt{1-\tr(P_m\rho)}\), whose proof
in \cite{OgawaNagaoka2007} does not use the dimension, gives
\[
 \rho_m\coloneqq\frac{P_m\rho P_m}{\tr(P_m\rho)}
 \longrightarrow\rho
 \quad\hbox{in trace norm}.
\]
Compress the two local POVMs to these finite supports. They remain
binary POVMs, and their score equals
\(\beta_m\coloneqq\tr(\rho_mW_n)\to\beta\), because \(W_n\) is bounded.
Apply the finite-dimensional bound in its globally valid form,
including the floor \(1/2\).

Extend each resulting local channel to the original space by testing
\(P_{j,m}\), using the finite-support channel on success, and preparing
a fixed qubit on failure. This is trace-preserving. The pullback of the target projector through this extraction lies between \(0\) and \(I\), so the overlap changes by at most
\(\|\rho-\rho_m\|_1\) when \(\rho_m\) is replaced by \(\rho\).
Taking the supremum over local channels and then the limit proves
\eqref{eq:main} for normal states in the tensor-product model.

\section{Scalar positivity estimates}
\label{app:positivity}

This appendix supplies the product estimate and the remaining partition
cases used in Section~\ref{sec:bound}. Throughout \(n\geq4\), the
angles lie in the reduced cube \([0,\pi/4]^n\) of
Section~\ref{sec:channels}, the
variables \(v_j\) lie in \([0,\kappa]\), and complement symmetry
allows \(h\leq d\).

\subsection{Product estimates}
\label{app:products}

We prove the three inequalities in Lemma~\ref{lem:products}.

\begin{proof}[Proof of Lemma~\ref{lem:products}]
First,
\[
 q_j-A_j=v_j(\kappa-v_j)\geq0.
\]
For \(c_j\coloneqq q_j(1-v_j)\), one has \(c_j\leq A_j\). If \(v_j>0\), then
\[
 \frac{1-c_j}{1-A_j}
 =\frac{a+1-av_j}{2-v_j}\leq\frac{a+1}{2}=r.
\]
At zero the corresponding undivided inequality is still true.
Telescope \(1-\prod_jc_j\), and bound its prefix products by the
corresponding products of \(A_j\). This yields
\(1-C_J\leq r(1-A_J)\).

Finally, expand \(1-A_J\) over nonempty subsets:
\[
 1-A_J=\sum_{\varnothing\ne S\subseteq J}
       \prod_{i\in S}b_i\prod_{i\in J\setminus S}(1-b_i)
 \geq(2^m-1)B_J,
\]
since \(1-b_i\geq b_i\). This proves the last assertion,
including \(m=1\) and the all-zero point.
\end{proof}

\subsection{At least two indices in the smaller part}
\label{app:two-indices}

For all \(h,d\geq2\), we bound the two Bell amplitudes together.
First, the dephasing probabilities satisfy \(p_j\geq b_j^{9/8}\): indeed,
\[
 \max_{0\leq v\leq\kappa}v(2-v)^9
 =\frac{9^9}{5^{10}}<64<a^8,
\]
where the maximum is at \(v=1/5\) and
\(a^2=3/2+s>2s\). Multiplication by \(v^8\) proves the
claim also at \(v=0\). The exponent \(9/8\) is chosen only to keep
the arithmetic short. Since \(Q_J\geq A_J\) by Lemma~\ref{lem:products},
the block variables of Lemma~\ref{lem:blocks} therefore satisfy
\[
 w\geq B_H^{9/8}A_D\geq x^{9/4},\qquad
 z\geq A_HB_D^{9/8}\geq y^{9/4}.
\]
Put \(u\coloneqq x^2+y^2=B_HA_D+A_HB_D\). As in
\cite[Supplemental C, Eq.~S57]{Cao2026}, this is
multiaffine in the \(b_j\), so its maximum on \([0,1/2]^n\)
is attained at a vertex. There its value is zero unless every
coordinate in at least one part is \(1/2\). If exactly one
part has this property, the value is at most \(1/4\), since that part
has at least two indices; if both do, it is \(2^{1-n}\). Thus
\(u\leq1/4\).

In the variables \(w,z,x,y\), \eqref{eq:lambda} reads
\(\lambda_{H,D}=1+wz+s(w+z)-k^2(x^2+y^2)+s k(wx+zy)\).
Dropping \(wz\geq0\), inserting the bounds on \(w,z\) and using
convexity of \(r\mapsto r^{9/8}\) and \(r\mapsto r^{13/8}\) at
\(r=x^2,y^2\) gives
\[
 \lambda_{H,D}\geq
 1-k^2u+2s(u/2)^{9/8}+2s k(u/2)^{13/8}.
\]
On \(0\leq u\leq1/4\), since \((u/2)^{1/8}\leq1\) and
\((u/2)^{5/8}\leq1/(2s)\), the derivative of the right-hand side is
at most
\[
 -k^2+\frac{9s}{8}+\frac{13k}{16}
 =-\frac{35+s}{16}<0.
\]
Finally, \(2^{1/8}>17/16\) and \(2^{5/8}>3/2\), as is
seen by raising both inequalities to the eighth power. Hence
\[
 \lambda_{H,D}\geq1-\frac{k^2}{4}+2^{-15/8}+k2^{-27/8}
 >\frac{39-26s}{64}>e_*>0,
\]
with \(e_*=(31-21s)/64\) as in the proof of
Theorem~\ref{thm:uniform-gap}, which uses this smaller constant.

\subsection{The case H is empty}

The target block requires a different estimate because for \(H=\varnothing\) the right side of \eqref{eq:interpolation} is nonpositive, since \(f_0=0\) and \(g_{0,n}<0\), and the bound of
Appendix~\ref{app:two-indices} needs two indices in each part. For \(h=0\) we follow \cite[Supplemental C, Case 2]{Cao2026}: a lower
polynomial that is concave in each variable reduces positivity to the
cube vertices. The new input is the bound \eqref{eq:lastvertex} at
\((\kappa,\ldots,\kappa)\); see the Remark below.

For \(h=0\), omit the index set on products over all \(n\) variables.
Dropping \(s k E\geq0\) and replacing \(B\) by its multilinear majorant \(2^n\prod_iv_i\), valid since \(b_i\leq2v_i\), yields
\begin{equation}
 \lambda_{\varnothing,[n]}\geq
 \mathcal P\coloneqq
 1+PQ+s(P+Q)-k^2\left(A+2^n\prod_iv_i\right)+s k C.
 \label{eq:lowerpolynomial}
\end{equation}
The majorant is used because each \(b_i=v_i(2-v_i)\) is concave, so \(-k^2B\) itself would be convex in each variable, whereas \(-k^2 2^n\prod_iv_i\) is affine in each. This polynomial is concave in each variable separately. If a subscript
\(-i\) denotes the product over the remaining indices, then
\[
 \partial_i^2\mathcal P
 =-2a^2P_{-i}Q_{-i}-2k^2A_{-i}+2k^2C_{-i}\leq0,
\]
because \(C_{-i}\leq A_{-i}\).
Repeated one-variable chord inequalities lower-bound \(\mathcal P\)
by a convex combination of its cube-vertex values.

At the all-zero vertex the value is zero. At a vertex with \(j\)
coordinates equal to \(\kappa\) and the others zero, \(1\leq j<n\),
the value is
\[
 1-(3+s)2^{-j}+s k\,2^{-3j/2}.
\]
For \(\chi\coloneqq2^{-j/2}\in(0,1/s]\), this is
\(1-(3+s)\chi^2+s k \chi^3\), whose derivative is
\[
 \chi[-2(3+s)+3s k \chi]\leq \chi(s-3)<0.
\]
Its minimum on that interval is therefore the \(j=1\) value, which is
zero.

At the vertex \((\kappa,\ldots,\kappa)\), the polynomial itself equals
\[
 \mathcal P(\kappa,\ldots,\kappa)
 =1+4^{-n}-3\cdot2^{-n}-k^2(2\kappa)^n+s k\,2^{-3n/2}.
\]
Since \(0<2\kappa<1\) and \(n\geq4\), dropping the two positive terms
gives
\begin{equation}
 \mathcal P(\kappa,\ldots,\kappa)
 \geq1-\frac3{16}-k^2(2\kappa)^4
 =8s-\frac{179}{16}>0.
 \label{eq:lastvertex}
\end{equation}
The last sign follows from \(2\cdot128^2>179^2\).
All vertex values are nonnegative, proving the case \(H=\varnothing\).

\begin{remark}
\label{rem:cao-vertex}
The concavity argument requires the value of the lower polynomial at
\((\kappa,\ldots,\kappa)\). Supplemental C, Case 2 of \cite{Cao2026}
checks the original determinant at this vertex in Eq.~S58.
Equation~\eqref{eq:lastvertex} supplies the required positive lower
bound for the polynomial for every \(n\geq4\).
\end{remark}

Together with the calculation for \(h=1\) in Section~\ref{sec:bound},
these cases give a nonnegative determinant in every block. The strict bounds away from the axes are also the inputs to
Theorem~\ref{thm:uniform-gap}.

\section{A construction with local classical registers}
\label{app:flags}

This construction, which uses classical registers as in \cite[Eq.~(2)]{Kaniewski2016} and, for \(n\) parties, \cite[Eqs.~(10), (38)]{Jha2026}, mixes a strategy of score \(\beta_Q\) with one of score \(\beta_B\), using a classical bit stored in a qubit and one data qubit at each party. It provides a direct attaining strategy with local dimension four, independently of the one-qutrit construction in Section~\ref{sec:dimension}.

Fix \(t\in[0,1]\). At each party, \(F_j\) is a qubit used as a classical
register, and the other qubit holds the conditional state.
Let \(\rho_Q\) be a locally rotated GHZ state attaining \(\beta_Q\),
with local observables \(A^Q_{j,x}\).
Let
\[
 \rho_B\coloneqq\rho_{n-1,Q}\otimes\proj0,
\]
where the first \(n-1\) parties attain the \((n-1)\)-party quantum
maximum \(2^{n/2}\), and take both observables of the last party to be \(Z\); write \(A^B_{j,x}\) for the resulting local observables.
For \(n-1=2\), the maximizing state is a Bell state.
For \(\rho_B\), \eqref{eq:bell} becomes \(W_{n-1}\otimes Z\), so its score is exactly \(\beta_B\). This state is separable across the cut between the last party and the others.

Define
\begin{align*}
 \rho_t\coloneqq{}&t\,\proj{0^n}_F\otimes\rho_Q
          +(1-t)\,\proj{1^n}_F\otimes\rho_B,\\
 \widehat A_{j,x}\coloneqq{}&
 \proj0_{F_j}\otimes A^Q_{j,x}
 +\proj1_{F_j}\otimes A^B_{j,x}.
\end{align*}
These observables are defined on the entire local space, including
the unused subspaces in which the classical registers have different
values.
On the register string \(0^n\) every observable acts as \(A^Q_{j,x}\), and
on \(1^n\) as \(A^B_{j,x}\), so
\[
 \beta(t)=t\beta_Q+(1-t)\beta_B.
\]

For any fixed product of local channels, the second branch remains
separable across the same cut. Across this cut, GHZ has two equal
Schmidt coefficients. Its overlap with any pure product
state is at most \(1/2\), and convexity gives the same bound for mixed
separable states. The first branch contributes at most \(1\).
Linearity for each fixed channel therefore implies
\[
 \Xi_n(\rho_t)\leq t+\frac{1-t}{2}=\frac{1+t}{2}.
\]
Equality is attainable locally. For register value \(0\), each party applies
its fixed unitary converting \(\rho_Q\) to the standard GHZ state.
For register value \(1\), each party discards the input and prepares \(\ket0\).
The target overlap is then exactly \(t+(1-t)/2\).

Thus this strategy attains the infimum in Theorem~\ref{thm:main}, with local
dimension four; for \(0<t<1\) its local ranks are \(4,\ldots,4,3\), with
product \(3\cdot4^{n-1}\), against the minimum \(3\cdot2^{n-1}\) of
Theorem~\ref{thm:minimal-dimension}. The classical registers combine the
branch-dependent measurement pairs into one pair of observables on a single
local space.

\section{Axis formulas}
\label{sec:axis-formulas}

The following exact formulas distinguish \(\tr(\rho T)=0\) from
\(\Xi_n(\rho)=L(\beta)\). They also describe what the spectral estimate of
Theorem~\ref{thm:uniform-gap} measures.

\begin{proposition}
\label{prop:axis-gap}
Suppose \(n\geq4\), and only \(v_i=v\) can be nonzero.
Put \(\theta_*\coloneqq\min\{\theta_i,\pi/2-\theta_i\}\).
The least eigenvalue of \(T\) on \(\ket{\Phi_n}^{\perp}\) is
\begin{equation}
 \Delta_{\mathrm{axis}}
 \coloneqq\frac1s+a\bigl(v-\sqrt{v(2-v)}\bigr)
 =k(1-\cos\theta_*).
 \label{eq:axis-gap}
\end{equation}
Thus, when \(\theta_*>0\),
\[
 \tr(\rho T)\geq
 \Delta_{\mathrm{axis}}\bigl(1-\bra{\Phi_n}\rho\ket{\Phi_n}\bigr).
\]
\end{proposition}

\begin{proof}
On the reduced axis, set \(p\coloneqq av\), \(q\coloneqq1-av\), and \(b_i=v(2-v)\).
The target block has eigenvalues \(0,q+s\).
The block with \(H=\{i\}\) has eigenvalues
\[
 p+\frac1s-a\sqrt{b_i},
 \qquad
 \frac1s+a\sqrt{b_i},
\]
and all remaining eigenvalues are \(1/s\).
The first displayed value is nonnegative and is no larger than
\(1/s\), since \(v\leq\sqrt{v(2-v)}\).
It is therefore the required minimum.
Substituting
\(v=1-(\cos\theta_*+\sin\theta_*)/s\) and
\(\sqrt{v(2-v)}=(\cos\theta_*-\sin\theta_*)/s\)
gives \eqref{eq:axis-gap}. Covariance under angle reflections covers the other
half of the angle interval.
\end{proof}

For an interior-axis target state \(\rho=\proj{\Phi_n}\),
\begin{equation}
 \beta=\beta_Q(1-v),\qquad
 F_{\mathrm{ext}}=L(\beta)=1-av,\qquad
 \Xi_n(\rho)=1.
 \label{eq:certificate-not-actual}
\end{equation}
Consequently, \(\tr(\rho T)=0\) alone does not imply
\(\Xi_n(\rho)=L(\beta)\).

\begin{proposition}
\label{prop:endpoint-extractability}
At a single commuting endpoint, let a normalized state supported
on \eqref{eq:endpoint-kernel} have matrix
\[
 \rho=
 \begin{pmatrix}p&z\\\overline z&1-p\end{pmatrix}
 \quad\text{in the basis }
 (\ket{\Phi_n},\Gamma_i\ket{\Phi_n}).
\]
Its exact extractability is
\begin{equation}
 \Xi_n(\rho)
 =\frac12+\sqrt{(p-\tfrac12)^2+(\operatorname{Im}z)^2}.
 \label{eq:endpoint-extractability}
\end{equation}
In particular, it equals \(1/2\) precisely when \(p=1/2\) and
\(\operatorname{Im}z=0\).
\end{proposition}

\begin{proof}
The vectors
\[
 \ket{\psi_\pm}
 \coloneqq\frac{\ket{\Phi_n}\pm\Gamma_i\ket{\Phi_n}}s
\]
are orthogonal product vectors across party \(i\) versus the rest.
Their factors on the rest are also orthogonal, since the local
marginal of \(\Phi_n\) is \(I/2\).
In this basis, the off-diagonal entry of \(\rho\) has modulus
\[
 d_0\coloneqq\sqrt{(p-\tfrac12)^2+(\operatorname{Im}z)^2}.
\]
The partial transpose of \(\rho\) across that cut has trace norm \(1+2d_0\).
Lemma~\ref{lem:negativity} therefore
bounds every product-channel overlap by \(1/2+d_0\).
A local unitary generated by \(\Gamma_i\) can adjust the
relative phase of \(\psi_+\) and \(\psi_-\), and attains this
overlap with \(\Phi_n=(\psi_++\psi_-)/s\).
This proves \eqref{eq:endpoint-extractability}.
When \(d_0=0\), the state is diagonal in the \(\psi_\pm\) basis
and hence separable across this cut.
The formula describes an individual joint qubit sector; the global
classification after local pinching is given in
Theorem~\ref{thm:refined-equality}.
\end{proof}

Finally, take exactly two nonzero \(v\)-coordinates \(u,v\).
Here \(D=[n]\) and \(P_D=B_D=E_D=0\), so \eqref{eq:lambda} reduces to \(1+sQ_D-k^2A_D+s kC_D\), which vanishes when \(u=0\) or \(v=0\) by Proposition~\ref{prop:kernel}; expanding gives
\begin{equation}
 \lambda_{\varnothing,[n]}
 =k^2uv\left(1-\frac{u+v}{s}+\frac{uv}{s}\right).
 \label{eq:two-coordinate-scaling}
\end{equation}
If \(u=v\) tends to zero, the trace of this block tends to \(k\),
and therefore
\[
 \lambda_{\min}(T)=\frac{k}{2}v^2+O(v^3).
\]
All other blocks have a positive gap at the ideal point and
remain bounded away from zero nearby.
Thus no uniform spectral bound for \(T\) can be
linear in distance to the axes in the \(v\)-coordinates.
Since \(v_i=1-\cos\varphi_i\) is quadratic in the angle error,
the fourth-root angular scaling after
\eqref{eq:single-block-stability} is intrinsic to this operator \(T\).

\section{Uniform stability of the original state}
\label{app:original-state-stability}

We prove Theorem~\ref{thm:original-state-stability}.
The strict separation of the qubit curve from the affine line forces a positive weight on the scalar sector at the qutrit party. This weight bounds the total angle error at the remaining parties.
We then control the weight outside the two branch vectors \(G,B\) defined below
and estimate their weights and coherence. All extraction channels
used below are products of local channels on the original spaces.

Both lemmas below are used here and again in
Appendix~\ref{app:further-stability-certification};
Lemma~\ref{lem:two-branch} is also used in
Appendix~\ref{app:local-certification}. The first inequality of
Lemma~\ref{lem:two-branch} is the gentle measurement lemma
\cite{Winter1999} for a projection, with the constant of
\cite{OgawaNagaoka2007}; we include its short proof.

\begin{lemma}
\label{lem:two-branch}
Let \(\rho\) be a state, let \(G,B\) be orthonormal vectors, and put
\[
 P\coloneqq\proj G+\proj B,\quad
 \ell\coloneqq1-\tr(P\rho),\quad
 p_G\coloneqq\bra G\rho\ket G,\quad
 p_B\coloneqq\bra B\rho\ket B,\quad
 c\coloneqq\bra G\rho\ket B.
\]
Then \(\|\rho-P\rho P\|_1\leq2\sqrt\ell\), and for every
\(0\leq t\leq1\),
\begin{equation}
 \|\rho-t\proj G-(1-t)\proj B\|_1
 \leq2\sqrt\ell+|p_G-t|+|p_B-(1-t)|+2|c|.
 \label{eq:two-branch}
\end{equation}
\end{lemma}

\begin{proof}
Purify \(\rho\) and split the purification into its \(P\) and \(I-P\)
parts. The difference between the rank-one operator of the purification
and that of its unnormalized \(P\)-part has trace norm
\(\sqrt{\ell(4-3\ell)}\leq2\sqrt\ell\), and the partial trace contracts
the trace norm of Hermitian operators. This proves the first estimate.
In the basis \(G,B\), the operator
\(P\rho P-t\proj G-(1-t)\proj B\) has diagonal entries \(p_G-t\) and
\(p_B-(1-t)\) and off-diagonal entry \(c\). Its diagonal part has trace
norm \(|p_G-t|+|p_B-(1-t)|\), and its off-diagonal part has trace norm
\(2|c|\). The triangle inequality gives \eqref{eq:two-branch}.
\end{proof}

\begin{lemma}
\label{lem:branch-coherence}
Let \(n\geq3\). Let \(e_0,e_1,e_2\) be an orthonormal basis of
\(\mathbb C^3\) and \(f_0,f_1\) orthonormal vectors of
\((\mathbb C^2)^{\otimes(n-1)}\), and put
\[
 \ket G\coloneqq\frac{\ket{e_0f_0}+\ket{e_1f_1}}{\sqrt2},\qquad
 \ket B\coloneqq\ket{e_2f_0},\qquad
 K_0\coloneqq\proj G+\frac12\proj B.
\]
Suppose that a product of local unitaries maps \(\ket G\), regarded as a
vector of
\(\operatorname{span}\{e_0,e_1\}\otimes(\mathbb C^2)^{\otimes(n-1)}\),
to \(\ket{\GHZ_n}\). Then every state \(\rho\) on
\(\mathbb C^3\otimes(\mathbb C^2)^{\otimes(n-1)}\) satisfies
\(F_0\coloneqq\tr(\rho K_0)\leq\Xi_n(\rho)\). With
\(g\coloneqq\Xi_n(\rho)-F_0\) and \(\ell,p_B,c\) as in
Lemma~\ref{lem:two-branch},
\begin{equation}
 |c|\leq\sqrt{2g(1+g)}+\sqrt{\ell p_B}
     \leq2\sqrt g+\sqrt\ell.
 \label{eq:branch-coherence}
\end{equation}
\end{lemma}

\begin{proof}
Identify the output qubit of the first party with
\(\operatorname{span}\{e_0,e_1\}\). For \(|\zeta|\leq1\), put
\(\alpha\coloneqq\sqrt{1-|\zeta|^2}\) and take the first-party Kraus
operators of \eqref{eq:coherent-channel},
\[
 M_0\coloneqq\ket{e_0}\!\bra{e_0}
       +\alpha\ket{e_1}\!\bra{e_1},\qquad
 M_1\coloneqq\zeta\ket{e_1}\!\bra{e_1}
       +\ket{e_0}\!\bra{e_2}.
\]
The two terms of \(M_1\) have orthogonal output vectors, and
\(\alpha^2+|\zeta|^2=1\), so these operators satisfy
\(M_0^*M_0+M_1^*M_1=I_3\). Every other party uses the identity, and the
product of local unitaries in the hypothesis is applied afterwards.
Put \(\ket{u_i}\coloneqq\ket{e_if_i}\) for \(i=0,1\). The pullback
\(K_\zeta\) of the target projector under this extraction is supported
on \(\operatorname{span}\{\ket{u_0},\ket{u_1},\ket B\}\). In this
ordered basis, it is
\[
 K_\zeta=\frac12
 \begin{pmatrix}
 1&\alpha&0\\
 \alpha&1&\overline\zeta\\
 0&\zeta&1
 \end{pmatrix},
\]
with zero action on the orthogonal complement. At \(\zeta=0\) it equals
\(K_0\), which proves \(F_0\leq\Xi_n(\rho)\).

Put \(r\coloneqq\operatorname{Re}\bra{u_0}\rho\ket{u_1}\) and
\(w\coloneqq|\bra{u_1}\rho\ket B|\); positivity and unit trace give
\(r\leq1/2\).
For fixed \(|\zeta|\), optimizing the phase of \(\zeta\)
gives overlap gain
\((\sqrt{1-|\zeta|^2}-1)r+|\zeta|w\) above \(F_0\).
For \(r\geq0\) its maximum over \(|\zeta|\) is \(\sqrt{r^2+w^2}-r\);
for \(r<0\), taking \(|\zeta|=1\) gives at least \(w\).
With \(r_+\coloneqq\max\{r,0\}\leq1/2\), it follows that
\[
 g\geq\sqrt{r_+^2+w^2}-r_+,
 \qquad w^2\leq g(1+g).
\]
The unit vector \(\ket{G_-}\coloneqq(\ket{u_0}-\ket{u_1})/\sqrt2\) is
orthogonal to \(G\) and \(B\), and
\(\ket G=\sqrt2\ket{u_1}+\ket{G_-}\). Positivity of \(\rho\) therefore
yields
\[
 |c|\leq\sqrt2\,w+\sqrt{\bra{G_-}\rho\ket{G_-}\,p_B}
     \leq\sqrt{2g(1+g)}+\sqrt{\ell p_B}.
\]
The last inequality in \eqref{eq:branch-coherence} follows from
\(g\leq1\) and \(p_B\leq1\).
\end{proof}

\begin{proof}[Proof of Theorem~\ref{thm:original-state-stability}]
We first prove the assertion for projective observables, with
\(\eta\leq t\leq1-\eta\). The passage from binary POVMs to such
observables on the same spaces is given at the end of the proof.
Write \(\widehat W\coloneqq W_n/\beta_Q\), so
\(b=\tr(\rho\widehat W)\) and \(L(\beta)=ab-1/s\).
The constants \(s,a,\kappa,c_*\) retain their definitions in
\eqref{eq:constants} and \eqref{eq:uniform-constant}.
The smallness requirements on \(\varepsilon\) will be collected
in an explicit positive threshold below.

\paragraph{Sector weights and angle control.}
By Lemma~\ref{lem:bell-spectrum}, a commuting projective pair at any party limits \(b\) to \(1/s\). Since \(b>1/s\), every local pair, including the qutrit pair, is noncommuting.
The qutrit pair therefore has one irreducible two-dimensional
sector (its qubit sector) and one common one-dimensional eigenspace
(its scalar sector).
Let \(q_E\) be the probability of the latter sector in \(\rho\).
The Bell operator preserves these sectors.
Writing \(b_Q,b_E\) for their conditional normalized scores, we have
\[
 b=(1-q_E)b_Q+q_Eb_E,\qquad b_E\leq\frac1s.
\]
The case \(q_E=1\) is excluded by \(b>1/s\), and \(b_Q\geq b\).
An extraction for the conditional qubit state extends to the
whole input by pinching the first party's two sectors and
preparing any fixed output on its scalar sector. Every other
party uses the same chosen channel in both sectors.
The scalar contribution to the overlap is nonnegative.
Taking the supremum over the qubit extraction gives
\[
 \Xi_n(\rho)\geq(1-q_E)F(b_Q)\geq(1-q_E)F(b)\geq F(b)-q_E.
\]
This bound also applies if \(q_E=0\).

By \eqref{eq:strict-gap}, the qubit curve is strictly above the affine line in the interior. Its concavity gives the positive constant
\[
 m_\eta\coloneqq
 \min\left\{
 F\!\left(\frac1s+\kappa\eta\right)-\frac{1+\eta}{2},
 F\!\left(\frac1s+\kappa(1-\eta)\right)-\frac{2-\eta}{2}
 \right\}>0.
\]
It is a lower bound for \(F(b)-L(\beta)\) throughout the stated
score interval. Thus, if \(\varepsilon\leq m_\eta/2\),
\begin{equation}
 q_E\geq m_\eta-\varepsilon\geq\frac{m_\eta}{2}.
 \label{eq:state-scalar-weight}
\end{equation}
The pinching in this argument is part of a permitted extraction; the state \(\rho\) keeps all its original coherences.

Put every qubit pair in the canonical coordinates
\eqref{eq:canonical}. Its variables \(v_j,p_j,q_j,x_j\) retain
their earlier meanings. The scalar sector at the first party
has observable values \(a_0,a_1\in\{+1,-1\}\).
If \(a_0=a_1\), embed it into the \(X\)-eigenvector of
eigenvalue \(a_0\) at angle \(0\); if \(a_1=-a_0\), use
the \(Z\)-eigenvector of eigenvalue \(a_0\) at angle \(\pi/2\).
Its first coordinate is consequently \(x_1=1\).
The remaining angle variables are shared by both first-party
sectors.

Let \(K\) be the pullback of the local extraction from
Section~\ref{sec:bound} to this physical input space, including
the local pinching at the first party, and set
\[
 T\coloneqq K-a\widehat W+\frac1s I.
\]
The operator is positive on both first-party sectors, and
\[
 0\leq\tr(\rho T)\leq\Xi_n(\rho)-L(\beta)\leq\varepsilon.
\]
In the scalar sector the auxiliary probability is exactly
\[
 \mathsf P_2(1,x_2,\ldots,x_n)
 =1-\prod_{j=2}^n(1-x_j).
\]
Theorem~\ref{thm:uniform-gap} therefore implies
\[
 1-\prod_{j=2}^n(1-x_j)
 \leq \frac{\varepsilon}{q_Ec_*}.
\]
If \(\varepsilon\leq c_*m_\eta/4\), the right side is at most
\(1/2\), by \eqref{eq:state-scalar-weight}. Since
\(-\log(1-u)\leq2u\) for \(0\leq u\leq1/2\), we obtain
\[
 \sum_{j=2}^n x_j
 \leq-\log\prod_{j=2}^n(1-x_j)
 \leq \frac{2\varepsilon}{q_Ec_*}.
\]
Define
\begin{equation}
 \nu\coloneqq\sum_{j=2}^n v_j,\qquad
 C_{\mathrm{ang}}\coloneqq\frac{4\kappa}{c_*m_\eta}.
 \quad\text{Then}\quad
 \nu\leq C_{\mathrm{ang}}\varepsilon.
 \label{eq:state-angle-sum}
\end{equation}
Bounding the sum keeps the subsequent estimates independent of \(n\).

\paragraph{Concentration on the two branches.}
We first bound the state's weight outside the subspace determined by the two rest strings of \(\Phi_n\), its two \(Y\)-strings on parties \(2,\ldots,n\). Inside it, we compare with
ideal rest angles to isolate the two attaining branches.
The first angle must stay away from its commuting endpoints for
the resulting axis gap to be uniform.
Consider the six-dimensional subspace
\[
 \mathcal S\coloneqq\mathbb C^3\otimes
 \operatorname{span}\left\{
   \bigotimes_{j=2}^n\ket{\tau_j}_Y,\,
   \bigotimes_{j=2}^n\ket{-\tau_j}_Y
 \right\},
\]
and let \(\Pi\) be its orthogonal projection.
Both \(\widehat W\) and \(K\) preserve this subspace and its
orthogonal complement. On the first party's qubit sector,
this follows from the complementary \(Y\)-string blocks.
For \(K\) on the full angle cube, both dephasing Paulis \(X,Z\) flip the local \(Y\) sign, so every product of them
maps \(\Phi_n\) into a complementary-string block.
On the scalar sector, use its endpoint qubit embedding.
At that endpoint the Bell operator and the dephasing channel
preserve the first party's indicated Pauli eigenspace.
Compression to this eigenspace couples only a rest string
and its complement, and therefore also preserves \(\mathcal S\).

Let \(\widehat W_{\mathrm{ax}}\) be the normalized Bell operator
obtained by setting every angle \(j\geq2\) to \(\pi/4\), keeping
the first party's two sectors and scalar signs fixed.
Let \(K_{\mathrm{ax}}\) be obtained from \(K\) by replacing the corresponding rest extraction channels by identities, and set
\[
 T_{\mathrm{ax}}
 \coloneqq K_{\mathrm{ax}}-a\widehat W_{\mathrm{ax}}+\frac1s I.
\]
The required comparison consists of three bounds:
\begin{align}
 \|(I-\Pi)\widehat W(I-\Pi)\|_\infty
 &\leq\sqrt{2\nu},\nonumber\\
 \|\Pi(\widehat W-\widehat W_{\mathrm{ax}})\Pi\|_\infty
 &\leq2\nu,\label{eq:state-operator-comparison}\\
 \|K-K_{\mathrm{ax}}\|_\infty
 &\leq2a\nu.\nonumber
\end{align}
We verify them directly from Lemma~\ref{lem:blocks} and \eqref{eq:channel}.

For a complementary block outside \(\mathcal S\), both index
sets \(H,D\) contain a rest index. Choose \(i\in H\), \(j\in D\)
with \(i,j\geq2\). In Lemma~\ref{lem:blocks},
\[
 x^2=B_HA_D\leq b_i,\qquad
 y^2=A_HB_D\leq b_j.
\]
Thus the squared norm of its normalized Bell block is at most
\[
 x^2+y^2\leq b_i+b_j\leq2\nu.
\]
As stated in Appendix~\ref{app:covariance}, the full-cube
formula uses the signed factors
\((\cos\theta_j-\sin\theta_j)/s\) in place of
\(\sqrt{b_j}\). These signs do not change the squared norm.
The same bound survives compression to the physical scalar
sector, since this compression has norm one and commutes
with the projection onto the rest strings.

Within \(\mathcal S\), there are two qubit Bell blocks.
Their entries consist of a first-party factor times
\(\prod_{j\geq2}(1-v_j)\), and a second first-party factor
times \(\prod_{j\geq2}(\cos\theta_j-\sin\theta_j)/s\).
The first-party factors have modulus at most one.
Making the rest angles ideal replaces the two products by
one and zero. Since there are at least two rest indices,
\[
 1-\prod_{j=2}^n(1-v_j)\leq\nu,\qquad
 \left|\prod_{j=2}^n
       \frac{\cos\theta_j-\sin\theta_j}{s}\right|
 \leq\sqrt{b_2b_3}
 \leq\frac{b_2+b_3}{2}\leq\nu.
\]
Each block difference therefore has norm at most \(2\nu\).
Again, compression to the scalar sector preserves the bound.
This proves the second estimate on the full angle cube.

For the third estimate, the rest channel at party \(j\) is
\((1-p_j)\operatorname{Id}+p_j\operatorname{Ad}_{\Gamma_j}\).
Its difference from the identity changes a norm-one operator
by at most \(2p_j\). Telescope over the rest channels.
Every intervening unital channel contracts operator norm,
as does the adjoint of the first-party extraction channel.
Thus the total change is at most
\(2\sum_{j\geq2}p_j=2a\nu\).
It follows that
\begin{equation}
 \|\Pi(T-T_{\mathrm{ax}})\Pi\|_\infty\leq4a\nu.
 \label{eq:state-axis-comparison}
\end{equation}

The first bound in \eqref{eq:state-operator-comparison},
applied with all rest angles ideal, shows that the axis Bell
operator vanishes on \(\mathcal S^\perp\).
On \(\mathcal S\), its nonnegative eigenvalues on the first-party qubit sector are
\(1-v_1\) and \(\sqrt{v_1(2-v_1)}\), and its eigenvalues on the scalar
sector are \(\pm1/s\).
Since \(v_1\leq\kappa\), its norm is \(1-v_1\).
Assume for the moment that
\[
 \nu\leq\frac1{16a^2},\qquad
 \nu\leq\frac{\kappa\eta}{4}.
\]
The first inequality implies \(a\sqrt{2\nu}\leq1/(2s)\), hence
\(\sqrt{2\nu}\leq1/s\leq1-v_1\).
Equation~\eqref{eq:state-operator-comparison} then gives
\[
 b\leq\|\widehat W\|_\infty\leq1-v_1+2\nu.
\]
Using \(b\geq1/s+\kappa\eta\), we obtain
\[
 v_1\leq\kappa(1-\eta)+2\nu
       \leq v_\eta\coloneqq\kappa(1-\eta/2)<\kappa.
\]
The first angle is consequently bounded away from its commuting
endpoints by a bound depending only on \(\eta\).

Since \(K\geq0\), the first estimate in
\eqref{eq:state-operator-comparison} gives, on
\(\mathcal S^\perp\),
\[
 T\geq\left(\frac1s-a\sqrt{2\nu}\right)I
       \geq\frac1{2s}I.
\]
Both summands in the orthogonal decomposition of
\(\tr(\rho T)\) are nonnegative, and hence
\[
 \tr\bigl(\rho(I-\Pi)\bigr)\leq2s\varepsilon.
\]

Let \(J\) identify the first qubit with its physical
two-dimensional sector, and put \(\ket G\coloneqq J\ket{\Phi_n}\).
Choose a unit vector \(\ket{e_2}\) for the physical scalar sector.
Let \(\ket{e_0}\) be the \(X\)- or \(Z\)-eigenvector specified
by its endpoint embedding, and choose \(\ket{e_1}\) orthogonal
to it in the canonical first qubit. Define
\[
 \ket{f_i}\coloneqq\sqrt2\,{}_1\!\braket{e_i}{\Phi_n}
 \quad(i=0,1),\qquad
 \ket B\coloneqq\ket{e_2f_0},\qquad P\coloneqq\proj G+\proj B.
\]
The first marginal of \(\Phi_n\) is \(I/2\), so the two rest
vectors are orthonormal and
\[
 \ket G=\frac{\ket{e_0f_0}+\ket{e_1f_1}}{\sqrt2}.
\]
Both \(G\) and \(B\) belong to \(\mathcal S\). They play the roles of
the vectors \(G,B\) of \eqref{eq:coherent-qutrit}, written in the
coordinates of the given strategy; Lemma~\ref{lem:branch-alignment}
relates the two pairs at the end of the proof.

On the qubit sector, \(T_{\mathrm{ax}}\ket G=0\), since its
eigenvalue there is
\(1-av_1-a(1-v_1)+1/s=0\).
Proposition~\ref{prop:axis-gap} gives a lower bound on
the orthogonal complement:
\[
 d_\eta\coloneqq\frac1s+
 a\left(v_\eta-\sqrt{v_\eta(2-v_\eta)}\right)>0.
\]
Indeed, the displayed axis-gap function decreases on
\([0,\kappa]\), is positive before \(\kappa\), and vanishes
there. On the physical scalar sector within \(\mathcal S\),
\(K_{\mathrm{ax}}=\proj B/2\) and the eigenvector of \(\widehat W_{\mathrm{ax}}\) with eigenvalue \(1/s\) is \(B\), by \eqref{eq:refined-endpoint}.
The two eigenvalues of \(T_{\mathrm{ax}}\) in this sector
are therefore \(0\) and \((a+1)/s\).
Consequently
\[
 T_{\mathrm{ax}}\big|_{\mathcal S}\geq d_\eta(\Pi-P).
\]
Equation~\eqref{eq:state-axis-comparison}, positivity of \(T\),
and the angle estimate imply
\[
 \tr\bigl(\rho(\Pi-P)\bigr)
 \leq\frac{\varepsilon+4a\nu}{d_\eta}.
\]
Combining the two parts of the orthogonal complement of \(P\)
gives
\begin{equation}
 \ell_\perp\coloneqq1-\tr(P\rho)
 \leq C_{\mathrm{leak}}\varepsilon,\qquad
 C_{\mathrm{leak}}
 \coloneqq2s+\frac{1+4aC_{\mathrm{ang}}}{d_\eta}.
 \label{eq:state-leakage}
\end{equation}

\paragraph{Branch weights and coherence.}
The first party can act as the identity on its qubit sector, map \(e_2\) to \(e_0\), and discard the coherences between these sectors.
With the identity at every other party, the pullback of the target projector to the full input space is exactly
\[
 K_0\coloneqq\proj G+\frac12\proj B.
\]
Compose with the fixed output unitaries. This is an allowed extraction, so
\(F_0\coloneqq\tr(\rho K_0)\leq\Xi_n(\rho)\).
Write
\[
 p_G\coloneqq\bra G\rho\ket G,\qquad
 p_B\coloneqq\bra B\rho\ket B,\qquad
 c\coloneqq\bra G\rho\ket B,\qquad
 p_G+p_B=1-\ell_\perp.
\]

The axis Bell operator has eigenvalues \(1-v_1\) and \(1/s\)
on \(G,B\), respectively, and no matrix elements between \(P\)
and its orthogonal complement. Its norm on the orthogonal complement of \(P\) is at most one. The operator comparisons therefore yield
\[
 b=(1-v_1)p_G+\frac{p_B}{s}+\xi,\qquad
 |\xi|\leq2\nu+2\ell_\perp
          \leq C_{\mathrm{err}}\varepsilon,\qquad
 C_{\mathrm{err}}\coloneqq2(C_{\mathrm{ang}}+C_{\mathrm{leak}}).
\]
Here one of the two terms \(\ell_\perp\) bounds the expectation of the
axis operator on \(\mathcal S\cap P^\perp\), and the other bounds
the expectation of the actual Bell operator on
\(\mathcal S^\perp\). The error within \(\mathcal S\),
including all its cross terms, is at most \(2\nu\).

Using \(a/s=1/2+1/s\), the preceding score identity gives
\begin{equation}
 F_0-L(\beta)
 =\frac{\ell_\perp}{s}+av_1p_G-a\xi.
 \label{eq:state-baseline-gain}
\end{equation}
Since \(F_0\leq\Xi_n(\rho)\leq L(\beta)+\varepsilon\),
\[
 av_1p_G\leq C_{\mathrm{gain}}\varepsilon,\qquad
 0\leq\Xi_n(\rho)-F_0\leq C_{\mathrm{gain}}\varepsilon,
 \qquad C_{\mathrm{gain}}\coloneqq1+aC_{\mathrm{err}}.
\]
Also \(b=1/s+\kappa t\), so
\[
 \kappa(t-p_G)=-v_1p_G-\frac{\ell_\perp}{s}+\xi.
\]
It follows that
\begin{equation}
 |p_G-t|\leq C_{\mathrm{wt}}\varepsilon,\qquad
 C_{\mathrm{wt}}
 \coloneqq\frac{C_{\mathrm{gain}}/a+
        C_{\mathrm{leak}}/s+C_{\mathrm{err}}}{\kappa}.
 \label{eq:state-branch-weight}
\end{equation}

The local channel family of \eqref{eq:coherent-channel} controls the
coherence directly on \(\rho\). The vector \(\ket G=J\ket{\Phi_n}\) is
locally unitarily equivalent to \(\ket{\GHZ_n}\), so
Lemma~\ref{lem:branch-coherence} applies with the vectors \(e_i,f_i\)
chosen above; its operator \(K_0\) is the one just defined, and its
\(\ell\) is \(\ell_\perp\). Hence
\begin{align}
 |c|
 &\leq2\sqrt{\Xi_n(\rho)-F_0}+\sqrt{\ell_\perp}
 \leq C_{\mathrm{coh}}\sqrt{\varepsilon},
 \label{eq:state-coherence}\\
 C_{\mathrm{coh}}
 &\coloneqq2\sqrt{C_{\mathrm{gain}}}+\sqrt{C_{\mathrm{leak}}}.
 \nonumber
\end{align}

\paragraph{Trace norm, constants, and binary POVMs.}
Lemma~\ref{lem:two-branch} applies to \(G,B\), with
\(\ell=\ell_\perp\). Since \(p_G+p_B=1-\ell_\perp\), we have
\(|p_B-(1-t)|\leq|p_G-t|+\ell_\perp\). Therefore, for
\(\varepsilon\leq1\),
\begin{align}
 \|\rho-t\proj G-(1-t)\proj B\|_1
 &\leq2\sqrt{\ell_\perp}+2|c|
           +2|p_G-t|+\ell_\perp\nonumber\\
 &\leq(2\sqrt{C_{\mathrm{leak}}}+2C_{\mathrm{coh}})
          \sqrt{\varepsilon}
       +(2C_{\mathrm{wt}}+C_{\mathrm{leak}})
          \varepsilon\nonumber\\
 &\leq C_\eta^{\mathrm{proj}}\sqrt{\varepsilon},
 \label{eq:state-distance-projective}\\
 C_\eta^{\mathrm{proj}}
 &\coloneqq2\sqrt{C_{\mathrm{leak}}}+2C_{\mathrm{coh}}
       +2C_{\mathrm{wt}}+C_{\mathrm{leak}}.
 \nonumber
\end{align}

Lemma~\ref{lem:branch-alignment}, with \(i=1\), \(J'=J\), the present
\(\ket{e_2}\) and \(\ket g=\ket{e_0}\), maps the present \(G,B\)
to the vectors \(G,B\) of \eqref{eq:coherent-qutrit} by one qutrit
unitary and individual qubit unitaries. Trace-norm invariance in
\eqref{eq:state-distance-projective} therefore gives the
required distance \(D_t(\rho)\) for projective strategies.

All smallness conditions imposed in this argument hold when
\begin{equation}
 \varepsilon\leq\varepsilon_\eta^{\mathrm{proj}}
 \coloneqq\min\left\{
 1,\frac{m_\eta}{2},\frac{c_*m_\eta}{4},
 \frac1{16a^2C_{\mathrm{ang}}},
 \frac{\kappa\eta}{4C_{\mathrm{ang}}}
 \right\}.
 \label{eq:state-constants}
\end{equation}
Every constant in this expression and in
\(C_\eta^{\mathrm{proj}}\) was defined above in terms of
\(\eta\) and the constants \(s,a,\kappa,c_*\); each is positive
and independent of \(n\).
The proof also includes \(\varepsilon=0\): the same inequalities
give \(\nu=\ell_\perp=0\), \(p_G=t\), and \(c=0\).

Finally, for an arbitrary binary-POVM strategy with state \(\rho\),
let \(\beta_*\) and \(t_*\) be the score and mixture parameter of the
reflections given by Lemma~\ref{lem:projective-replacement} on the
same state and local spaces. By that lemma,
\begin{equation}
 0\leq t_*-t\leq2\varepsilon,\qquad
 0\leq\Xi_n(\rho)-L(\beta_*)\leq\varepsilon.
 \label{eq:state-povm-reduction}
\end{equation}
If \(\varepsilon\leq\eta/4\), then
\(\eta/2\leq t_*\leq1-\eta/2\).
Apply the projective estimate with interior parameter
\(\eta/2\) to this same \(\rho\).
The canonical branch vectors are independent of the mixture
parameter, so
\[
 \|\omega_{t_*}-\omega_t\|_1
 =2|t_*-t|\leq4\varepsilon.
\]
Consequently the theorem holds with the explicit choices
\[
 \varepsilon_\eta
 \coloneqq\min\{\varepsilon_{\eta/2}^{\mathrm{proj}},\,\eta/4\},
 \qquad
 C_\eta\coloneqq C_{\eta/2}^{\mathrm{proj}}+4.
\]
Keeping the two powers of \(\varepsilon\) separate in
\eqref{eq:state-distance-projective} also proves
\eqref{eq:original-state-stability-separated}.
To specify its coefficients, evaluate
\(C_{\mathrm{leak}},C_{\mathrm{coh}},C_{\mathrm{wt}}\)
at the interior parameter \(\eta/2\), and set
\[
 A_\eta\coloneqq2\sqrt{C_{\mathrm{leak}}}+2C_{\mathrm{coh}},
 \qquad
 B_\eta\coloneqq2C_{\mathrm{wt}}+C_{\mathrm{leak}}+4.
\]
The final \(4\) in \(B_\eta\) accounts for the change from
\(t_*\) to \(t\). With these choices,
\(C_\eta=A_\eta+B_\eta\).
The cap of \(2\) follows because the trace-norm distance between
two states is at most \(2\).
\end{proof}

\paragraph{Measurement angles.}
The forced weight of the scalar sector, in which \(x_1=1\), makes
\(\mathsf P_2\) of first order in the remaining \(x_j\), and the gain
estimate controls \(v_1\) directly. This gives order \(\sqrt\varepsilon\)
for the angles here, compared with order \(\varepsilon^{1/4}\) for a
single sector after \eqref{eq:single-block-stability}.
For the projective strategies in the proof of Theorem~\ref{thm:original-state-stability}, let \(\theta_1\)
be the canonical angle of the qutrit's two-dimensional
sector, and let \(\theta_j\), \(j\geq2\), be the angles of the
other parties. The constants below are evaluated at \(\eta\). If
\[
 \varepsilon\leq\min\left\{
 \varepsilon_\eta^{\mathrm{proj}},\frac{\eta}{2C_{\mathrm{wt}}}
 \right\},
\]
then \eqref{eq:state-branch-weight} and \(t\geq\eta\) give
\(p_G\geq\eta/2\). Together with
\(av_1p_G\leq C_{\mathrm{gain}}\varepsilon\) and
\eqref{eq:state-angle-sum}, this yields
\[
 \sum_{j=1}^n v_j
 \leq\left(C_{\mathrm{ang}}+
       \frac{2C_{\mathrm{gain}}}{a\eta}\right)\varepsilon.
\]
Since \(v_j=1-\cos|\theta_j-\pi/4|\) and
\(v_j\geq(16\kappa/\pi^2)|\theta_j-\pi/4|^2\),
\[
 \left(\sum_{j=1}^n|\theta_j-\pi/4|^2\right)^{1/2}
 \leq\frac{\pi}{4\sqrt\kappa}
 \sqrt{\left(C_{\mathrm{ang}}+
       \frac{2C_{\mathrm{gain}}}{a\eta}\right)\varepsilon}.
\]
The constants are independent of \(n\). For binary POVMs,
the same conclusion holds for every projective replacement as in
Lemma~\ref{lem:projective-replacement}, which maximizes the score on
the same state and local spaces, with
\(\eta\) replaced throughout by \(\eta/2\) and with the
additional requirement \(\varepsilon\leq\eta/4\), by
\eqref{eq:state-povm-reduction}. The controlled angles belong to
this projective replacement.

\begin{proposition}
\label{prop:fixed-observable-stability}
Let \(n\geq3\), and let the parties use exactly the observables of
Theorem~\ref{thm:qutrit-attainer} on
\(\mathbb C^3\otimes(\mathbb C^2)^{\otimes(n-1)}\), in the coordinates
of that theorem. Let \(\rho\) be a state on this space whose score
\(\beta=\tr(\rho W_n)\) for these observables lies in
\([\beta_B,\beta_Q]\). With
\(t=(\beta-\beta_B)/(\beta_Q-\beta_B)\), let \(\omega_t\)
be the state \eqref{eq:qutrit-attainer}. If \(\varepsilon\geq0\) and
\(\Xi_n(\rho)-L(\beta)\leq\varepsilon\), then
\[
 \|\rho-\omega_t\|_1\leq
 \min\{2,\,4\sqrt{1+\sqrt2}\,\sqrt{\varepsilon}+(4+3\sqrt2)\,\varepsilon\}.
\]
Numerically,
\(4\sqrt{1+\sqrt2}=6.2150958961\ldots\) and
\(4+3\sqrt2=8.2426406871\ldots\), so the right side is below \(2\)
exactly when \(\varepsilon<\varepsilon_*\approx0.0591922709\).
\end{proposition}

\begin{proof}
Write \(\widehat W\coloneqq W_n/\beta_Q\), so that
\(L(\beta)=a\tr(\rho\widehat W)-1/s\). Let \(G\) and \(B\) be the
vectors of Proposition~\ref{prop:coherent-qutrit}. The first-party
channel of Theorem~\ref{thm:qutrit-attainer}, with the identity
elsewhere and the fixed output unitaries, has pullback
\(K_0=\proj G+\frac12\proj B\), so
\(F_0\coloneqq\tr(\rho K_0)\leq\Xi_n(\rho)\).

Both first-party observables are block diagonal in the two sectors.
With all angles ideal, Lemma~\ref{lem:blocks} shows that on the qubit
sector \(\widehat W\) has the simple eigenvalues \(1\) (eigenvector
\(G\)) and \(-1\), and \(0\) otherwise. On the scalar sector,
\(\widehat W\) is the compression of the Bell operator at the
commuting endpoint \(\theta_1=0\) to the eigenspace \(X_1=+1\). By
Lemma~\ref{lem:blocks} and the proof of
Theorem~\ref{thm:qutrit-attainer}, it has the simple eigenvalues
\(1/s\) (eigenvector \(B\)) and \(-1/s\), and \(0\) otherwise.
Hence \(T_0\coloneqq K_0-a\widehat W+s^{-1}I\) vanishes exactly on
\(\operatorname{span}\{G,B\}\) and satisfies \(T_0\geq s^{-1}(I-P)\),
where \(P\coloneqq\proj G+\proj B\).

Let \(p_G\), \(p_B\), \(c\) and \(\ell_\perp\coloneqq\ell\) be the
quantities of Lemma~\ref{lem:two-branch} for the present \(G\), \(B\)
and \(P\), and put \(\gamma\coloneqq F_0-L(\beta)\). Then
\(\gamma=\tr(\rho T_0)\geq\ell_\perp/s\geq0\),
\(\gamma\leq\varepsilon\) and
\(\Xi_n(\rho)-F_0\leq\varepsilon-\gamma\). The relations
\(F_0=p_G+p_B/2\) and \(p_G+p_B=1-\ell_\perp\) give
\(p_G-t=\ell_\perp+2\gamma\) and
\(p_B-(1-t)=-2(\ell_\perp+\gamma)\).

Since \(\ket G=J\ket{\Phi_n}\) is locally unitarily equivalent to
\(\ket{\GHZ_n}\), Lemma~\ref{lem:branch-coherence} applies with
\(\ket{e_0}=\ket{x+}\), \(\ket{e_1}=\ket{x-}\), \(\ket{e_2}=\ket2\),
\(\ket{f_0}=\sqrt2\,{}_1\!\braket{x+}{\Phi_n}\) and
\(\ket{f_1}=\sqrt2\,{}_1\!\braket{x-}{\Phi_n}\); as in the proof of
Proposition~\ref{prop:coherent-qutrit}, these give
\(\ket G=(\ket{e_0f_0}+\ket{e_1f_1})/s\) and \(\ket B=\ket{e_2f_0}\).
The lemma gives
\[
 |c|\leq2\sqrt{\Xi_n(\rho)-F_0}+\sqrt{\ell_\perp}
 \leq2\sqrt{\varepsilon-\gamma}+\sqrt{\ell_\perp}.
\]
Finally, combine Lemma~\ref{lem:two-branch} and Cauchy--Schwarz with
\(\ell_\perp\leq s\gamma\):
\[
 \|\rho-\omega_t\|_1
 \leq4\sqrt{\ell_\perp}+4\sqrt{\varepsilon-\gamma}
   +3\ell_\perp+4\gamma
 \leq4\sqrt{1+s}\,\sqrt{\varepsilon}+(4+3s)\,\varepsilon.
\]
The bound \(2\) is the trace-norm diameter of the state space.
\end{proof}

\begin{remark}
\label{rem:negativity-certificate}
For a state \(\rho\) on finite-dimensional local spaces, let
\(\mathcal N_1(\rho)\coloneqq(\|\rho^{T_1}\|_1-1)/2\) be its negativity
\cite{VidalWerner2002} across the cut separating party~\(1\) from the
others. Lemma~\ref{lem:negativity} gives
\(\Xi_n(\rho)\leq\frac12\|\rho^{T_1}\|_1=\frac12+\mathcal N_1(\rho)\).
With Theorem~\ref{thm:main}, every finite-dimensional strategy of score
\(\beta\in[\beta_B,\beta_Q]\) therefore satisfies
\[
 0\leq\Xi_n(\rho)-L(\beta)\leq\mathcal N_1(\rho)-\frac t2,
 \qquad\text{in particular}\qquad
 \mathcal N_1(\rho)\geq\frac t2 .
\]
The attaining mixture \(\omega_t\), whose qutrit is at party~\(1\), has
\(\|\omega_t^{T_1}\|_1=1+t\) (the case \(c=0\) in the proof of
Proposition~\ref{prop:coherent-qutrit}), so \(t/2\) is the least such
negativity at every score; across any other single party the
negativity of \(\omega_t\) equals \(1/2\).
Consequently, with the same other hypotheses and constants, the
hypothesis \(\Xi_n(\rho)-L(\beta)\leq\varepsilon\) of
Theorems~\ref{thm:actual-stability} and~\ref{thm:original-state-stability}
may be replaced by \(\mathcal N_1(\rho)-t/2\leq\varepsilon\), and the
hypothesis \(\Xi_n(\rho)=L(\beta)\) of
Theorems~\ref{thm:necessary-equality}, \ref{thm:refined-equality}
(its necessary conditions) and~\ref{thm:minimal-equality-state} by
\(\mathcal N_1(\rho)=t/2\). The same replacement applies to Proposition~\ref{prop:fixed-observable-stability}. If the qutrit is placed at party~\(i\), as in
Remark~\ref{rem:qutrit-position}, the
negativity across party~\(i\) is used instead.
These hypotheses are stronger: \(\omega_t\) with its qutrit
moved to party~\(2\) satisfies \(\Xi_n=L(\beta)\) but has
\(\mathcal N_1=1/2\), and the converse statements, which conclude
\(\Xi_n=L(\beta)\), do not give \(\mathcal N_1=t/2\).
On the coherent family,
\(\mathcal N_1(\omega_{t,c})-t/2=\Xi_n(\omega_{t,c})-(1+t)/2\), so
\eqref{eq:coherent-distance} and Proposition~\ref{prop:stability-exponent}
hold with this negativity excess, and the exponent \(1/2\) remains
optimal.
The negativity excess can be computed from a description of \(\rho\),
or bounded through any upper bound on \(\|\rho^{T_1}\|_1\). Like the extractability excess, it is
not determined by the Bell statistics: those of \(\omega_{t,c}\) with the
fixed measurements do not depend on \(c\).
\end{remark}

\section{Stability near the quantum maximum}
\label{app:further-stability-certification}

The following spectral facts will be used for both stability and
certification on one qutrit and qubits. The sum of squared eigenvalues
and its bound on the second eigenvalue are due to Scarani and
Gisin~\cite[arXiv:quant-ph/0103068v1, Appendix~B, Eqs.~(31)--(32)]{ScaraniGisin2001}; we use
the normalization of \eqref{eq:spectral-sum} and describe the scalar
sector and simultaneous branch alignment explicitly.

\begin{lemma}
\label{lem:qutrit-spectral}
Let \(n\geq3\), and let a projective binary strategy on
\(\mathbb C^3\otimes(\mathbb C^2)^{\otimes(n-1)}\) have normalized
score \(b>1/s\). Then the first-party pair has a noncommuting two-dimensional
sector with projector \(Q_1\) and a one-dimensional complement
spanned by a unit vector \(e_2\). Put
\(Q\coloneqq Q_1\otimes I\), \(E\coloneqq I-Q\), and
\(\widehat W\coloneqq W_n/\beta_Q\). Its largest positive eigenvalues on
these two sectors are \(\lambda\) and \(\mu\), with normalized
eigenvectors \(G'\) and \(B'\), respectively, and
\[
 b\leq\lambda\leq1,\qquad 0<\mu\leq1/s.
\]
Every other positive eigenvalue on \(Q\) is at most
\(\sqrt{1-\lambda^2}\leq\sqrt{1-b^2}\); on \(E\) it is at most
\(\sqrt{(1-\lambda^2)/2}\leq\sqrt{(1-b^2)/2}\).
One product of individual local unitaries sends \(G',B'\) to the
canonical vectors \(G,B\) of Proposition~\ref{prop:coherent-qutrit}.
There are also an orthonormal basis \(e'_0,e'_1\) of
\(\operatorname{ran}Q_1\) and orthonormal rest states \(f_0,f_1\)
such that
\begin{equation}
 \ket{G'}=\frac{\ket{e'_0f_0}+\ket{e'_1f_1}}s,
 \qquad \ket{B'}=\ket{e_2f_0}.
 \label{eq:qutrit-spectral-branches}
\end{equation}
\end{lemma}

\begin{proof}
By Lemma~\ref{lem:bell-spectrum}, a commuting pair at any party would
give \(b\leq1/s\).
Thus every qubit pair is noncommuting, and the qutrit has one
noncommuting two-dimensional sector and one scalar sector. Put the
noncommuting pairs in canonical coordinates. With the product
probabilities of \eqref{eq:variables}, write
\[
 A_j=\frac{(\cos\theta_j+\sin\theta_j)^2}{2},\quad
 b_j=\frac{(\cos\theta_j-\sin\theta_j)^2}{2},\quad
 A_{\mathrm{rest}}\coloneqq\prod_{j=2}^n A_j,\quad
 B_{\mathrm{rest}}\coloneqq\prod_{j=2}^n b_j.
\]
Here \(A_j\geq b_j\), \(A_j+b_j=1\), and
\(A_{\mathrm{rest}}\geq B_{\mathrm{rest}}\).
On \(Q\), Lemma~\ref{lem:bell-spectrum} gives the normalized eigenvalues
\(\pm\sqrt{B_HA_D+A_HB_D}\), and
\(\lambda^2=A_1A_{\mathrm{rest}}+b_1B_{\mathrm{rest}}\leq A_{\mathrm{rest}}\). This is the largest squared
eigenvalue: the principal pair, that of the split \(H=\varnothing\), dominates the pair for any split
\(H\sqcup D=[n]\), since
\[
 A_HA_D+B_HB_D-(B_HA_D+A_HB_D)
 =(A_H-B_H)(A_D-B_D)\geq0.
\]
By \eqref{eq:spectral-sum}, the sum of positive eigenvalue squares on
\(Q\) is \(1\). The remaining positive eigenvalues on
\(Q\) are therefore at most \(\sqrt{1-\lambda^2}\).

To obtain the scalar spectrum, embed \(e_2\) as an eigenvector of
\(\Gamma\coloneqq X\) at \(\theta_1=0\), with the common sign as
eigenvalue, if the two signs agree, and of \(\Gamma\coloneqq Z\) at
\(\theta_1=\pi/2\), with the first sign as eigenvalue, if they are
opposite. Both canonical first observables are then \(\pm\Gamma\),
so by linearity the normalized endpoint Bell operator is
\(\Gamma\otimes\widehat W_{\mathrm{rest}}\) for an operator
\(\widehat W_{\mathrm{rest}}\) on parties \(2,\ldots,n\),
and its compression to the scalar sector is
\(\pm\widehat W_{\mathrm{rest}}\).
Since \(\Gamma\) flips the first party's \(Y\) sign,
\(\widehat W_{\mathrm{rest}}\) is
off-diagonal on each plane spanned by complementary product strings
on parties \(2,\ldots,n\), and its square there is one half the sum
of the two complementary product atoms, by
Lemma~\ref{lem:bell-spectrum} with \(\theta_1\in\{0,\pi/2\}\), so that
\(A_1=b_1=1/2\). Its two eigenvalues are the signed square
roots of that quantity. The same factorization over nontrivial
splits of \(\{2,\ldots,n\}\), with \(A_j>b_j\) for every rest party,
shows that the principal plane is uniquely largest. Its positive
eigenvalue satisfies
\[
 \mu^2=\frac{A_{\mathrm{rest}}+B_{\mathrm{rest}}}{2}\leq\frac12.
\]
The sum of the squared positive eigenvalues on the scalar sector
is \(1/2\). Every other positive eigenvalue in this sector is
at most
\[
 \sqrt{\tfrac12-\mu^2}
 \leq\sqrt{\frac{1-\lambda^2}{2}},
\]
because \(2\mu^2=A_{\mathrm{rest}}+B_{\mathrm{rest}}\geq\lambda^2\).
Now \(b\leq\|\widehat W\|=\max\{\lambda,\mu\}\), and
\(b>1/s\geq\mu\) gives \(\lambda\geq b\). In particular
\(\lambda^2>1/2>1-\lambda^2\), so the eigenvalue \(\lambda\) on \(Q\) is simple.

The principal eigenvectors use the same two locally orthogonal rest
product strings \(u\coloneqq\bigotimes_{j=2}^n\ket{\tau_j}_Y\) and
\(w\coloneqq\bigotimes_{j=2}^n\ket{-\tau_j}_Y\): the rest parties keep
their canonical coordinates in both first-party sectors, and in
both sectors the principal plane is the split with every rest index
in \(D\). For orthonormal first-party vectors
\(\chi_0,\chi_1\) in \(\operatorname{ran}Q_1\), and phases \(\vartheta_G,\vartheta_B\), they have
the form
\[
 \ket{G'}=\frac{\ket{\chi_0u}+e^{\mathrm i\vartheta_G}\ket{\chi_1w}}s,
 \qquad
 \ket{B'}=\ket{e_2}\otimes
          \frac{\ket u+e^{\mathrm i\vartheta_B}\ket w}s.
\]
Local unitaries map the factors of \(u,w\) to \(\ket0,\ket1\),
and \(\chi_0,\chi_1,e_2\) to \(\ket0,\ket1,\ket2\).
A phase at one rest party removes \(\vartheta_B\); phases on the first
party's \(\ket1\) and \(\ket2\) then remove the remaining phase in
\(G'\) and the overall phase of \(B'\), respectively.
This single product unitary sends the two vectors to the standard GHZ
state and \(\ket2\) tensored with a GHZ state on the remaining parties. Composing with the inverse of the steps applied to \(G,B\) in the
proof of Lemma~\ref{lem:branch-alignment} (its case \(i=1\),
\(\ket g=\ket{x+}\)) gives the asserted alignment. Finally, take
\[
 \ket{f_{0,1}}\coloneqq\frac{\ket u\pm e^{\mathrm i\vartheta_B}\ket w}s,
 \qquad
 \ket{e'_{0,1}}\coloneqq
 \frac{\ket{\chi_0}\pm e^{\mathrm i(\vartheta_G-\vartheta_B)}\ket{\chi_1}}s
\]
to obtain \eqref{eq:qutrit-spectral-branches}.
\end{proof}

\begin{proof}[Proof of Proposition~\ref{prop:high-score-stability}]
First take projective measurements. Use the principal eigenvectors
\(G',B'\) and their common local alignment from
Lemma~\ref{lem:qutrit-spectral}, and put
\(P\coloneqq\proj{G'}+\proj{B'}\).
In compatible coordinates write
\[
 G'=\frac{\ket{e_0f_0}+\ket{e_1f_1}}s,
 \qquad B'=\ket{e_2f_0},\qquad s=\sqrt2.
\]
Here \(e_0,e_1\) span the qutrit's qubit sector,
\(e_2\) spans its scalar sector, and \(f_0,f_1\) are orthonormal.
The channel that acts as the identity on the qubit sector and maps
\(e_2\) to \(e_0\), discarding coherence between the sectors, pulls
the GHZ projector back to
\(K_0\coloneqq\proj{G'}+\tfrac12\proj{B'}\), after the common output
unitaries. With \(a=1+1/s\) and \(\widehat W=W_n/\beta_Q\), set
\[
 T_0\coloneqq K_0-a\widehat W+I/s.
\]
Its eigenvalues on \(G',B'\) are
\(a(1-\lambda)\) and \(a(1/s-\mu)\), respectively, and are
nonnegative. Every other positive Bell eigenvalue is at most
\(\sqrt{1-b^2}\), by Lemma~\ref{lem:qutrit-spectral}.
Consequently
\begin{equation}
 T_0\geq d(b)(I-P)\geq0.
 \label{eq:high-score-witness}
\end{equation}

Define
\[
 F_0\coloneqq\tr(\rho K_0),\quad
 \gamma\coloneqq F_0-L(\beta),\quad
 g\coloneqq\Xi_n(\rho)-F_0,\quad
 e\coloneqq\Xi_n(\rho)-L(\beta)=\gamma+g.
\]
Both \(\gamma\) and \(g\) are nonnegative. Put
\(\ell\coloneqq1-\tr(P\rho)\), \(p_G\coloneqq\bra{G'}\rho\ket{G'}\),
\(p_B\coloneqq\bra{B'}\rho\ket{B'}\) and
\(c\coloneqq\bra{G'}\rho\ket{B'}\). Then
\eqref{eq:high-score-witness} and
\(F_0=p_G+p_B/2\) give
\begin{equation}
 \ell\leq\frac{\gamma}{d(b)},\qquad
 p_G-t=\ell+2\gamma,\qquad
 p_B-(1-t)=-2(\ell+\gamma).
 \label{eq:high-score-populations}
\end{equation}
In particular \(p_B\leq1-t\).

To control \(c\), apply Lemma~\ref{lem:branch-coherence} to \(G',B'\) in
these coordinates. Its hypothesis holds by the alignment of
Lemma~\ref{lem:qutrit-spectral}, and its \(K_0\), \(F_0\) and \(g\) are
the present ones. With \(p_B\leq1-t\), it gives
\[
 |c|\leq\sqrt{2g(1+g)}+\sqrt{\ell p_B}
       \leq\sqrt{2g(1+g)}+\sqrt{\ell(1-t)}.
\]
The alignment maps \(t\proj{G'}+(1-t)\proj{B'}\) to \(\omega_t\),
so \(D_t(\rho)\) is at most the trace-norm distance from \(\rho\) to this
operator. Lemma~\ref{lem:two-branch}, together with
\eqref{eq:high-score-populations}, gives
\begin{align}
 D_t(\rho)
 &\leq2(1+\sqrt{1-t})\sqrt{\gamma/d(b)}
      +2\sqrt{2g(1+g)}
      +\left(3/d(b)+4\right)\gamma.
 \label{eq:high-score-budget}
\end{align}
Indeed \(|p_G-t|+|p_B-(1-t)|=3\ell+4\gamma\).
Apply Cauchy--Schwarz to the two square-root terms using
\(\gamma+g=e\leq\varepsilon\) and \(g\leq\varepsilon\).
This proves \eqref{eq:high-score-stability} for projective strategies.

For binary POVMs, maximize the score on the same state and spaces
and select reflection maximizers, as in the extreme-point argument
of Lemma~\ref{lem:projective-replacement}.
Let \(b_*\geq b\), \(t_*\geq t\) be their score parameters.
Inequality~\eqref{eq:high-score-witness} for these reflections gives
\(\Xi_n(\rho)-L(\beta_*)\geq0\).
For this replacement, define \(F_0\) and \(g\) as above and put
\(\gamma\coloneqq F_0-L(\beta_*)\), so that
\(\Xi_n(\rho)-L(\beta_*)=\gamma+g\).
Writing \(h\coloneqq L(\beta_*)-L(\beta)=(t_*-t)/2\), we have
\(e=\gamma+g+h\), with all three terms nonnegative.
Also \(d(b_*)\geq d(b)\), \(1-t_*\leq1-t\), and
\(\|\omega_{t_*}-\omega_t\|_1=4h\).
Apply \eqref{eq:high-score-budget} at the replacement score and
add \(4h\). Its linear terms are bounded by
\((3/d(b)+4)e\); Cauchy--Schwarz bounds the square-root terms
as before. This gives the same estimate, with the cap \(2\)
coming from the trace-norm diameter. At \(b=1\), the actual
excess is zero, so applying the estimate with that exact excess
gives \(D_1(\rho)=0\).
\end{proof}

\begin{proof}[Proof of Proposition~\ref{prop:stability-to-maximum}]
Put \(\eta_0\coloneqq\min\{\eta,1/4\}\),
\(b_{75}\coloneqq1/s+3\kappa/4\), and \(d_{75}\coloneqq d(b_{75})>0\).
Theorem~\ref{thm:original-state-stability} applies on
\(\eta\leq t\leq3/4\) with parameter
\(\eta_0\). On \(3/4\leq t\leq1\), estimate \eqref{eq:high-score-stability}, for
\(\varepsilon\leq1\), is at most
\[
 \left[\sqrt{9/d_{75}+16}+3/d_{75}+4\right]\sqrt\varepsilon.
\]
Thus the explicit choices
\[
 \varepsilon_\eta^+\coloneqq\min\{1,\varepsilon_{\eta_0}\},\qquad
 C_\eta^+\coloneqq\max\left\{C_{\eta_0},
              \sqrt{9/d_{75}+16}+3/d_{75}+4\right\}
\]
cover \(\eta\leq t\leq1\). All constants are independent of \(n\).
\end{proof}

\section{Proofs of the certification estimates}
\label{app:local-certification}

\begin{proof}[Proof of Theorem~\ref{thm:local-certification}]
We first assume that the Bell observables are projective.
The range of the prescribed \(Q_1\) is then the noncommuting
two-dimensional sector of Lemma~\ref{lem:qutrit-spectral}: a commuting restriction to
\(Q_1\), together with its scalar complement, would make the
whole first-party pair commute and force \(b\leq1/s\), by
Lemma~\ref{lem:bell-spectrum}.
Use the lemma's \(Q,E,\widehat W,G',B'\), and put
\[
 P\coloneqq\proj{G'}+\proj{B'},\qquad
 T_{\mathrm{cert}}\coloneqq Q+E/s-\widehat W,\qquad
 g(b)\coloneqq\frac{1-\sqrt{1-b^2}}s.
\]
On \(G',B'\), the eigenvalues of \(T_{\mathrm{cert}}\) are
\(1-\lambda\geq0\) and \(1/s-\mu\geq0\).
On their orthogonal complements in \(Q,E\), respectively, the
lemma bounds them below by \(s g(b)\) and \(g(b)\). Consequently
\begin{equation}
 T_{\mathrm{cert}}\geq g(b)(I-P),\qquad
 \kappa\delta=\tr(\rho T_{\mathrm{cert}})\geq g(b)\ell,\qquad
 \ell\coloneqq1-\tr(P\rho)\leq r(b)\delta.
 \label{eq:local-certification-leakage}
\end{equation}
The middle identity follows from
\(\tr(\rho T_{\mathrm{cert}})=q+(1-q)/s-b\); it also proves \(\delta\geq0\).
This is the spectral-gap step of subspace verification: the
expectation of a positive operator bounds the weight outside the
chosen subspace~\cite[Lemma~A1]{ChenEtAl2025}.

We next recover the populations and coherence within that subspace.
Write \(p_G\coloneqq\bra{G'}\rho\ket{G'}\),
\(p_B\coloneqq\bra{B'}\rho\ket{B'}\), and
\(c\coloneqq\bra{G'}\rho\ket{B'}\). The nonnegative residuals
\(r_Q\coloneqq q-p_G\) and \(r_E\coloneqq1-q-p_B\) sum to \(\ell\), so
\begin{equation}
 |p_G-t|+|p_B-(1-t)|
 =|\delta-r_Q|+\delta+r_E\leq2\delta+\ell.
 \label{eq:local-certification-populations}
\end{equation}
In the compatible coordinates \eqref{eq:qutrit-spectral-branches},
put \(m'\coloneqq\bra{e'_0}\rho_1\ket{e_2}\). Then
\[
 |m'|\leq v,\qquad
 \bra{e'_0}\tr_{2,\ldots,n}(P\rho P)\ket{e_2}=c/s.
\]
By Lemma~\ref{lem:two-branch},
\(\|\rho-P\rho P\|_1\leq2\sqrt\ell\).
An off-diagonal entry of any Hermitian operator \(M\) has modulus
at most \(\|M\|_1/2\), by testing against the norm-one Hermitian
operator supported on that entry and its conjugate. Applying these
facts to the reduced difference gives
\begin{equation}
 |m'-c/s|\leq\sqrt\ell,
 \qquad |c|\leq s v+s\sqrt\ell.
 \label{eq:local-certification-coherence}
\end{equation}
The product unitary of Lemma~\ref{lem:qutrit-spectral} sends
\(G',B'\) to \(G,B\) simultaneously. Lemma~\ref{lem:two-branch} and
\eqref{eq:local-certification-populations} therefore give
\begin{equation}
 \begin{split}
 D_t(\rho)
 &\leq2\sqrt\ell+2\delta+\ell+2|c|\\
 &\leq2(1+s)\sqrt\ell+2\delta+\ell+2s v.
 \end{split}
 \label{eq:local-certification-combination}
\end{equation}
Together with \eqref{eq:local-certification-leakage}, this proves
the asserted expression for projective measurements.

For binary POVMs, maximize the score of the same state over
Hermitian contractions, restricting the first-party operators to
\(\mathcal B(\operatorname{ran}Q_1)\oplus\mathbb C\).
As in Lemma~\ref{lem:projective-replacement}, compactness and
separate linearity allow a maximizer at extreme points in every
coordinate. These are reflections on each qubit and on \(Q_1\),
with scalar signs on its complement. The replacement preserves
\(\rho,Q_1,q,v\). Let its normalized score be \(b_*\geq b\),
put \(t_*\coloneqq(b_*-1/s)/\kappa\geq t\) and
\(\delta_*\coloneqq q-t_*\geq0\).
Then \(\delta=\delta_*+(t_*-t)\). The same aligned branches give
\(\|\omega_{t_*}-\omega_t\|_1=2(t_*-t)\).
Since \(r(b)\) decreases with \(b\),
\[
 \begin{split}
 &2(1+s)\sqrt{r(b_*)\delta_*}
 +(2+r(b_*))\delta_*+2(t_*-t)\\
 &\hspace{3em}\leq
 2(1+s)\sqrt{r(b)\delta}+(2+r(b))\delta.
 \end{split}
\]
The coherence term \(2sv\) is unchanged. The triangle inequality
therefore proves the same estimate for the original measurements;
the trace-norm diameter of the state space supplies the cap \(2\).
\end{proof}

\begin{proof}[Proof of \eqref{eq:fixed-local-certification}]
For the fixed observables of Theorem~\ref{thm:qutrit-attainer},
the spectral calculation in the proof of
Proposition~\ref{prop:fixed-observable-stability} gives eigenvalues
\(1,-1,0\) on \(Q\) and \(1/s,-1/s,0\) on \(E\). The positive
eigenvalues \(1\) on \(Q\) and \(1/s\) on \(E\) are simple, with
eigenvectors \(G,B\), respectively. Thus, with
\(P\coloneqq\proj G+\proj B\) and \(\ell\coloneqq1-\tr(P\rho)\), the operator
\(Q+E/s-\widehat W\) vanishes on \(P\) and is at least
\((I-P)/s\). Hence \(\delta\geq0\) and
\(\ell\leq(s-1)\delta\) throughout \(1/s\leq b\leq1\).
The population and coherence estimates above hold in the fixed
coordinates of Proposition~\ref{prop:coherent-qutrit}, where
\(e'_0=\ket{x+}\) and \(e_2=\ket2\) as in the proof of
Proposition~\ref{prop:fixed-observable-stability}, so \(m'=m\) and
\(|m|\) replaces \(v\). Consequently
\[
 \|\rho-\omega_t\|_1
 \leq2(1+s)\sqrt{(s-1)\delta}+(1+s)\delta+2s|m|,
\]
which is \eqref{eq:fixed-local-certification} because
\((1+s)^2(s-1)=1+s\), together with the diameter bound.
\end{proof}

\end{document}